\documentclass[tbtags,reqno]{amsart}

\edef\savecatcodeat{\the\catcode`@}
\catcode`\@=11

\def\tb@ifSpecChars#1#2{#1}
\def\tb@ifNoSpecChars#1#2{#2}

\def\tableau{%
  \bgroup
  \@ifstar{\let\Tif\tb@ifNoSpecChars\tb@tableauB}
          {\let\Tif\tb@ifSpecChars\tb@tableauB}}

\def\tb@tableauB{
  \@ifnextchar[{\tb@tableauC}{\tb@tableauC[]}}

\def\tb@tableauC[#1]{\hbox\bgroup%
    \let\\=\cr
    \def\bl{\global\let\tbcellF\tb@cellNF}%
    \def\tf{\global\let\tbcellF\tb@cellH}
    \dimen2=\ht\strutbox \advance\dimen2 by\dp\strutbox%
    \ifx\baselinestretch\undefined\relax%
    \else%
       \dimen0=100sp \dimen0=\baselinestretch\dimen0%
       \dimen2=100\dimen2 \divide\dimen2 by\dimen0%
    \fi%
    \let\tpos\tb@vcenter
    \tb@initYoung
    \tb@options#1\eoo
    \let\arrow\tb@arrow%
    \dimen0=\Tscale\dimen2%
    \dimen1=\dimen0 \advance\dimen1 by \tb@fframe%
    \lineskip=0pt\baselineskip=0pt
    \def\tb@nothing{}%
    \def\endcellno{$\rss\egroup\bss\egroup}
    \def\endcell{\endcellno\kern-\dimen0}
    \def\begincell{\vbox to\dimen0\bgroup\vss\hbox to\dimen0\bgroup\hss$}%
    \let\overlay\tb@overlay%
    \let\fl\tb@fl%
    \let\lss\hss\let\rss\hss\let\tss\vss\let\bss\vss
    \def\mkcell##1{
        \let\tbcellF\tb@cellD
        \def\tb@cellarg{##1}
        \ifx\tb@cellarg\tb@nothing\let\tb@cellarg\tb@cellE\fi%
        \begincell\tb@cellarg\endcellno
        \tbcellF}
    \let\savecellF\tbcellF
     \Tif{\catcode`,=4\catcode`|=\active}{}\tb@tableauD}%

\let\tb@savetableauD\tableauD
{
    \catcode`|=\active \catcode`*=\active \catcode`~=\active%
    \catcode`@=\active
\gdef\tableauD#1{%
  \Tif{
    \mathcode`|="8000 \mathcode`*="8000%
    \mathcode`~="8000 \mathcode`@="8000%
    \def@{\bullet}%
    \let|\cr
    \let*\tf
    \let~\sk
  }{}%
  \tpos{\tabskip=0pt\halign{&\mkcell{##}\cr#1\crcr}}%
  \global\let\tbcellF\savecellF
  \egroup
  \egroup}
}
\let\tb@tableauD\tableauD
\let\tableauD\tb@savetableauD
\let\tb@savetableauD\undefined

\def\tb@options#1{\ifx#1\eoo\relax\else\tb@option#1\expandafter\tb@options\fi}

\def\tb@option#1{%
  \if#1t\let\tpos\tb@vtop\fi
  \if#1c\let\tpos\tb@vcenter\fi
  \if#1b\let\tpos\vbox\fi
  \if#1F\tb@initFerrers\fi
  \if#1Y\tb@initYoung\fi
  \if#1s\tb@initSmall\fi
  \if#1m\tb@initMedium\fi
  \if#1l\tb@initLarge\fi
  \if#1p\tb@initPartition\fi
  \if#1a\tb@initArrow\fi
}

\def\tb@vcenter#1{\ifmmode\vcenter{#1}\else$\vcenter{#1}$\fi}

\def\tb@vtop#1{\hbox{\raise\ht\strutbox\hbox{\lower\dimen0\vtop{#1}}}}

\def\tb@initPartition{\def\Tscale{.3}}
\def\tb@initSmall{\def\Tscale{1}}
\def\tb@initMedium{\def\Tscale{2}}
\def\tb@initLarge{\def\Tscale{3}}

\def\tb@initArrow{\dimen2=1.25em}

\def\tb@initYoung{%
  \def\tb@cellE{}
  \let\tb@cellD\tb@cellN
  \def\sk{\global\let\tbcellF\tb@cellNF}}
\def\tb@initFerrers{%
  \def\tb@cellE{\bullet}
  \let\tb@cellD\tb@cellNF
  \def\sk{\bullet}}

\tb@initMedium

\def\tb@sframe#1{%
  \vbox to0pt{
    \vss
    \hbox to0pt{%
      \hss
      \vbox to\dimen1{
        \hrule depth #1 height0pt
        \vss
        \hbox to\dimen1{
          \vrule width #1 height\dimen1
          \hss
          \vrule width #1
          }%
        \vss
        \hrule height #1 depth 0in
        }%
      \kern-\tb@hframe
      }%
    \kern-\tb@hframe}}

\def\tb@hframe{.2pt}\def\tb@fframe{.4pt}\def\tb@bframe{2pt}
\def\tb@cellH{\tb@sframe{\tb@bframe}}       
\def\tb@cellNF{}                            
\def\tb@cellN{\tb@sframe{\tb@fframe}}       
\let\tbcellF\tb@cellN                       

\def\tb@rpad{1pt}
\def\tb@lpad{1pt}
\def\tb@tpad{1.8pt}
\def\tb@bpad{1.8pt}

\def\tb@overlay{\endcell\@ifnextchar[{\tb@overlaya}{\begincell}}
\def\tb@overlaya[#1]{\vbox to\dimen0\bgroup%
  \tb@overlayoptions#1\eoo%
  \tss\hbox to\dimen0\bgroup\lss$}
\def\tb@overlayoptions#1{\ifx#1\eoo\relax\else\tb@overlayoption#1\expandafter\tb@overlayoptions\fi}

\def\tb@overlayoption#1{
  \if#1t\def\tss{\vskip\tb@tpad}\let\bss\vss\fi
  \if#1c\let\tss\vss\let\bss\vss\fi
  \if#1b\def\bss{\vskip\tb@bpad}\let\tss\vss\fi
  \if#1l\def\lss{\hskip\tb@lpad}\let\rss\hss\fi
  \if#1m\let\lss\hss\let\rss\hss\fi
  \if#1r\def\rss{\hskip\tb@rpad}\let\lss\hss\fi
}

\def\tb@fl{\endcell\begincell\vrule depth 0pt width \dimen0 height \dimen0 \endcell\begincell}

\@ifundefined{diagram}{}{
\def\tb@arrowpad{.5}

\newoptcommand{\tb@arrow}{\@ne}[2]{%
  \endcell
   \begingroup%
   \let\dg@getnodesize\tb@getnodesize
   \dg@USERSIZE=#1\relax%
   \ifnum\dg@USERSIZE<\@ne \dg@USERSIZE=\@ne \fi%
   \dg@parse{#2}%
   \dg@label{\tb@draw{#1}{#2}}}

\def\tb@getnodesize#1#2#3#4#5{\dimen3=\tb@arrowpad\dimen2 #4=\dimen3 #5=\dimen3\relax}
\def\tb@getnodesize#1#2#3#4#5{\ifnum#2=0\ifnum#3=0\tb@getnodesizetail{#4}{#5}\else\tb@getnodesizehead{#4}{#5}\fi\else\tb@getnodesizehead{#4}{#5}\fi}
\def\tb@getnodesizetail#1#2{\dimen3=.5\dimen2 #1=\dimen3 #2=\dimen3}
\def\tb@getnodesizehead#1#2{\dimen3=.5\dimen2 #1=\dimen3 #2=\dimen3}

\def\tb@draw#1#2#3#4{%
        \dg@X=0\dg@Y=0\dg@XGRID=1\dg@YGRID=1\unitlength=.001\dimen0%
        \dg@LBLOFF=\dgLABELOFFSET \divide\dg@LBLOFF\unitlength%
        \dg@drawcalc
        \begincell
        \let\lams@arrow\tb@lams@arrow
        \begin{picture}(0,0)\begingroup\dg@draw{#1}{#2}{#3}{#4}\end{picture}%
        \endcell
        \endgroup
        \begincell}
}

\def\tb@lams@arrow#1#2{%
 \lams@firstx\z@\lams@firsty\z@
 \lams@lastx#1\relax\lams@lasty#2\relax
 \lams@center\z@
 \N@false\E@false\H@false\V@false
 \ifdim\lams@lastx>\z@\E@true\fi
 \ifdim\lams@lastx=\z@\V@true\fi
 \ifdim\lams@lasty>\z@\N@true\fi
 \ifdim\lams@lasty=\z@\H@true\fi
 \NESW@false
 \ifN@\ifE@\NESW@true\fi\else\ifE@\else\NESW@true\fi\fi
 \ifH@\else\ifV@\else
  \lams@slope
  \ifnum\lams@tani>\lams@tanii
   \lams@ht\ten@\p@\lams@wd\ten@\p@
   \multiply\lams@wd\lams@tanii\divide\lams@wd\lams@tani
  \else
   \lams@wd\ten@\p@\lams@ht\ten@\p@
   \divide\lams@ht\lams@tanii\multiply\lams@ht\lams@tani
  \fi
 \fi\fi
 \ifH@  \lams@harrow
 \else\ifV@ \lams@varrow
 \else \lams@darrow
 \fi\fi
}

\catcode`\@=\savecatcodeat
\let\savecatcodeat\undefined

\usepackage{mathrsfs}
\usepackage{latexsym}
\usepackage{amsthm,amsopn,amsfonts,amssymb,epsfig,color,pstricks}
\usepackage[active]{srcltx}
\numberwithin{equation}{section}

\DeclareMathOperator{\ch}{ch}

\makeatletter
\renewcommand{\subsubsection}{\@startsection
{subsubsection}
{3}
{0mm}
{\baselineskip}
{-0.5\baselineskip}
{\normalfont\normalsize\bfseries}}
\makeatother

\newtheorem{theorem}{Theorem}
\newtheorem{lemma}[theorem]{Lemma}
\newtheorem{proposition}[theorem]{Proposition}

\newtheorem{conjecture}[theorem]{Conjecture}

\newtheorem{corollary}[theorem]{Corollary}

\theoremstyle{remark}

\newtheorem*{acknow}{Acknowledgments}

\def\itta{I.1}
\def\ittb{I.2}
\def\ittc{I.3}
\def\ittd{I.4}
\def\itte{I.5}
\def\ittf{II.1}
\def\ittg{II.2}
\def\itth{II.3}
\def\itti{II.4}
\def\ittj{II.5}
\def\ittk{III.1}
\def\ittl{III.2}
\def\ittm{III.3}
\def\ittn{III.4}
\def\itto{IV.1}
\def\ittp{IV.2}
\def\ittq{IV.3}
\def\ittr{IV.4}

\def\la{{\lambda}}

\def\cal L{{\mathcal L}}
\def\P{{\mathcal P}}
\def\S{{\mathcal S}}

\def\Z{{\mathbb Z}}

\def\aa{\alpha}

\def\Circle{\pspicture(0.5,0.5)\pscircle(0.25,0.25){.25}\endpspicture}

\newcommand{\cercle}[1]{\ensuremath{\setlength{\unitlength}{1ex}\begin{picture}(2.8,2.8)\put(1.4,0.7){\circle{2.8}\makebox(-5.6,0){#1}}\end{picture}}}
\newcommand{\tcercle}[1]{\ensuremath{\setlength{\unitlength}{1ex}\begin{picture}(2.8,2.8)\put(1.4,1.4){\circle{2.8}\makebox(-5.6,0){#1}}\end{picture}}}

\def\B{{\mathcal B}}
\def\C{{\mathcal C}}
\def\R{{\mathcal R}}

\let\la\lambda
\let\La\Lambda
\let\a\alpha
\def\ad{\a_{k,r}}
\def\Lad{\La_{k,r,N}}

\let\Om\Omega

\let\ta\theta

\let\rw\rightarrow

\newcommand{\LL}{\ensuremath{\langle\!\langle}}
\newcommand{\RR}{\ensuremath{\rangle\!\rangle}}

\newcommand{\coeff}[1]{\ensuremath{\underset{#1}{\mathrm{coeff}}}}

\def\cd{\circledast}
\def\S{\mathcal{S}}
\def\B{\mathcal{B}}

\def\lrw{\leftrightarrow}

\begin{document}
\title[Jack superpolynomials with negative fractional parameter]{Jack superpolynomials with negative fractional parameter:
 clustering  properties and super-Virasoro ideals}
\author{Patrick Desrosiers} \thanks{patrick@inst-mat.utalca.cl}
\address{Instituto de Matem\'atica y F\'{\i}sica, Universidad de
Talca, Casilla 747, Talca, Chile.}
\author{Luc Lapointe}
\thanks{lapointe@inst-mat.utalca.cl}
\address{Instituto de Matem\'atica y F\'{\i}sica, Universidad de
Talca, Casilla 747, Talca, Chile.}
\author{Pierre Mathieu} \thanks{pmathieu@phy.ulaval.ca}
\address{D\'epartement de physique, de g\'enie physique et
d'optique, Universit\'e Laval,  Qu\'ebec, Canada, G1K 7P4.}

\begin{abstract}
The Jack polynomials $P_\la^{(\a)}$ at $\a=-(k+1)/(r-1)$ indexed by certain
$(k,r,N)$-admissible partitions are known to span an ideal $I^{(k,r)}_N$
of the space of symmetric functions in $N$ variables.  The ideal
$I^{(k,r)}_N$ is invariant under the
 action of certain differential operators which include
half the Virasoro algebra. Moreover, the Jack polynomials in
$I^{(k,r)}_N$ admit clusters of size at most $k$: they
vanish when $k+1$ of their variables are identified, and 
they 
do not vanish when only $k$ of them are identified.
We generalize most of these properties to superspace using
orthogonal eigenfunctions of the supersymmetric extension of the
trigonometric Calogero-Moser-Sutherland model known as
Jack superpolynomials.
In particular, we show that the Jack  superpolynomials $P_{\Lambda}^{(\alpha)}$
at $\a=-(k+1)/(r-1)$
indexed by certain
$(k,r,N)$-admissible superpartitions span an ideal ${\mathcal I}^{(k,r)}_N$
of the space of symmetric polynomials in $N$ commuting variables and
$N$ anticommuting variables.  We prove that the ideal
${\mathcal I}^{(k,r)}_N$ is  
stable with respect to the action of the negative-half of the 
super-Virasoro algebra.  {In addition}, we show that the Jack superpolynomials
in 
${\mathcal I}^{(k,r)}_N$ vanish when $k+1$ of their commuting
variables are equal, 
and conjecture that they do not vanish when only $k$  
of them are identified.
This allows us to conclude that the standard Jack polynomials with prescribed symmetry  should satisfy similar clustering properties.  {Finally, we conjecture that the elements of ${\mathcal I}^{(k,2)}_N$ provide a basis for the subspace of symmetric superpolynomials in $N$ 
variables that vanish when $k+1$ commuting variables are set equal to each other.}
\end{abstract}


 \maketitle

\let\la\lambda

\tableofcontents

\section{Introduction}

\subsection{Some remarkable properties of the Jack polynomials}
 
The Jack polynomials $P_{\lambda}^{(\alpha)}$ 
form a  basis of the space ${\bf \Lambda}_N$ of
symmetric polynomials in $N$ variables \cite{Mac}.  When expanded in terms of  
monomial symmetric functions,
\begin{equation}
P^{(\a)}_\la=m_\la+\sum_{\mu<\la}c_{\la\mu}(\alpha)\,m_\mu,\label{tri}
\end{equation}
where the order on partitions is the usual dominance ordering,
the coefficients $c_{\la\mu}(\alpha)$ are ratios of polynomials in $\aa$ with 
{\it positive} integral coefficients. This is a consequence of the positivity
\cite{Knop} of the so-called integral form of the Jack polynomials: 
for some $v_\la (\alpha) \in\mathbb N[\alpha]$ we have that
\begin{equation}\label{JvsP}
J^{(\a)}_\la=v_\la(\aa) P^{(\a)}_\la 
\end{equation}
has monomial expansion coefficients that belong to  $\mathbb N[\alpha]$.
Here is an  example of $P^{(\a)}_\la$ illustrating this point:
\begin{equation}
P_{(3)}^{(\aa)}=m_{(3)}+
\frac{3}{\aa(2\aa+1)}m_{(2,1)}+\frac{6}{(\aa+1)(2\aa+1)}m_{(1,1,1)}.
\end{equation}
The polynomials 
$P_\la^{(\aa)}$ are thus clearly well defined for any real value of 
$\aa> 0$ but they may have poles for negative real values of $\aa$ 
(it is clearly the case for $\aa=-1$ or $-1/2$ in the above example).

It was shown in 
\cite{FJMM} that for particular  classes of partitions, the polynomials 
$P_\la^{(\aa)}$ at certain negative rational values of $\alpha$
not only are regular (i.e., have no poles)  but also have remarkable properties.
Let $k \geq 1$ and $r \geq 2$ 
be integers such that $k+1$ and $r-1$ are coprime.
For partitions $(\lambda_1,\dots,\lambda_N)$ (admitting entries equal to zero), 
we say that 
$\lambda$ is $(k,r,N)$-admissible if 
\begin{equation} \label{eqadmi}
 \la_i- \la_{i+k}\geq r \quad {\text{for all}} \quad 1 \leq i \leq N-k \, .
\end{equation}
The Jack polynomials $P_{\lambda}^{(\ad)}$, where
$\ad=-(k+1)/(r-1)$, do not have poles when $\lambda$ is 
$(k,r,N)$-admissible. 
Somewhat unexpectedly, the Jack polynomials at $\ad$ indexed by 
$(k,r,N)$-admissible partitions vanish whenever $k+1$ of their variables are equal.
\footnote{The best way of seeing the necessity of the condition that 
 $k+1$ and $r-1$ be coprime
 is  to explore the analogous vanishing properties stemming from Macdonald polynomials \cite{Mac}. This is analyzed in \cite{FJMMqt}, where it is found that the parameters $q$ and $t$  of the Macdonald polynomials must be specialized as follows: $q=\omega_1 u^{-(k+1)/m}$ and $t=u^{(r-1)/m}$. In these expressions,  $u$ is an indeterminate, $\omega_1^{(r-1)/m}$ is a primitive $m^{th}$ root of unity and $m$ is the greatest common divisor of $k+1$ and $r-1$ (which are thus not required to be relatively prime in this context). However, the relationship $q=t^\a$ which  (together with $t\rw 1$) yields the Jack limit forces $\omega_1$ to be equal to 1, and thus $m=1$. 
}
Moreover, the approach to zero of these vanishing Jack polynomials can be made precise:
$P^{(\ad)}_\la$  does not vanish when $x_1=\cdots =x_k=x$ and vanishes
with exponent at least $r$ when  $x_{k+1}\rw x$ \cite{BH}, that is,  
as $(x-x_{k+1})^{s}$ with $s\geq r$.  This is a special case of what is often called the {\it clustering property.}
We stress that in most cases of small degree, the exponent $s$ is exactly equal to $r$.\footnote{
There are very few exceptions to this rule. Specifically, with the restrictions $n\leq 12\,, N\leq 8$ and $k,r\leq6$, there is a total of 3619 admissible partitions and only 18 exceptions:
\begin{align*}
&(2,3,3)\text{-adm}:\quad (4,2,0)\;(6,3,0)\;(8,4,0)\; (5,3,1)\; (6,4,2)\;(7,4,1)\\
&(4,3,5)\text{-adm}:\quad (4,4,2,0,0)\;(5,3,3,0,0)\;(4,2,2,2,0)\; (4,3,2,1,0)\\
&(3,4,4)\text{-adm}:\quad (7,4,0,0)\;(5,3,1,0)\;(5,4,2,0)\\
&(4,4,5)\text{-adm}:\quad (5,3,3,1,0)\\
&(2,5,3)\text{-adm}:\quad (6,3,0)\;(8,4,0)\;(7,4,1)\\
&(4,5,5)\text{-adm}:\quad (6,4,1,1,0)\end{align*}
 We observe that $N=k+1$ in all these exceptional cases, so that the exceptions can be readily ruled out by imposing $N>k+1$.  However, this inequality does not provide a fine delimitation of the exceptional cases since among the 3619  admissible partitions, only  1324 of them are such that $N>k+1$. We observe also that the difference  $\la_1-\la_{k+1}$ is 
not only $\geq r$ in all the exceptional cases, but is actually 
$\geq r+1$.
Moreover, in all these cases the order of the zero  is $r+1$.}   For instance, with $r=2$ and $k=1$, we have
\begin{equation} \label{eq420}
P^{(-2)}_{(4,2,0)}(x_1,x_2,x_3)= (x_1-x_2)^2  (x_1-x_3)^2  (x_2-x_3)^2   
\end{equation}
which thus vanishes with exponent 2 when any two variables are identified. 
As another example, consider
$\a=-3/2$ ($r=3,\, k=2)$, with $x_1=x_2=x$:
\begin{equation}\label{eq4310}
P^{(-3/2)}_{(4,3,1,0)}(x,x,x_3,x_4)=2x\, (x_3+x_4)(x_4-x)^3(x_3-x)^3.\end{equation}
illustrating again neatly the cluster property.
This property has recently been proved in \cite{BF} for a special class of admissible partitions.\footnote{This class corresponds to the
minimal-degree staircase-type partitions with steps of width $k$ and relative height $r$. To be more specific,
these are partitions of the form $(\cdots,(3r)^k,(2r)^k,r^k,0^k)$, where $p^k$ means that the part $p$ is repeated $k$ times. As observed in \cite{BH}, the 
corresponding Jack polynomials $P^{(\ad)}_\la$ are translational invariant, i.e., invariant under the transformation $x_i\rw x_i+a$. The Jack polynomial
$P^{(-2)}_{(4,2,0)}$ given in \eqref{eq420} 
illustrates this property. 
Being translation invariant, they are annihilated by the operator $\sum_{i=1}^N\partial_{x_i}$. These special Jack polynomials (and their Macdonald extensions) have been further studied in \cite{JL}.
}

 Since the  Jack polynomials {$P_{\lambda}^{(\a)}$ for $(k,r,N)$-admissible partitions do not have poles at $\a=\ad$}, we can define the space
\begin{equation}
 I_N^{(k,r)} = {\rm span}_{\mathbb C} \bigl\{ 
P^{(\ad)}_{\lambda} \, \big| \, 
\lambda~{\rm is}~(k,r,N){\text-}{\rm admissible}
\bigl\}
\end{equation}
which turns out to be an 
ideal of the ring of symmetric polynomials in $N$ variables over 
$\mathbb C$.  It is also stable under the action of certain differential operators that realize half of the Virasoro algebra.  Now, consider the subspace of symmetric polynomials 
in $N$ variables that vanish   whenever $k+1$ variables are equal:
\begin{equation}
F_N^{(k)} = \{ f(x_1,\dots,x_N) \in {\bf \Lambda}_N \, | \,  f(x_1,\dots,x_N)=0 {\rm ~if~}
x_1=\cdots=x_{k+1} \} \, .
\end{equation}
From the clustering property of the Jack polynomials, 
we get $I_N^{(k,r)} \subseteq F_N^{(k)}$ for all $r$.  But 
more surprisingly, it can be shown that $I_N^{(k,2)}=F_N^{(k)}$,
which provides a connection between Jack polynomials and
the representation theory  of $\widehat{sl}(2)_k$ \cite{FS,AKS}.

These remarkable features of the Jack polynomials have attracted much interest in physics. In particular,
the clustering property of admissible $P^{(\ad)}_\la$ makes them useful  trial wavefunctions for the fractional Hall effect: for $r=2$, these are related to the Read-Reyazi states \cite{RR}. 
In this context, the restriction on sequences of $k+1$ contiguous quasi-particle modes (the parts of the partition)  can be interpreted as  a sort of generalization of the Pauli exclusion principle \cite{BH}.
In conformal field theory, the $P^{(\ad)}_\la$, for admissible $\la$, correspond to the polynomial part of the correlators of $N$ fundamental parafermionic fields, fields that generate
a generalized $\mathcal Z_{k}^{(r/2)}$  parafermionic algebra \cite{ZF,JM} underlying the $W(k+1,k+r)$ minimal models \cite{BGS,ES} (a connection already alluded to in \cite{FJMM}).

\subsection{Extension to superspace}
The main objective of the present work is to extend the results of \cite{FJMM} to superspace. The Jack polynomials have a superspace extension, $P^{(\a)}_\La$, called the Jack superpolynomials or Jack polynomials in superspace \cite{DLM3,DLM7}, which first appeared in the study of the supersymmetric generalization of the trigonometric Calogero-Moser-Sutherland model.  
The Hamiltonian operator for this model was obtained in \cite{SS} (see also \cite{BTW}) by following the method of Freedman and Mende \cite{FM}:
\begin{equation}\label{susyH}
\mathscr{H}=\frac{1}{2}\sum_{i=1}^N\left(x_i \partial_{x_i}\right)^2-\frac{1}{\alpha^2}\sum_{1\leq i<j\leq N}\frac{x_i x_j}{(x_i-x_j)^2}\left(1-\alpha\kappa_{ij}\right),
\end{equation}
where
\begin{equation}\label{eqkappa}
\kappa_{ij}=1-(\theta_i-\theta_j)\left(\partial_{\theta_i}-\partial_{\theta_{j}}\right).
\end{equation}
The variables $x_j=e^{\mathrm{i}\phi_j}$ describe the positions of the $N$ bosonic particles on the unit circle while the $\theta_i$'s stand for their fermionic partners (which can also be interpreted as internal degrees of freedom of the $x_j$). In more mathematical terms,  $\mathscr{H}$ is an operator that acts on differentiable functions  depending upon a set $x=\{x_1,\ldots,x_N\}$ of commuting variables and a set $\theta=\{\theta_1,\ldots,\theta_N\}$ of anti-commuting variables, with the additional assumption that $\theta_ix_j=x_j\theta_i$ for all $i,j$.

 The systematic search for the eigenfunctions of $\mathscr{H}$ started in \cite{DLM1}.   The task was simplified by the following basic observation: for any polynomial $f(\theta_1,\ldots,\theta_N)$, the operator in \eqref{eqkappa} satisfies
\begin{equation}\kappa_{ij}f(\ldots,\theta_i,\ldots,\theta_j,\ldots)=f(\ldots,\theta_j,\ldots,\theta_i,\ldots),\qquad\forall\,i,j. 
\end{equation} This means that the operators $\kappa_{ij}$ provide an action of the symmetric group $S_N$ on polynomials in $\theta$.  It was shown that the Hamiltonian \eqref{susyH} has eigenfunctions of the form 
\begin{equation} \label{eigen}
\Psi(x;\theta)=\prod_{1\leq i<j\leq N}|x_i-x_j|^{1/\alpha}
f_\La(x;\theta)\end{equation}   
where $f_\Lambda$ is a homogeneous polynomial in $x$ and $\theta$ satisfying 
\begin{equation}\label{symsuperpol}\kappa_{ij}K_{ij}\,f_\Lambda(x;\theta) = f_\Lambda(x;\theta),\qquad \forall\, i,j.\end{equation}  
In the last equation, $K_{ij}$ is an operator   that  interchanges the variables $x_i$ and $x_j$.  Any polynomial such as  $f_\Lambda(x;\theta)$ in \eqref{symsuperpol} is called a symmetric superpolynomial.

The Jack superpolynomials $P^{(\a)}_\La$, which are indexed by superpartitions, are symmetric superpolynomials providing orthogonal eigenfunctions of
the Hamiltonian $\mathscr{H}$ (of the form described in \eqref{eigen}).
In short, a superpartition $\Lambda$
of fermionic degree $m$ is
a pair of partitions $\La=(\La^\cd, \La^*)$ 
such that the skew-diagram 
$\La^\cd/\La^*$ is both a vertical and  a horizontal $m$-strip (see Section~\ref{sect22} for the relevant definitions). 
 Note that the number of entries (possibly including zeroes) of
$\Lambda^\cd$ and $\Lambda^*$ will always be equal to $N$.
A convenient representation of $\La$  is obtained by circling the $m$
entries of $\La^*$ such that 
$\La^\cd_i-\La^*_i=1$. For instance, with $N=7$,  
$$\La^\cd=(6,4,4,3,2,1,0),\quad\text{and}\quad\La^*=(5,4,3,3,1,0,0)\quad\Longrightarrow\quad \La=(\cercle{5},4,\cercle{3},3,\cercle{1},\cercle{0},0).$$

Let again $k\geq 1$ and $r \geq 2$ be integers such that $k+1$ and $r-1$
are coprime.  We say that the superpartition $\Lambda$  is $(k,r,N)$-admissible if
\begin{equation} \label{eqadmisup}
 \Lambda_i^{\cd}- \La_{i+k}^*\geq r \quad {\text{for all}} \quad 1 \leq i \leq N-k \, .
\end{equation}
For example, $(\cercle{7},7,5,\cercle{4},\cercle{2},\cercle{1},0)$ is $(2,3,7)$-admissible whereas $(8,\cercle{4},3,\cercle{1},\cercle{0})$ is $(1,2,5)$-admissible. Note that when $m=0$, we have $\La^\cd=\La^*$.
Hence the superpartition $\Lambda=(\La^*,\La^*)$ 
can be thought as
an ordinary partition,  and the conditions in \eqref{eqadmisup} reduce to those in \eqref{eqadmi}.

Let $\Lambda$ be a $(k,r,N)$-admissible superpartition of fermionic degree $m$.
We show that the
Jack superpolynomial $P_{\Lambda}^{(\ad)}$, where
$\ad=-(k+1)/(r-1)$, does not have poles.
   We prove that
$P_{\Lambda}^{(\ad)}$ vanishes whenever $k+1$ of its commuting 
variables are equal.\footnote{{Given that the Jack superpolynomials can be constructed out of the non-symmetric Jack polynomials (see eq.\eqref{jackinnonsym}) and that the latter can be recovered from the non-symmetric Macdonald polynomials, the last two properties can also be deduced from the results of \cite{Kas}.}} 
 Furthermore, 
we conjecture that if  $N\geq k+m+1$, and $m>0$ 
then
\begin{equation} \label{newcluster}
 P_{\Lambda}^{(\ad)}(x_1,\dots,x_{N-k-1},x',\overbrace{x,\ldots,x}^{k \text{ times}}) 
\Big |_{\theta_1 \cdots \theta_m} \quad \text{vanishes as}\quad (x-x')^r
\quad \text{when} \quad x\to x'
\end{equation}
where the notation 
$ |_{\theta_1 \cdots \theta_m}$ refers to the
coefficient of $\theta_1 \cdots \theta_m$.
   In other words, we prove that the polynomial in \eqref{newcluster} vanishes when $x=x'$, and we conjecture that the multiplicity of the factor $(x-x')$ is {exactly} equal to $r$ whenever $N \geq k+m+1$ and $m>0$. 
This is the clustering property of the Jack superpolynomials.
Note that the condition $N \geq m+k+1$ ensures that the sets
$\{1,...,m\}$ and $\{N-k,\dots,N\}$ do not intersect. 
Stronger  conjectures concerning clustering properties 
for the Jack polynomials with prescribed symmetry that will be discussed
later on in the introduction 
(and in more details in Section~\ref{les0}) imply that if $m>0$ then
\begin{equation} \label{newcluster2}
P_{\Lambda}^{(\ad)}(x_1,\dots,x_{N-k-1},x',\overbrace{x,\ldots,x}^{k \text{ times}}) 
\quad \text{vanishes as}\quad  (x-x')^{r-1}  \quad \text{when} \quad x\to x'
\, ,
\end{equation}
where we stress that here the full superpolynomial is concerned while
 in \eqref{newcluster} only the
coefficient of $\theta_1 \cdots \theta_m$ was considered.

The regularity of the Jack superpolynomials we just mentioned
allows us
to define the space
\begin{equation} \label{defidealI}
 {\mathcal I}_N^{(k,r)} = {\rm span}_{\mathbb C} \bigl\{ 
P^{(\ad)}_{\Lambda} \, \big| \, 
\Lambda~{\rm is}~(k,r,N){\text-}{\rm admissible}
\bigl\} \, .
\end{equation}
We show that  ${\mathcal I}_N^{(k,r)}$ 
is an ideal of the ring of symmetric superpolynomials in 
$N$ variables\footnote{ {When we refer to  a superpolynomial in $N$ variables, it is understood that each variable is considered as a pair of commuting and anticommuting variables.}}  over 
$\mathbb C$. That is, ${\mathcal I}_N^{(k,r)}$ is an ideal of the ring
$$
{\bf \Lambda}_N^{\theta} = 
\mathbb C[x_1,\dots,x_N,\theta_1,\dots,\theta_N]^{S_N} \, ,
$$
where $S_N$ acts diagonally on the two sets of variables.
Moreover, we prove that ${\mathcal I}_N^{(k,r)}$ is  stable 
with respect to the action of the negative-half of the super-Virasoro algebra.
In order to demonstrate these central results, we first need 
to establish 
very intricate properties of the Jack superpolynomials that can be seen
as superspace extensions of the technical results of \cite{FJMM}.
We then need to
obtain the explicit action of certain generators
of the negative-half of the super-Virasoro algebra on the Jack superpolynomials
(reminiscent of Pieri rules).
We stress that the very fact that these results can be established is a confirmation 
of the natural character and 
the richness of these superpolynomials.

Now consider the subspace of symmetric superpolynomials 
in $N$ variables that vanish   whenever $k+1$ commuting 
variables are equal:
\begin{equation}
{\mathcal F}_N^{(k)} = \bigl\{ f(x,\theta) 
\in {\bf \Lambda}_N^{\theta} \, 
\big| \,  f(x,\theta)=0 {\rm ~if~}
x_1=\cdots=x_{k+1} \bigr\} \, .
\end{equation}
{From} our earlier discussion about the
clustering property of the Jack superpolynomials, it follows
that ${\mathcal I}_N^{(k,r)} 
\subseteq {\mathcal F}_N^{(k)}$ 
for all $r$.  We conjecture that, as in the non-supersymmetric case, 
${\mathcal I}_N^{(k,2)}={\mathcal F}_N^{(k)}$.
In other words, the Jack superpolynomials in $N$ variables whose superpartitions are $(k,2,N)$-admissible would furnish a linear basis of the ring
$\mathcal F_N^{(k)}$. This conjecture 
clearly points toward superconformal field theoretical applications of the Jack superpolynomials, for instance for the representation theory of the superspace version of $\widehat{sl}(2)_k$ model or the related $\mathcal Z_k$ parafermionic theory \cite{PLL}.

\subsection{Cochain interpretation of the ideals}
Among all the elements of the super-Virasoro algebra, the generators $G_{1/2}$ and $G_{-1/2}$ are particularly interesting.  They can be written as a linear combination  of the operators 
\begin{equation} q=\sum_i\theta_i\partial_{x_i}\quad \text{and}\quad \tilde q=\sum_i\theta_i x_i\partial_{x_i}\end{equation}
together with their respective adjoints $q^\perp$ and $\tilde q^\perp$, whose precise definition is not relevant for the moment (they correspond respectively
to the operators $q^\perp$ and $Q^\perp$ defined in \eqref{homo}).  
The important point here is that both $q$ and $\tilde q$ can be interpreted as exterior derivatives. 

It is indeed well known that the Hamiltonian operators in supersymmetric quantum mechanics are equivalent to the Laplace-Beltrami operators in classical differential geometry (see for example \cite{FGR}).  From this point of view, if the variables $x_j$ are real, a superpolynomial $f_{i_1,\ldots,i_m}(x)\theta_{i_1}\cdots \theta_{i_m}$  and the left-multiplication by $\theta_j$ are respectively  interpreted as the $m$-form field $f_{i_1,\ldots,i_m}(x)dx_{i_1}\wedge\cdots \wedge dx_{i_m}$ and the exterior product $dx_j\wedge\cdot$. The case where the $x_j$'s belong to the unit circle $\mathbb{T}$ in $\mathbb{C}$ is similar except that the $\theta_i$'s behave as the 1-forms $(\mathrm{i}x_j)^{-1}dx_j$.   Thus, the operator $q$ (resp., $\tilde q$) is equivalent to the exterior derivative 
$d:\, \Omega^m(M)\to \Omega^{m+1}(M)$ when  $M$ is equal to $\mathbb{R}^N$ (resp., $\mathbb{T}^N$).  

Now let us decompose the ring of symmetric superpolynomials as ${\bf \Lambda}_N^{\theta}=\oplus_{m=0}^N  {\bf \Lambda}_{N,m}^{\theta}$, where ${\bf \Lambda}_{N,m}^{\theta}$ denotes the sets of homogeneous elements of ${\bf \Lambda}_N^{\theta}$  whose degree in $\theta$ is equal to $m$, so each element   ${\bf \Lambda}_{N,m}^{\theta}$ can be thought as a {\it symmetric $m$-form}.  Note that ${\bf \Lambda}_{N,0}$ is equal to the usual ring ${\bf \Lambda}_N$ of symmetric polynomials.  All these observations can be summarized by the following  exact cochain complex:    
\begin{equation}
0 \xrightarrow{\quad \quad}
 { \mathbb C(\alpha)}\xrightarrow{\quad i\quad} {\bf \Lambda}_{N} \xrightarrow{\quad d \quad}  {\bf \Lambda}_{N,1}^{\theta}\xrightarrow{\quad d \quad } \cdots \xrightarrow{\quad d \quad } {\bf \Lambda}_{N,N}^{\theta}\xrightarrow{\quad d \quad}0,    
\end{equation}   
where $d$ is either equal to $q$ or $\tilde{q}$ {and $i$ refers to the inclusion}.  In both cases, the
exactness follows easily from the anticommutation relations between $d$ and
$d^\perp$.

Finally, given that the ideal $\mathcal{ {I}}_{N}^{(k,r)}$ is also naturally graded with respect to the degree of its elements in $\theta$,  we can write  
$\mathcal{ {I}}_{N}^{(k,r)}=\sum_{m= 0}^N {I}_{N,m}^{(k,r)}$, 
 where 
\begin{equation}
 {I}_{N,m}^{(k,r)} = {\rm Span}_{\mathbb C} \bigl\{ 
P^{(\ad)}_{\Lambda} \, \big| \, 
\Lambda~{\rm is}~(k,r,N){\text-}{\rm admissible}\text{ and contains exactly $m$ circles} \bigl\}.
\end{equation}
We stress that ${I}_{N,0}^{(k,r)}$  is nothing but the ideal ${I}_{N}^{(k,r)}$ introduced in \cite{FJMM}.
Then, as consequence of the stability of the ideal ${\mathcal I}_N^{(k,r)}$ under the action of the super-Virasoro generators $G_{1/2}$ and $G_{-1/2}$, we obtain 
the following exact cochain {complex}:   
\begin{equation}
0 \xrightarrow{\quad \quad}
 { \mathbb C} \xrightarrow{\quad i\quad} {I}_{N}^{(k,r)} \xrightarrow{\quad d\quad}  {I}_{N,1}^{(k,r)} \xrightarrow{\quad d \quad} \cdots \xrightarrow{\quad d\quad } {I}_{N,N}^{(k,r)} \xrightarrow{\quad d \quad}0 \, ,
\end{equation}
where $d$ is again either equal to $q$ or $\tilde{q}$.

\subsection{Consequences for the Jack polynomials with prescribed symmetry}
Let $\La$ be a superpartition with $m$ circles.  Take all the parts of $\La$ that are circled and order them in decreasing order. This gives a partition  $\La^a=(\La_1,\ldots,\Lambda_m)$ with strictly decreasing parts.  
Now take the parts of $\La$ that are not circled and form a partition $\La^s=(\La_{m+1},\ldots,\La_{N})$.     This allows us to write the Jack superpolynomial labeled by $\La$ as follows: 
\begin{equation} \label{prescribed}
P^{(\a)}_\La(x,\ta)= \ta_1\cdots \ta_m \S^{(\alpha)}_{\La^a,\La^s}(x) +\text{permutations},
\end{equation}
where $\S^{(\alpha)}_{\La^a,\La^s}(x)$ is known as a Jack polynomial with prescribed (or mixed) symmetry \cite{BDF,Baker}, since it is antisymmetric in the variables $x_1,\cdots x_m$ and symmetric in the remaining ones.  

  The general {conjectured} clustering property given in \eqref{newcluster} 
readily implies that the Jack polynomials with prescribed symmetry satisfy a similar property.  To be more precise, let $\La$ be a $(k,r,N)$-admissible superpartition with $m$ circles and let $\La^a,\La^s$ be its associated pair of  partitions as described above.  Then  for all $N\geq k+m+1$    \begin{equation}
  \label{clusterprescribed} \S^{(\ad)}_{\La^a,\La^s}
(x_1,\dots,x_{N-k-1},x',\overbrace{x,\ldots,x}^{k \text{ times}}) 
\quad \text{vanishes as}\quad (x-x')^{r} 
\,  \quad \text{when} \quad x\to x'
.
\end{equation}

Now consider the Jack polynomial with prescribed symmetry
{divided by the Vandermonde determinant $\Delta(x_1,\cdots,x_m)$, that is:}
\begin{equation}
\P_{\La,m}^{(\ad)}=\S^{(\ad)}_{\La^a,\La^s}/\Delta(x_1,\cdots,x_m) \, .
\end{equation}
{Because $\S^{(\ad)}_{\La^a,\La^s}$ is antisymmetric in the first $m$ variables,}
$\P_{\La,m}^{(\ad)}$ is  symmetric
in the variables $x_1,\dots,x_m$ and in the variables $x_{m+1},\dots,x_N$.
When the fermionic degree $m$ is smaller than $r$, we actually obtain 
a stronger clustering property.   We show that if
$r>m$ and $\Lambda$ is $(k,r,N)$-admissible then
$\P_{\La,m}^{(\ad)}$ vanishes whenever any $k+1$ of the variables 
$x_1,\dots,x_N$ are equal.
Here again the approach to zero of the vanishing Jack polynomial
with prescribed symmetry can 
be made precise.
Let $x_{i_1}=\cdots=x_{i_k}=x$ and let $x'$ be a variable that does 
not belong to $\{x_{i_1},\dots,x_{i_k} \}$. Let also 
$a$ be the number of elements in 
$\{x_{i_1},\dots,x_{i_k},x' \} \cap \{x_1,\cdots, x_m\}$.  
We conjecture that
\begin{equation}
 (x-x')^{r-a} \quad \text{divides} \quad
\P_{\La,m}^{(\ad)}=\frac{\S^{(\ad)}_{\La^a,\La^s}}{\Delta(x_1,\cdots,x_m)} \, . \end{equation} 
Note that the conjecture still seems to be valid if $r\leq m$, {although then}
the clustering property can be lost since $\P_{\La,m}^{(\ad)}$ 
does not necessarily vanish when $k+1$ variables are equal ($a$ can be equal to $r$).  Note also that for almost all admissible superpartitions, the multiplicity of the factor $(x-x')$ is exactly equal to {$r-a$}, but as for the non-superspace case, there is no general rule for determining which superpartitions lead to multiplicity strictly greater than {$r-a$}.   However, it seems possible to predict the exact multiplicity of the factor $(x-x')$  by further restricting the set of admissible superpartitions.    As before, let $\La$ be a $(k,r,N)$-admissible superpartition with $k+1$ and $r-1$ being coprime.  Assume moreover that $r>m>0$  and $N\geq k+m+1$.  Then we conjecture that the multiplicity of the factor $(x-x')$ in $\P_{\La,m}^{(\ad)}$ is exactly equal to $r-a$.\footnote{Again this bound on $N$ does not provide a fine demarcation of the exceptional cases. As indicated after Conjecture \ref{conjecrma},  with $ n\leq10,\,N\leq 8,\,k,r\leq 6$ and $m>0$, there are 17914 admissible superpartitions and 489 exceptions. However,  only 2189 admissible superpartitions do satisfy the bounds $N\geq k+m+1 $ and $ r>m$.}

\subsection{Relation with conformal field theory}

{Although direct applications of Jack superpolynomials have not been nailed down yet, $(k,2,N)$-admissible superpartitions have appeared in disguised form
in the description of the basis of states in the superconformal minimal models $\mathcal{SM}(2,4k+4)$ \cite{Mel,FJM}. (Analogously,  $(k,2,N)$-admissible partitions describe the states of the minimal models $\mathcal{M}(2,2k+3)$ \cite{FNO}.) 

Recall that the states in superconformal highest-weight modules are generated by the super-Virasoro modes $L_n$ and $G_r$, with $n,r<0$, where $r\in\Z +\frac12$ in the Neveu-Schwarz sector, while it  is integer in the Ramond sector. The basis of states in the Verma module (free of singular vectors) is
\begin{equation}L_{-n_1}\cdots L_{-n_p}G_{-r_1}\cdots G_{-r_m}|\text{hws}\rangle,\qquad n_i\geq n_{i+1}\geq 1, \quad r_i>r_{i+1}>0,
\end{equation}
for all values of $p$ and $m$.
Such operator strings can be mixed and reordered in decreasing values of the mode indices $n_i,\,r_j$. The resulting sequence of indices is related to a superpartition. In the Ramond sector, we take the convention that if a $G$ mode and some $L$ modes have equal indices, the $G$ is placed at the left. Circling the entries $r_i$, the result  is a superpartition $\La$ where the $m$ parts  
for which $\La_i^\cd-\La_i^*=1$ have been circled (and, by construction, with no vanishing part).
In the Neveu-Schwarz sector, the $r_i$ entries are first reduced by $\frac12$ and then circled. This again leads to a  superpartition, this one allowing a 0 circled-entry.

Now, the basis of states for the irreducible modules in the   $\mathcal{SM}(2,4k+4)$ model (up to a boundary condition on the maximal number of parts $\leq 1$ that characterizes the highest-weight state) is precisely given by admissible superpartitions with $r=2$ (cf. Section 4.2 in \cite{Mel} for the Neveu-Schwarz sector and the appendix A of \cite{FJM} for a sketch  of the proof that applies to both sectors).}

 \subsection{Organization of the article}
 Basic definitions and relevant properties of the  Jack superpolynomials are reviewed in Section \ref{SsJ}.  Jack superpolynomials are shown to be eigenfunctions of a pair of Sekiguchi-type operators in Section~\ref{Sseki}.
The admissibility conditions for superpartitions are introduced in  
Section~\ref{Sregu}, along with 
the proofs of the
regularity of the $P_\La^{(\ad)}$'s when 
$\La$ is admissible or almost admissible (to be defined later on).
We introduce certain super Lie algebras (including the negative half of 
the super-Virasoro algebra) in 
Section \ref{Lie}.  In Section~\ref{ideal}
we obtain the explicit 
action of some elements of the super Lie algebras
on the Jack superpolynomials,
and introduce the
differential ideals $\mathcal I_N^{(k,r)}$. 
The vanishing of the admissible
Jack polynomials when $k+1$ commuting variables are identified is
shown in Section \ref{les0}, along with a conjecture on the
clustering properties of Jack superpolynomials. Finally, the appendix
contains a few technical proofs which we felt were not suited for the
main body of the article.

We should stress that our proofs are in general modeled on those of \cite{FJMM}.
Nevertheless, as was to be expected, 
the presence of anticommuting variables makes  
most of the proofs much more involved.

\section{Jack superpolynomials:  definitions and basic properties}
\label{SsJ}
\subsection{Polynomials in superspace}
Polynomials in superspace (or  superpolynomials) are 
 polynomials in the usual commuting $N$ variables $x_1,\cdots ,x_N$  and the $N$ anticommuting variables $\ta_1,\cdots,\ta_N$. They are said to be
  symmetric if they are invariant with respect to the interchange of $(x_i,\ta_i)\lrw (x_j,\ta_j)$ for any $i,j$ 
\cite{DLM1}.

The symmetry requirement can be phrased in terms of exchange operations.
For any $\sigma \in S_N$, we define
\begin{equation}\label{defKk}
\mathcal{K}_{\sigma}=\kappa_{\sigma}K_{\sigma},
\qquad\text{where}\quad\begin{cases}
&K_{\sigma}\,:\, (x_1,\dots,x_N) \mapsto (x_{\sigma(1)}, \dots,
x_{\sigma(N)})
\\
&\kappa_{\sigma}\,\;:\, (\theta_1,\dots,\theta_N) \mapsto (\theta_{\sigma(1)}, \dots,
\theta_{\sigma(N)})
.\end{cases}\end{equation}
Then a polynomial in superspace $P(x;\theta)$, with $x=(x_1,\ldots,x_N)$ and $\theta=(\theta_1,\ldots,\theta_N)$, is symmetric when
\begin{equation}
\mathcal{K}_{\sigma}P(x;\theta)=P(x;\theta) \qquad {\rm for~all~} \sigma \in S_N 
\, .
\end{equation}

\subsection{Superpartitions: diagrammatic representation and the dominance ordering} \label{sect22}
Let us first recall some definitions 
related to partitions \cite{Mac}.
A partition $\lambda=(\lambda_1,\lambda_2,\dots)$ of degree $|\lambda|=d$
is a vector of non-negative integers such that
$\lambda_i \geq \lambda_{i+1}$ for $i=1,2,\dots$ and such that
$\sum_i \lambda_i=d$.  The length $\ell(\lambda)$
of $\lambda$ is the number of non-zero entries of $\lambda$.
Each partition $\lambda$ has an associated Ferrers diagram
with $\lambda_i$ lattice squares in the $i^{th}$ row,
from the top to bottom. Any lattice square in the Ferrers diagram
is called a cell, where the cell $(i,j)$ is in the $i$th row and $j$th
column of the diagram.  Given  a partition $\lambda$, its
conjugate $\lambda'$ is the diagram
obtained by reflecting  $\lambda$ about the main diagonal.
Given a cell $s=(i,j)$ in $\lambda$, we let
\begin{equation} \label{arml}
a_{\lambda}(s)=\lambda_i-j\, , \qquad   a'_{\lambda}(s)=j-1 \, , \qquad
l_{\lambda}(s)=\lambda_j'-i \, ,   \quad  {\rm and} \quad l_{\lambda}'(s)=i-1  \, .
\end{equation}
The quantities $a_{\lambda}(s),a_{\lambda}'(s),l_{\lambda}(s)$ and $l_{\lambda}'(s)$
are respectively called the arm-length, arm-colength, leg-length and
leg-colength.  For instance, if $\lambda=(8,5,5,3,1)$
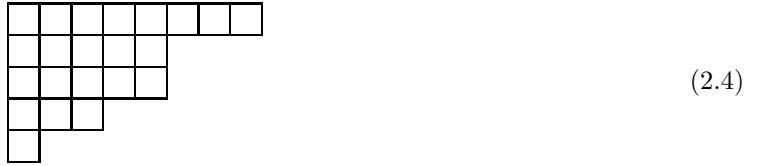
\begin{equation}
{\tableau[scY]{&&&&&&&\\&&&&\\&&&&\\& &  \cr \cr }}
\end{equation}
we have that $a_{\lambda}(3,2)=3,a_{\lambda}'(3,2)=1,
l_{\lambda}(3,2)=1$ and $l_{\lambda}'(3,2)=2$.
We say that the diagram $\mu$ is contained in $\la$, denoted
$\mu\subseteq \la$, if $\mu_i\leq \la_i$ for all $i$.  Finally,
$\la/\mu$ is a horizontal (resp. vertical) $n$-strip if $\mu \subseteq \lambda$, $|\lambda|-|\mu|=n$,
and the skew diagram $\la/\mu$ does not have two cells in the same column
(resp. row).

As mentioned in the introduction, a 
superpartition $\Lambda$
of fermionic degree $m$ is
a pair of partitions $\La=(\La^\cd, \La^*)$ 
such that the skew-diagram 
$\La^\cd/\La^*$ is both a vertical and  a horizontal $m$-strip.  
 Such a superpartition  is said to have degree  $(n|m)$ if $\sum_i \La_i^* =n$.  We refer to $n$ as
the total degree 
of $\La$. The length $\ell(\Lambda)$ of a superpartition $\Lambda$
is the length of the partition $\Lambda^\cd$.

A diagrammatic representation of $\La$ is given by 
the Ferrers diagram of $\La^*$  with
circles added in the cells corresponding to $\Lambda^\cd/\Lambda^*$.
 The  conjugate of a
superpartition $\La=(\Lambda^\cd, \Lambda^*)$ is simply given by
$\La'=({\Lambda^\cd}', {\Lambda^*}')$. Hence, as in the case of partitions,
$\Lambda'$ is
the superpartition whose diagram is obtained by interchanging the 
rows and the columns in the diagram of $\La$.  
For instance, as seen in the diagram that follows, the
conjugate of $\La=(\La^\cd,\La^*)=
\bigl((6,4,4,3,2,1),(5,4,3,3,1)\bigr)$ is
$\La'=\bigl((6,5,4,3,1,1),(5,4,4,2,1)\bigr)$.
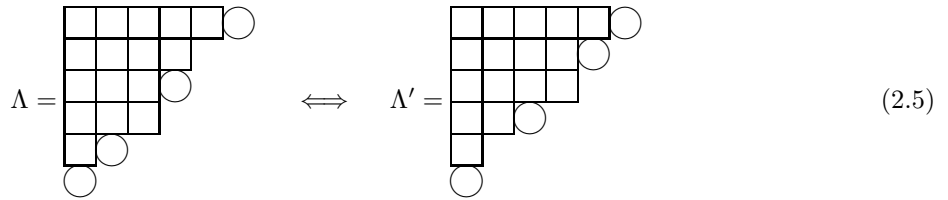
\begin{equation} \label{exdia}
\La={\tableau[scY]{&&&&&\bl\tcercle{}\\&&&\\&&&\bl\tcercle{}\\&&\\&\bl\tcercle{}\\ \bl\tcercle{}}}\,  \quad \iff \quad 
\La'={\tableau[scY]{&&&&&\bl\tcercle{}\\&&&&\bl\tcercle{}\\&&&\\&&\bl\tcercle{}\\&\bl\\ \bl\tcercle{}
}}
\end{equation}

We will occasionally need the original definition of a superpartition 
(see \cite{DLM1}): a superpartition $\La$ of length $\ell$ is
a pair of partitions 
\begin{equation}\label{sppa}
\La=(\La^{a};\La^{s})=(\La_1,\ldots,\La_m;\La_{m+1},\ldots,\La_\ell),
\end{equation}
such that
\begin{equation}\label{sppb}
\La_1>\ldots>\La_m\geq0 \quad  \text{ and}
\quad \La_{m+1}\geq \La_{m+2} \geq \cdots \geq
\La_\ell > 0 \, .\end{equation}
\noindent Note that $m$ corresponds in this definition to
the fermionic degree of $\Lambda$. 
The equivalence between the two definitions is quite obvious: {the parts of $\La$ that belong to $\La^a$ are the parts of $\La^*$ such that
$\La^\circledast_k-\La^*_k=1$.}

Finally, the dominance order on
 superpartitions  is defined as follows \cite{DLMeva}:
\begin{equation} \label{eqorder1}
 \Omega\leq\Lambda \quad \text{iff}
 \quad \Omega^* \leq \Lambda^*\quad \text{and}\quad
\Omega^{\circledast} \leq  \Lambda^{\circledast} , ,
\end{equation}
where the order on partitions is the usual dominance ordering:
\begin{equation}
\lambda \geq \mu \quad \iff \quad \sum_{i} \lambda_i =\sum_i \mu_i
\quad \text{and}\quad \lambda_1+ \cdots +\lambda_i \geq
\mu_1+ \cdots +\mu_i\text {~for all~} i
\end{equation}

\subsection{Monomial polynomials in superspace}
The simplest example of a symmetric superpolynomial is the super-version of the monomial polynomials.  Let $\Lambda=(\Lambda^a;\Lambda^s)$ be as in
\eqref{sppa}.  We then define
\begin{equation}
m_\La(x,\theta)={\sum_{\sigma \in S_N} }' \mathcal{K}_\sigma \left(\theta_{1}
\cdots\theta_{m} x_1^{\Lambda_1} \cdots x_N^{\Lambda_N} \right),
\end{equation} 
where the prime indicates that the sum is taken over
distinct permutations of $\theta_{1}
\cdots\theta_{m} x_1^{\Lambda_1} \cdots x_N^{\Lambda_N}$, {with the understanding that $\La_{\ell+1}=\cdots =\La_N=0$. }
This expression illustrates a generic property of  symmetric superpolynomials: because the $\ta_i$ are anticommuting, the polynomial in $x$ whose coefficient
is $\theta_1\cdots\theta_m$ is antisymmetric in the variables $x_1,\cdots, x_m$ and symmetric in the remaining ones.
This mixed symmetry of each component of $m_\La$ is clearly seen in the following example:
\begin{align}
m_{(1,0;1,1)}(x;\theta)=\;&\ta_1\ta_2(x_{1}-x_2)x_3x_4+\ta_1\ta_3(x_{1}-x_3)x_2x_4+\ta_1\ta_4(x_{1}-x_4)x_2x_3\nonumber\\+\,&\ta_2\ta_3(x_{2}-x_3)x_1x_4+\ta_2\ta_4(x_{2}-x_4)x_1x_3+\ta_3\ta_4(x_{3}-x_4)x_1x_2.
\end{align}
This example also illustrates the rationale for qualifying the number of
circles in $\Lambda$ as the fermionic degree: it is the number of 
$\ta$ factors in each constituent of the superpolynomial.

\subsection{Jack superpolynomials: eigenfunction characterization}
\label{charaphys}
The Jack superpolynomials $P_\La^{(\alpha)}$ can be characterized by the following
two conditions: 
\begin{align} \label{Ptriangular}
& (1)\quad P_\La^{(\aa)} =m_\La+\sum_{\Om<\La}
c_{\La\Om}(\alpha)\,m_\Om \, , \qquad \text{where} \quad  
c_{\Lambda \Omega}(\alpha) \in \mathbb Q (\alpha)
\\
\label{Dval}
& (2)\quad D\,P_\La^{(\alpha)}=
e_{\La^*}(\alpha)\,P_\La^{(\aa)} \qquad \text{and}\qquad \Delta\,
P_\La^{(\alpha)}=\tilde e_{\La}(\alpha)\,P_\La^{(\alpha)}\, ,
\end{align}
where the operators
$D$ and $\Delta$ are given by
\begin{align} \label{eqD}
 &D= \frac{1}{2}\sum_{i=1}^N  \alpha x_i^2\partial_{x_i}^2
+\sum_{1 \leq i\neq j \leq N}\frac{x_ix_j}{x_i-x_j}\left(\partial_{x_i}-\frac{\theta_i-\theta_j}{x_i-x_j}\partial_{\theta_i}\right), \\ \label{eqDelta}
& \Delta= \sum_{i=1}^N \alpha x_i\theta_i\partial_{x_i}\partial_{\theta_i}+
\sum_{1 \leq i\neq j \leq N}
\frac{x_i\theta_j+x_j\theta_i}{x_i-x_j}\partial_{\theta_i}, 
\end{align}
and their eigenvalues by
\begin{equation}
e_{\La^*}(\alpha)= \alpha b({\La^*}')-b(\La^*)\,  \quad {\rm and} \quad
\tilde e_\La(\alpha)=\alpha|\La^a|-|{\La'}^a| \, ,
\end{equation}
with $b(\la)=\sum_{i=1}^{\ell(\la)}(i-1)\la_i$.

The Jack superpolynomials of lowest degrees are tabulated in \cite[Appendix A]{DLM3}. In the absence of anti-commuting variables, the two conditions characterize the ordinary Jack polynomials \cite{Stan}.  Note that the operator $D$
is related to the operator $\mathscr{H}$ of \eqref{susyH} through the relation
\begin{equation}
2\alpha \Psi_0^{-1}\,( \mathscr{H}-
\mathscr{E}_0)\,\Psi_0=2D+(2N-2+{\alpha}) \sum_i
x_i\partial_{x_i}, 
\end{equation}
where $\Psi_0=\prod_{i<j}|x_i-x_j|^{1/\alpha}$, and $\mathscr{E}_0=N(2N-1)(N-1)/12\alpha^2$.
The operator $\Delta$ is related in a similar way  to a 
conserved operator of  the supersymmetric version of the trigonometric Calogero-Moser-Sutherland model. 
Actually, for this model there are $4N$ conserved quantities: $2N$ mutually commuting bosonic quantities, $H_n$ and $I_n$ for $n=1,\cdots, N$, and $2N$ non-commuting fermionic quantities
\cite{DLM1,DLM3,DLM7}. Up to conjugation by the ground-state
 wavefunction $\Psi_0$, the Jack superpolynomials are eigenfunctions of the $2N$ bosonic operators. 
The necessity of a double eigenvalue problem is clear from the fact that $D$ has degenerate eigenvalues, the  latter being insensitive to the `fermionic nature of the parts', that is, whether a row of $\La$ ends with a circle or not. 
This degeneracy is lifted by $\Delta$, whose 
eigenvalue only depends upon $\Lambda^a$, that is, upon the rows of 
$\Lambda$ that ends with a circle.

\subsection{Normalization, evaluation formula and duality}
When the number of variables is infinite, 
there is a natural scalar product 
on the space of symmetric superfunctions.  Let
\begin{equation}\label{spower}
p_\La=\tilde{p}_{\La_1}\cdots\tilde{p}_{\La_m}p_{\La_{m+1}}\cdots p_{\La_\ell},\qquad\text{where}\quad \tilde{p}_n=\sum_i\theta_ix_i^n\quad\text{and}\quad p_n=
\sum_ix_i^n \, .
\end{equation}  
We define
\begin{equation} \label{scap} \LL \, 
{p_\La} \, | \, {p_\Om }\, \RR_\alpha=(-1)^{\binom{m}2}\, \alpha^{{\ell}(\La)}\, z_{\La^s}
\delta_{\La,\Om}\,, \end{equation}
where $z_{\La^s} $ is given by
\begin{equation}  \label{zlam}
z_{\La^s}=\prod_{i \geq 1} i^{n_{\La^s}(i)} {n_{\La^s}(i)!}\, ,
\end{equation}
with $n_{\La^s}(i)$ the number of parts in $\La^s$ equal to $i$.

The expression of the norm of a Jack superpolynomial involves basic diagram data. Recall that for each cell $s=(i,j)\in\la$  we defined in \eqref{arml}
the arm-length  $a_\la(s)$ and  the  leg-length $l_\la(s)$ of the cell.
We define two $\a$-deformations of the hook length of a square in a superpartition $\La$,  the upper and  lower-hook lengths respectively given by \cite{DLMeva}:
\begin{align}\label{2hook}
&h^{(\alpha)}_\La(s)={l}_{\Lambda^{\circledast}}(s)+\alpha(a_{\La^*}(s)+1),\nonumber\\
&h^\La_{(\alpha)}(s)=l_{\La^*}(s)+1+\alpha\,{a}_{\Lambda^{\circledast}}(s).
\end{align} Note that these  generalize the two hook-lengths of \cite{Stan}.

Let $\B\La$ (the bosonic content of $\La$) be the set of squares in the diagram of  $\La$ that do not lie
at the {intersection of} a row containing a circle {and} a
column containing a circle. The  expression for the norm of a Jack superpolynomial reads:
\begin{align}  \label{norm}
 \|P_\La\|^2&:=(-1)^{\binom{m}{2}} \,
\LL P_\La^{(\alpha)}|P_\La^{(\alpha)}\RR=\prod_{s\in\B\La} \frac{h^{(\alpha)}_\La(s)}{h^\La_{(\alpha)}(s)}\quad
 \left( =
\prod_{s\in\La} \frac{h^{(\alpha)}_\La(s)}{h^\La_{(\alpha)}(s)} \right).
 \end{align} 
The expression in parenthesis follows from the equality of the two hook lengths for the squares in $\La/\B\La$; this alternative form will be useful below.

We now introduce the so-called evaluation formula for Jack superpolynomials
\cite{DLMeva}.
Let
\begin{equation}\label{defPs}
\P_{\Lambda,m}(x):=\frac{P_\Lambda^{(\alpha)}\big |_{\ta_1\cdots\ta_m}}{\Delta(x_1,\cdots,x_m)}.\end{equation}
The label $m$ reminds that $\P$ is symmetric with respect to to the first $m$ variables $(x_1,\cdots, x_m)$ and also symmetric with respect to the $N-m$ remaining ones.  Note that $\P_{\Lambda,m}(x)=
\S_{\Lambda^a,\Lambda^s}/\Delta(x_1,\cdots,x_m)$,  where $\S_{\Lambda^a,\Lambda^s}$
is the Jack polynomial with prescribed symmetry introduced 
in \eqref{prescribed}.
 The evaluation formula is most simply described in terms of the
skew diagram $\S\La=\La^\cd/(m,m-1,\dots,1)$, where as usual $m$ is the
fermionic sector of the superpartition $\Lambda$. We have
\begin{equation}\label{spe}
\P_{\Lambda,m}(x_1=\cdots=x_N=1)=\frac1{v_\La(\a)}  \prod_{s \in \S\La} 
b^{(\a,N)}_\La(s),
\end{equation}  
where  
\begin{equation}
b_\La^{(\alpha,N)}(s) =N-(i-1)+\alpha(j-1) \qquad {\rm and} \qquad
v_{\La}(\a) =  \prod_{s\in\B\La} {h^\La_{(\alpha)}(s)}\, .
\end{equation}
Note that  when $m=0$, the evaluation  formula for Jack superpolynomials
reduces to usual one (cf. \cite{Stan})
\begin{equation}\label{spec}
P_\la^{(\alpha)}(1,\cdots,1)= \prod_{(i,j) \in\la}\frac{N-(i-1) +\a(j-1)}{\la_j'-(i-1)+\a(\la_i-j)}\,.
\end{equation}

We conclude this section by mentioning  a
useful duality property of $P_\La^{(\alpha)}$.
Let $\hat \omega_{\alpha}$ stand for the endomorphism of the space
of symmetric polynomials in superspace
defined on the power
sums as
\begin{equation}
\hat \omega_{\alpha} (p_n) = (-1)^{n-1} \alpha \, p_n \qquad {\rm and}
\qquad \hat \omega_{\alpha} (\tilde p_n) = (-1)^{n} \alpha \, \tilde p_n.
\end{equation}
We have
\begin{equation}\label{dual}
\hat \omega_{\alpha} (P_{\La}^{(\alpha)})
= (-1)^{\binom{m}{2}}
\|P_\Lambda \|^2\, P_{\La'}^{(1/\alpha)} .
\end{equation}

\subsection{Non-symmetric Jack polynomials}

The characterization of the Jack superpolynomials that  will be most useful
in this article involves non-symmetric Jack polynomials.
 We collect here
the most relevant properties of the non-symmetric Jack polynomials.
Their relation to the Jack superpolynomials will be presented in the following subsection.

We consider the Dunkl-Cherednik operators (see \cite{Ber,Knop,Op})
\begin{equation}
{\mathcal D}_i = \alpha x_i {\partial_{x_i}} + \sum_{1\leq j<i} \frac{x_i}{x_i-x_j} (1-K_{i,j}) + \sum_{i<j \leq N} \frac{x_j}{x_i-x_j} (1-K_{i,j}) + 1-i \, ,
\end{equation}
where $K_{i,j}$ is the operator that exchanges the variables $x_i$ and $x_j$.
The $\mathcal D_i$'s are mutually commuting  $[\mathcal D_i,\mathcal D_j]=0$
and obey the relations
\begin{equation} \label{relhecke}
\mathcal D_i K_{i,i+1} - K_{i,i+1} \mathcal D_{i+1} = 1 \qquad \text{and} \qquad
 \mathcal D_i K_{j,j+1} = K_{j,j+1} \mathcal D_{i} \quad {\text{if} } \quad
 i \neq j,j+1 \, .
\end{equation}

The non-symmetric Jack polynomial, $E_{\eta}(x;\alpha)$, indexed
by a composition with $N$ parts (some of them possibly equal to zero),
can be characterized as the unique
polynomial, whose coefficient of
$x^{\eta}$ is equal to 1, such that
\begin{equation}
{\mathcal D}_i \, E_{\eta}(x;\alpha) =  \bar \eta_i \, E_{\eta}(x;\alpha) \qquad
 \forall i=1,\dots,N,
\end{equation}
where the eigenvalue $\bar \eta_i$
is given by
\begin{equation}\label{etab}
\bar \eta_i = \alpha \eta_i - \#\{ k < i \, | \, \eta_k \geq \eta_i \} -
 \#\{ k > i \, | \, \eta_k > \eta_i \}.
\end{equation}

As is the case for partitions, the diagram of
 a composition $\eta$ with $N$ parts
is the set of cells $(i,j)$
such that
$1 \leq i \leq N$ and $1 \leq j \leq \eta_i$.  For instance,
if $\eta=(0,1,3,0,0,6,2,5)$, the diagram of $\eta$ is
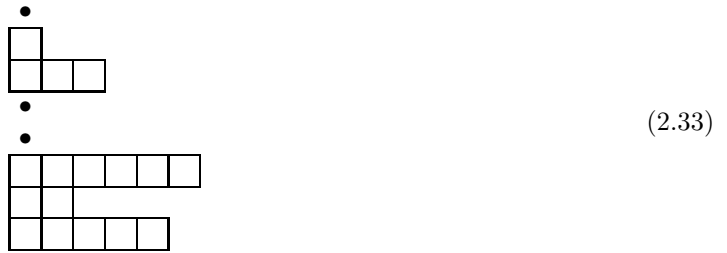
\begin{equation}
{\tableau[scY]
{\bl \bullet \\  \\ && \\ \bl \bullet \\ \bl \bullet \\  & & & & &  \\
 & \\
& & & &  \\ }}
\end{equation}
where a $\bullet$ represents an entry of length zero.   For each cell
$s=(i,j) \in \eta$, we define the following 
hook-length $d_{\eta}(s)$ \cite{Knop}:
\begin{equation}\label{ald}
d_{\eta}(s)  =  \alpha(  \eta_i -j +1) + \# \{ k<i \, | \, j \leq \eta_k+1 \leq \eta_i  \} + \# \{ k>i \, | \, j \leq \eta_k \leq \eta_i  \} +1.
\end{equation} 

\subsection{Jack superpolynomials: symmetrization construction}

 Given a superpartition $\La=(\La^a;\La^s)$ of the form \eqref{sppa},
we define the composition $\tilde \Lambda$ as
\begin{equation}
\tilde \Lambda := (\Lambda_m,\dots,\Lambda_1,\Lambda_N,\dots,\Lambda_{m+1}) \, .
\end{equation}
It was shown  in \cite[Section 9]{DLM1} that the Jack superpolynomials
defined in Section~\ref{charaphys} can
be obtained from the
non-symmetric
Jack polynomials through the following relation:
\begin{equation} \label{jackinnonsym}
{P}_{\Lambda}^{(\alpha)} = \frac{(-1)^{\binom{m}{2}}}{f_{\Lambda^s}}
\sum_{w \in S_N} {\mathcal K}_{w} \, \theta_1 \cdots \theta_m \,
E_{\tilde \Lambda}(x;\alpha) \, ,
\end{equation}
where $f_{\Lambda^s}$ stands for
\begin{equation}  f_{\La^s}=\prod_i n_{\La^s}(i)!, \end{equation}
with $n_{\La^s}(i)$ being the multiplicity of $i$ in $\La^s$ defined in 
\eqref{zlam}
and $\mathcal K_w$ is defined in \eqref{defKk}.

Note that the composition $\tilde \Lambda$ is of a very special form.  Its
first $m$ rows (resp. last $N-m$ rows)
are strictly increasing (resp. weakly increasing).
Diagrammatically, it is made of two partitions (the first one of
which without repeated parts)
drawn in the French
notation (largest row in the bottom).
For instance if $\Lambda=(3,1,0;5,3,3,0,0)$, we have
$\tilde \Lambda=(0,1,3,0,0,3,3,5)$ whose diagram is given by
\begin{equation}
{\tableau[scY]
{\bl \bullet \\  \\ && \\ \bl \bullet \\ \bl \bullet \\  & &  \\
 & & \\
& & & &  \\ }}
\end{equation}

\section{The Sekiguchi operator and its superspace extension}
\label{Sseki}
A key technical tool in the study of the vanishing conditions of the ordinary Jack polynomials 
in \cite{FJMM} is their characterization  as eigenfunctions of the Sekiguchi operators. 

This construction is closely related to the description of $P_\la^{(\alpha)}$ 
as the symmetrization of a non-symmetric Jack polynomial, which can be seen as the special case  $m=0$ of \eqref{jackinnonsym}. This construction is also rooted in the integrability of the underlying eigenvalue problem, namely, the Calogero-Moser-Sutherland model: the Sekiguchi operator is a generating function of 
$N$ independent conserved quantities in involution.

Since the Jack superpolynomials are eigenfunctions of two basic operators, or more generally, two independent mutually commuting sets of $N$ operators, we need here two  Sekiguchi-type operators. Their introduction (which is new) is the subject of this section. The results that follow rely heavily on the description of $P_\La^{(\alpha)}$ given in  \eqref{jackinnonsym}.

Our pair of Sekiguchi operators is composed of the usual Sekiguchi operator
\begin{equation}
S(u,\alpha) = \prod_{i=1}^N (\mathcal D_i +u),
\end{equation}
together with its supersymmetric counterpart
\begin{equation}
\tilde S(u,\alpha) = \sum_{m=0}^N  \frac{1}{m! (N-m)!}
\sum_{\sigma \in S_N} \mathcal K_{\sigma}
\left(\prod_{i=1}^m (\mathcal D_i+\alpha +u) 
 \prod_{j=m+1}^N
(\mathcal D_j +u)\right) \pi_{1,\dots,m}\, ,
\end{equation}
where
\begin{equation}
\pi_{1,\dots,m}  = \prod_{i=1}^m \theta_i \partial_{\theta_i}  \prod_{j=m+1}^N (1-\theta_j \partial_{\theta_j}) 
\end{equation}

The correctness of these choices is justified by the {proposition} that follows.
But let us observe at once that it is not surprising to find 
the usual Sekiguchi operator $S(u,\a)$ among our pair of operators. 
The operators $S(u,\a)$ and $\tilde S(u,\a)$ are the generating functions of the bosonic conservation laws for the supersymmetric trigonometric Calogero-Moser-Sutherland model.
In particular, $S(u,\a)$ generates the elementary symmetric functions in the
quantities $\mathcal D_i$'s, which are functionnaly equivalent to the
quantities $\mathcal H_n$ of \cite{DLM1,DLM3}.  The connection between
$\tilde S(u,\a)$ and the quantities $\mathcal I_n$ of \cite{DLM1,DLM3}
is less obvious.
\begin{proposition} \label{LemmaSeki}
We have 
\begin{equation}
S(u,\alpha) \, P_{\Lambda}^{(\aa)} = \varepsilon_{\Lambda^*} (u,\alpha) \, 
P_{\Lambda}^{(\aa)}
\qquad {\rm and} \qquad \tilde S(u,\alpha) \, 
P_{\Lambda}^{(\aa)} = \varepsilon_{\Lambda^\circledast} (u,\alpha) \, P_{\Lambda}^{(\aa)} \, ,
\end{equation}
where $\varepsilon_{\lambda}(u,\alpha)$ is given, for
any partition $\lambda$, by
\begin{equation}
\varepsilon_{\lambda}(u,\alpha) = \prod_{i=1}^N \bigl(\alpha \lambda_i + 1-i+u \bigr)\,.
\end{equation}
\end{proposition}
\begin{proof}
It is easy to verify, using \eqref{relhecke}, 
that for all $i=1,\dots,N-1$ we have 
\begin{equation}\label{eqcommu}
K_{i,i+1}(\mathcal D_i+c)(\mathcal D_{i+1}+c)=(\mathcal D_i+c)(\mathcal D_{i+1}+c)K_{i,i+1}\, ,
\end{equation}
where $c$ is an arbitrary constant.
Hence $K_{w} S(u,\a)=S(u,\a) K_{w}$ 
for all $w \in S_N$.   And since $S(u,\a)$ does not act on the variables $\theta$, we also have ${\mathcal K}_{w} \theta_1 \cdots \theta_m
S(u,\a)=S(u,\a) {\mathcal K}_{w} \theta_1 \cdots \theta_m$ 
for all $w \in S_N$.  Therefore, using \eqref{jackinnonsym},
to prove the first statement of the proposition
we simply need to show that
\begin{equation} \label{eqseki1}
 \left(\prod_{i=1}^N (\mathcal D_i+u) \right) E_{\tilde \Lambda}(x;\alpha)
 = \varepsilon_{\Lambda^*}(u,\alpha) \, 
E_{\tilde \Lambda}(x;\alpha)  .
\end{equation}

Similarly, we will now show that to prove the second statement, it suffices
to prove that
\begin{equation} \label{eqseki2}
\left( \prod_{i=1}^m (\mathcal D_i+\alpha+u) \prod_{j=m+1}^N
(\mathcal D_j +u) \right) 
E_{\tilde \Lambda}(x;\alpha) = \varepsilon_{\Lambda^{\circledast}}(u,\alpha) \, 
E_{\tilde \Lambda}(x;\alpha)  \, .
\end{equation}
First observe that 
\begin{equation}
\pi_{1,\dots,l}\, {\mathcal K}_w \theta_1 \cdots \theta_m =
\left\{
\begin{array}{ll}
 {\mathcal K}_w \theta_1 \cdots \theta_m & {\rm if~} l=m {\rm ~and~} w 
\in S_m \times S_{N-m} \\
0 & {\rm otherwise}
\end{array}
\right.
\end{equation}
where $S_m \times S_{N-m}$ stands for the subgroup of $S_N$ made out of permutations of the
first $m$ elements and the last $N-m$ elements respectively.  Using
\eqref{jackinnonsym} again (forgetting the multiplicative factor),
we get
\begin{align}
& \tilde S(u,\alpha)  \,  \sum_{w \in S_N} {\mathcal K}_w 
\theta_1\dots \theta_m E_{\tilde \Lambda}(x;\alpha) \nonumber \\
& \qquad =\frac{1}{m! (N-m)!}
\sum_{\sigma \in S_N} {\mathcal K}_{\sigma} \left( \prod_{i=1}^m (\mathcal D_i+\alpha+u) \prod_{j=m+1}^N
(\mathcal D_j +u) \right) \sum_{w \in S_m \times S_{N-m}} {\mathcal K}_w 
\theta_1\dots \theta_m E_{\tilde \Lambda}(x;\alpha)
\end{align}
{From} \eqref{eqcommu}, we can deduce that 
$\prod_{i=1}^m (\mathcal D_i+\alpha+u) \prod_{j=m+1}^N
(\mathcal D_j +u)$ commutes with $\mathcal K_w \theta_1 \cdots \theta_m$ 
for every
$w \in S_m \times S_{N-m}$.  Hence
\begin{equation}
\tilde S(u,\alpha)
\,  \sum_{w \in S_N} {\mathcal K}_w 
\theta_1\dots \theta_m E_{\tilde \Lambda}(x;\alpha)
= 
\sum_{\sigma \in S_N} {\mathcal K}_{\sigma} \theta_1 \cdots \theta_m
\left( \prod_{i=1}^m (\mathcal D_i+\alpha+u) \prod_{j=m+1}^N
(\mathcal D_j +u) \right) 
E_{\tilde \Lambda}(x;\alpha)
\end{equation}
and, as claimed, \eqref{eqseki2} implies the second statement of the
proposition.

We have left to prove expressions \eqref{eqseki1} and \eqref{eqseki2}.
Let $\eta= \tilde \La$ and suppose that $\eta_i=r$.
It is easy to get from \eqref{etab} that the eigenvalue $\bar \eta_i$
of $\mathcal D_i$ is 
\begin{equation}
\bar \eta_i = \alpha r - \#\{{\rm rows~of~} \Lambda^* 
{\rm~of~size~larger~than~}r\} 
 - \#\{{\rm rows~of~} \tilde \La {\rm~of~size~}r {\rm ~above~row~}i\} \, .
\end{equation}
Therefore, letting 
\begin{equation} \label{eqji}
j_i=\#\{{\rm rows~of~} \Lambda^* {\rm~of~size~larger~than~}r\} 
 + \#\{{\rm rows~of~} \tilde \La {\rm~of~size~}r {\rm ~above~row~}i\} +1
\end{equation}
we have $\{j_1,\dots,j_N \}=\{1,\dots,N\}$, $\Lambda_{j_i}^*=r$,
and $\bar \eta_i = \alpha \Lambda_{j_i}^* +1-{j_i}$, which gives 
\eqref{eqseki1}.  

Continuing with the same notation, we have that if
$i$ belongs to $\{1,\dots,m\}$ then
$\eta_i=r$ is the highest row of size $r$ in $\eta$,
and thus by \eqref{eqji} $\Lambda_{j_i}^*$ is also the highest row of size
$r$ in $\Lambda^*$.
Hence, in this case
\begin{equation}
\bar \eta_i +\alpha= \alpha \Lambda_{j_i}^{\circledast} +1-{j_i} \, .
\end{equation}
If $i$ does not belong to $\{1,\dots,m\}$, then we have as before
\begin{equation}
\bar \eta_i = \alpha \Lambda_{j_i}^* +1-{j_i}  = 
\alpha \Lambda_{j_i}^{\circledast} +1-{j_i} 
\end{equation}
and \eqref{eqseki2} follows. 
\end{proof}

Proposition~\ref{LemmaSeki} has the following important corollary.
\begin{corollary} \label{corotri} We have
\begin{align}
S(u,\alpha) m_{\Lambda} &= \varepsilon_{\Lambda^*}(u,\alpha) m_{\Lambda}
+ \sum_{\Omega < \Lambda} b_{\Lambda \Omega}(u,\alpha) \, m_{\Omega}, \\
\tilde S(u,\alpha) m_{\Lambda}& = \varepsilon_{\Lambda^{\circledast}}(u,\alpha) m_{\Lambda}
+ \sum_{\Omega < \Lambda} \tilde b_{\Lambda \Omega}(u,\alpha) \, m_{\Omega}, 
\end{align}
 for some $b_{\Lambda \Omega}(u,\alpha),\,\tilde b_{\Lambda \Omega}(u,\alpha) \in \mathbb Q(u,\alpha)$.
\end{corollary}
\begin{proof}
Suppose there exists a $\Gamma$ such that  $S(u,\alpha)$ on  $m_{\Gamma}$
is not triangular, and take $\Lambda$ to be minimal among those $\Gamma$'s.
Then, by  \eqref{Ptriangular}, 
\begin{equation}
S(u,\alpha) P_{\Lambda}^{(\aa)} = S(u,\alpha) m_{\Lambda} + \sum_{\Omega < \Lambda} c_{\Lambda \Omega}(\alpha) S(u,\alpha) m_{\Omega}
\end{equation}
By our hypothesis on the minimality of $\Lambda$, 
all the monomials $m_{\Delta}$ that occur in 
$S(u,\alpha) m_{\Omega}$ are such that $\Lambda > \Omega \geq \Delta$.
Therefore if $m_{\Upsilon}$, with $\Upsilon \not \leq \Lambda$, appears in
$S(u,\alpha) m_{\Lambda}$ then it also appears in $S(u,\alpha) P_{\Lambda}^{(\aa)}$.
This contradicts Proposition~\ref{LemmaSeki}.
\end{proof}

\section{Admissible superpartitions  and regularity of the Jack superpolynomials }
\label{Sregu}
\subsection{Admissibility conditions for superpartitions}\label{admS}
Let $k\geq 1$ and $r \geq 2$ be integers such that $k+1$ and $r-1$ are
coprime.
As stated in the introduction, we say
that a superpartition $\Lambda$ is $(k,r,N)$-admissible
 (allowing zeroes as entries in $\Lambda^{\cd}$ and $\Lambda^*$) if
\begin{equation}\label{admi}
\La^\cd_i-\La^*_{i+k}\geq r \qquad (1 \leq i \leq N-k) \, .
\end{equation}
For $m=0$, in which case $\La=(\Lambda^*,\Lambda^*)$ 
is essentially an ordinary partition, this reduces to the usual 
admissibility criterion \cite{FJMM}.

As a short digression, let us make a remark on the enumeration of admissible (super)partitions when $r=2$. The $(k,2,N)$-admissible partitions  are precisely those that describe the combinatorics of the sum-side of the Andrews-Gordon generalization of the Rogers-Ramanujan identity \cite{An,Andr}. 
Similarly, the enumeration of the  $(k,2,N)$-admissible superpartitions are captured by the sum-side of a generalization of the Andrews-Gordon identity. The simplest way of obtaining this connection is to  transform a superpartition $\La$ into an overpartition $\Om$ by considering $\La^\cd$ and putting  an over-bar above each entry $\La_i^\cd$ for which $\La^\cd_i-\La^*_i=1$. As we just mentioned, this produces an overpartition \cite{CL}, namely, a partition where  the final occurrence of a part can be overlined.  For instance, we have
\begin{equation}
{\tableau[scY]{&&&\bl\tcercle{}\\&\\&\bl\tcercle{}\\ \\\bl\tcercle{}}}
\longleftrightarrow \quad (\bar{4},2,\bar{2},1,\bar{1}) .
\end{equation}
With this map, a $(k,2)$-admissible superpartition is transformed into an overpartition that satisfies precisely the restriction condition introduced in \cite{Lo,CM}: $\Omega_i - \Omega_{i+k}\geq 1 $ if $\Omega_{i+k}$ is overlined and $\geq 2$ otherwise. The generating function for these restricted overpartitions that generalizes the Andrews multiple-sum is presented in Theorems 1.4, 1.5 and eq. (6.1) of \cite{CM}. Recall that in superconformal highest-weight modules,  the states are generated by the action  of the Virasoro modes $L_n$ and its supersymmetric partner $G_r$, for $n,r<0$. In the Neveu-Schwarz sector $r\in\Z+\frac12$, while it  is integer in the Ramond sector.

\subsection{Regularity of $P_{\La}^{(\ad)}$ when $\La$ is admissible}
As in the non-supersymmetric case, singularities can occur in a Jack superpolynomial $P_{\Lambda}^{(\alpha)}$ at the special value $\a=\ad=-(k+1)/(r-1)$.
Given this, our first task is 
to verify that 
the Jack superpolynomials $P^{(\ad)}_{\La}$ are regular at $\ad$ 
(i.e., do not have poles) when $\La$ is $(k,r,N)$-admissible.  
 As already pointed out in the introduction, following relation
\eqref{jackinnonsym}
this can be deduced from the results of \cite{Kas}.  We choose nevertheless
to present our own proof of the regularity:  the methods it uses
will then be used again and again in the much harder proof of
Proposition~\ref{propquasi}.

It is known that $\left(\prod_{s \in \tilde \Lambda}
d_{\tilde \Lambda}(s)\right)\, E_{\tilde \Lambda}$, 
with $d_\eta$ defined in \eqref{ald},
belongs to $\mathbb N[\alpha,x_1,\dots,x_N]$ (see \cite[Theorem 4.11]{Knop}).
Therefore, using \eqref{jackinnonsym}, we have that 
$\left( \prod_{s \in \tilde \Lambda} d_{\tilde \Lambda}(s) \right) \, P_{\Lambda}^{(\alpha)}$ cannot have singularities at $\alpha=\ad$.  In other words, checking the
regularity of $P^{(\ad)}_{\Lad}$ amounts to
checking that
$
\prod_{s \in \tilde \Lambda}
d_{\tilde \Lambda}(s)
$
does not have zeroes when $\alpha =\ad$ if $\La = \Lad$.  

The following lemma gives a simpler expression for $\prod_{s \in \tilde \Lambda}
d_{\tilde \Lambda}(s)$.
\begin{lemma} We have
\begin{equation}
\prod_{s \in \tilde \Lambda}
d_{\tilde \Lambda}(s) = f_{\La^s}^{-1}\prod_{s\in \La} {h^\La_{(\alpha)}(s)}  
\prod_{i=N-\ell(\Lambda^s)}^{N} d_{\tilde \Lambda}\bigl((i,1)\bigr) 
,\end{equation}
where $f_{\La^s}$, $h^\La_{(\alpha)}(s)$ and $d_{\tilde \Lambda}(s)$ are defined respectively in \eqref{zlam}, \eqref{2hook}, and \eqref{ald}. 
\end{lemma}
\begin{proof}
\cite[Lemma 6]{LLN} states that
\begin{equation}
\prod_{s \in \tilde \Lambda}
d_{\tilde \Lambda}(s) = f_{\La^s}^{-1}\prod_{s\in {\mathcal B\La}} {h^\La_{(\alpha)}(s)}  
\prod_{i=N-\ell(\Lambda^s)}^{N} d_{\tilde \Lambda}\bigl((i,1)\bigr)
\prod_{1 \leq j < i \leq m} d_{\tilde \Lambda}\bigl((i,\tilde \Lambda_j+1)\bigr)
\end{equation}
The lemma then follows from the relation
\begin{equation}
\prod_{\Lambda/{\mathcal B \Lambda}} h^\La_{(\alpha)}(s)=
\prod_{1 \leq j < i \leq m} d_{\tilde \Lambda}\bigl((i,\tilde \Lambda_j+1)\bigr) \, .
\end{equation}
\end{proof}
Note that it is conjectured in \cite[Conjecture 33]{DLM7}
that the integral form (see \eqref{JvsP})
of the Jack superpolynomials is given by
\begin{equation}
J_{\Lambda}^{(\alpha)} = v_{\Lambda} (\alpha) \, P_{\Lambda}^{(\alpha)} 
\end{equation}
where $v_{\Lambda} (\alpha)=\prod_{s\in {\mathcal B\La}} {h^\La_{(\alpha)}(s)}$.
If this were proven, it 
would suffice (as is the case for Jack polynomials)
to show that the simpler expression 
$v_{\Lad} (\ad)$ does not have zeroes to demonstrate 
the regularity of the admissible
Jack superpolynomials.

We now show that $P^{(\ad)}_{\La}$ is regular.
\begin{proposition} \label{lemmapoles}
The Jack superpolynomial $P_{\La}^{(\alpha)}$ has no pole at
$\alpha=\ad$ when $\La$ is $(k,r,N)$-admissible.
\end{proposition} 
\begin{proof}
As mentioned before, we need to prove that 
$\prod_{s \in {\tilde \Lambda}} d_{\tilde \Lambda}(s)$ does not vanish
at $\alpha=\ad$. 
 In view of the previous lemma, it is sufficient to show
that
$h^\Lambda_{(\ad)}(s)\ne0$  
for any $s \in \Lambda$ and that $d_{\tilde \Lambda}((i,1)) \neq 0$ when 
$\alpha=\ad$
for any $i$ such that $(i,1)$ belongs to $\tilde \Lambda$.

We first show that $h^\Lambda_{(\ad)}(s)\ne0$  at $\alpha=\ad$
for any $s \in \Lambda$.  This amounts to showing that
\begin{equation} \label{coco}
{\La^*_j}'-i+1+\ad(\La^\cd_i -j)\not=0 \qquad \forall\, (i,j)\in \Lambda.
\end{equation}
This is proved as in Lemma 2.1 and 2.2 of \cite{FJMM}  by deriving a contradiction (namely $t(r-1)\geq tr$ for a positive integer $t$).
The argument relies crucially on the fact that $r-1$ and $k+1$ are coprime and also on the following simple consequence of the admissibility criterion:
\begin{equation}\label{ine1}
\La^\cd_i-\La^*_j\geq \left\lfloor \frac{j-i}{k}\right \rfloor r .
\end{equation}
We also need  the following (obvious) inequality: for any partition $\la$, we have 
$\la_{\la'_j}\geq j
$.
 Indeed, recall that 
\begin{equation}
\la'_i = {\rm Card}\; \{j: \la_j\geq i\}\quad \Rightarrow \quad \la_i = {\rm Card}\; \{j: \la'_j\geq i\}, \end{equation}
so that with $i = \la_j'$ 
\begin{equation} \label{ine2}
\la_{\la'_j}= {\rm Card}\; \{\ell: \la'_\ell\geq \la'_j\}\geq j , \end{equation}
that is, the number of columns that are $\geq \la'_j$  is certainly at least $j$. It is also clear that
\begin{equation} \label{ine3}
\la_{\la'_j}= j \qquad\text{if}\qquad \la'_{j+1}<\la'_j. \end{equation}

We are now in position to establish our result. Suppose that (\ref{coco})  is violated, that is that for some square $s=(i,j)\in \La$ we have
\begin{equation} \label{cocof}
{\La^*_j}'-i+1=\frac{k+1}{r-1}(\La^\cd_i -j). 
\end{equation}
 This requires 
\begin{equation}
\La_i^\cd-j= t(r-1)\quad \text{and} \quad {\La_j^*}'-i+1= t(k+1)\quad\text{ for some integer $t$}. 
\end{equation}
We have that $\La^\cd_i-j\geq 0$ implies $t\geq 0$.   We first rule out
the case $t=0$. If $\La^\cd_i-j=0$ then  $\La^\cd_i=j$.  
But for $(i,j)\in\La$, this requires $\La^\cd_i=\La^*_i$ and the equality $\La^*_i=j$ would then imply that ${\La^*_j}'\geq i$ contradicting the other resulting  relation ${\La^*_j}'=i-1$. Using the inequalities (\ref{ine1}) and (\ref{ine2}), we then have the contradiction  
\begin{equation}
t(r-1)= \La^\cd_i-j \geq \La^\cd_i-\La^*_{{\La^*_j}'} \geq \left\lfloor \frac{{\La^*_j}'-i}{k}\right\rfloor r\geq \left\lfloor \frac{t(k+1)-1}{k}\right\rfloor 
r\geq tr ,\quad\text{with}\; t\geq1 \, . \end{equation}
 This demonstrates that $h^{\La}_{(\ad)}(s)\ne 0$ for any $s\in \La$. 

We now have left to show that
$d_{\tilde \Lambda}\bigl((i,1)\bigr) \neq 0$ at $\alpha=\ad$
for any $i$ such that $(i,1)$ belongs to $\tilde \Lambda$.
Let $\Omega=(\La^* + 1^N, \La^{\cd}+1^N)$ (the superpartition obtained
by adding a column of length $N$ to  $\Lambda$).  It is easy to see
that $\Omega$ is still $(k,r,N)$-admissible and that 
$d_{\tilde \Lambda}\bigl((i,1)\bigr)=h_{(\alpha)}^{\Omega}(j,1)$ for some $j$.  
The result then follows from 
what we showed above: $h^{\La}_{(\ad)}(s)\ne 0$ for any $s\in \La$. 
\end{proof}

\subsection{Regularity of $P_\La^{(\ad)}$ when $\La$ is almost admissible}

We will now show  that the Jack superpolynomials indexed by
superpartitions obtained from a 
$(k,r,N)$-admissible superpartition by removing (resp. adding) a circle or
by changing a circle (resp. square) into a square (resp. circle) are also
regular
at $\alpha= \ad$ (these partitions will be called
 {\it almost} $(k,r,N)$-admissible).
This is somewhat more subtle than
showing the regularity of $P_{\La}^{(\ad)}$ when $\La$ is $(k,r,N)$-admissible. 
 The proof generalizes the methods developed in \cite{FJMM} for a similar purpose.
\begin{lemma} If $P_{\La}^{(\alpha)}$ has a pole at $\alpha =\alpha_0$ then
there exists a partition $\Omega < \Lambda$ such that
\begin{equation}
\varepsilon_{\Omega^*}(u,\alpha) 
= \varepsilon_{\Lambda^*}(u,\alpha) \qquad {\rm and} \qquad
\varepsilon_{\Omega^{\cd}}(u,\alpha) 
= \varepsilon_{\Lambda^{\cd}}(u,\alpha).
\end{equation}
\end{lemma}
\begin{proof}
The proof will be skipped since it is 
essentially the same as the one of \cite[Lemma 2.4]{FJMM}.
The proof relies on Proposition~\ref{LemmaSeki}, Corollary~\ref{corotri}
and the triangularity \eqref{Ptriangular}.
\end{proof}

\begin{lemma}\label{lemsw}
 If $P_{\Lambda}^{(\alpha)}$ has a pole at $\alpha=\alpha_0$, then there exists
a partition $\Omega < \Lambda$ and permutations $w,\sigma \in S_N$ (at least one of them distinct from the identity) such that
\begin{align} \label{3.9}
\Omega^*_i &= \Lambda_{w(i)}^* + (w(i)-i) \frac{r-1}{k+1} ,\\ \label{3.10}
w(i) & \equiv i \mod k+1,
\end{align}
and
\begin{align} \label{3.11}
\Omega^\cd_i &= \Lambda_{\sigma(i)}^\cd + (\sigma(i)-i) \frac{r-1}{k+1}, \\
\label{3.12}
\sigma(i) & \equiv i \mod k+1.
\end{align}
\end{lemma}
\begin{proof}
This follows from the previous lemma as in \cite[Lemma 2.5]{FJMM}.
\end{proof}

\begin{lemma}\label{lemadm}
  Let $\Lambda$ be a $(k,r,N)$-admissible superpartition.
Then $\Lambda^*$ and $\Lambda^\cd$ are $(k+1,r,N)$-admissible partitions.
In particular,
\begin{equation}\label{2ine}
\La^*_i-\La^*_j\geq \left\lfloor \frac{j-i}{k+1}\right\rfloor r \qquad\text{and}\qquad \La^\cd_i-\La^\cd_j\geq \left\lfloor \frac{j-i}{k+1}\right\rfloor r \,.
\end{equation}
\end{lemma}
\begin{proof}
We have $\Lambda_{i+1}^\cd-\Lambda_{i+k+1}^* \geq r$, and thus
\begin{equation}\Lambda_{i}^*-\Lambda_{i+k+1}^* \geq \Lambda_{i+1}^*-\Lambda_{i+k+1}^* \geq r-1 \, .
\end{equation}
The equality $\Lambda_{i+1}^*-\Lambda_{i+k+1}^* = r-1$ 
can
only occur if $\Lambda_{i+1}^\cd-\Lambda_{i+1}^*=1$, in which case, 
$\Lambda^*_i \geq \Lambda_{i+1}^\cd
 > \Lambda_{i+1}^*$.  This gives $\Lambda_{i}^*-\Lambda_{i+k+1}^* \geq r$
as wanted.

Similarly, we have
\begin{equation}\Lambda_{i}^\cd-\Lambda_{i+k+1}^\cd \geq \Lambda_{i}^\cd-\Lambda_{i+k}^\cd \geq r-1 \, .
\end{equation}
This time, the case $\Lambda_{i}^\cd-\Lambda_{i+k}^\cd=r-1$ can
only happen if $\Lambda_{i+k}^\cd-\Lambda_{i+k}^*=1$, in which case
$\Lambda^\cd_{i+k} > \Lambda_{i+k}^*
 \geq \Lambda_{i+k+1}^\cd$.
This gives $\Lambda_{i}^\cd-\Lambda_{i+k+1}^\cd \geq r$.
\end{proof}

The next proposition says that the 
almost $(k,r,N)$-admissible superpartitions
are regular.  The proof, being rather involved, will be 
relegated to Appendix \ref{Prop8}.
\begin{proposition} \label{propquasi}
 Let $\Lambda$ be a superpartition obtained from a 
$(k,r,N)$-admissible superpartition $\Gamma$ by doing one of the following:
\begin{itemize}
\item[i)] removing a circle
\item[ii)] adding a circle
\item[iii)] changing a circle into a square
\item[iv)] changing a square into a circle
\end{itemize}
Then $P_{\Lambda}^{(\aa)}$ does not have a pole at $\alpha=\ad$.
\end{proposition}

\section{Lie superalgebras of differential operators}\label{Lie}

In this section, we introduce two sets  of differential operators.
The first one
forms a super Lie algebra isomorphic to $sl(1,2)$ 
while the other one  gives the negative-half of the super-Virasoro algebra. 
In the following section, we will determine the action of some elements of
these algebras on a generic $P_\La^{(\a)}$. These results
will then be used to
characterize the ideal spanned by the $P_{\La}^{(\ad)}$'s indexed by
$(k,r,N)$-admissible superpartitions.

We define the following bi-homogeneous differential operators of first order 
whose action preserves
the ring ${\bf \Lambda}^{\theta}=\mathbb C[x_1,\dots,x_N,\theta_1,\dots,\theta_N]^{S_N}$.
\begin{align}\label{homo}
\nabla& =\sum_{i=1}^N\partial_{x_i} &
\nabla^\perp& =\sum_{i=1}^Nx_i\left(x_i\partial_{x_i}+\theta_i\partial_{\theta_i}+\frac{N}{\alpha}\right)\nonumber\\
q& =\sum_{i=1}^N\theta_i\partial_{x_i} &
q^\perp& =\sum_{i=1}^Nx_i\partial_{\theta_i}\nonumber\\
Q&= \sum_{i=1}^N\theta_i\left(\frac{N}{\alpha}+x_i\partial_{x_i}\right)
&
Q^\perp& =\sum_{i=1}^N\partial_{\theta_i}\nonumber\\
E&= \sum_{i=1}^N\left(\frac{N}{\alpha}+ x_i\partial_{x_i}\right) &
\mathcal{E}& =\sum_{i=1}^N\left(x_i\partial_{x_i}+ \theta_i\partial_{\theta_i}\right).
\end{align}
The operators given above generate an eight-dimensional Lie superalgebra $\mathcal A$ whose 
(anti-)commutation relations are given in the following table.\footnote{The algebra $\mathcal A$ is isomorphic to $sl(1|2)$ (see for instance \cite{Dic}), 
with the following correspondence:
$$H=\frac12(E+\mathcal E), \quad Z=\frac12(E-\mathcal E),\quad E^+=i\nabla^\perp,\quad E^-=i\nabla,\quad F^+=iq^\perp,\quad F^-=Q^\perp,\quad \bar F^+=Q,\quad \bar F^-=iq.$$}
\begin{center}\label{tab1}
\begin{tabular}{l|cccccccc}
              &$ E$ & $\mathcal{E}$ & $q$ & $Q$ &$ q^\perp $& $Q^\perp$   & $\nabla$   & $\nabla^\perp$  \\ \hline
 $ E$         & 0   & 0             & $-q$ & 0 & $q^\perp$   & 0          & $-\nabla$ & $\nabla^\perp$ \\
 $\mathcal{E}$&   & 0               & 0 & $Q$ & 0            & $-Q^\perp$  &$-\nabla$ & $\nabla^\perp$ \\
  $q$         &     &              &  0 &  0 &  $\mathcal{E}$          &   $\nabla$  & 0        &   $Q$\\
 $ Q$         &     &         &   & 0  &      $\nabla^\perp$          & $E$      & $-q$  &  0 \\
  $q^\perp$& &            &    &      &  0 & 0  & $-Q^\perp$  &  0 \\
  $Q^\perp$ &       &   &      &    &   &  0 &  0 & $q^\perp$  \\
 $\nabla$   &  &    &    &    &   &   &  0 &  $E+\mathcal{E}$ \\
$\nabla^\perp$& &    &   &   &  &     &    &  0 \\
\end{tabular}
\end{center}
Note in particular that all commutation relations can be obtained by applying the rule $(AB)^\perp=B^\perp A^\perp$ and the super-Jacobi identities to the following six basic anti-commutations:
\begin{align}\label{4gene}
 &\{q,q\}=0, \quad \{q,Q\}=0,\quad \{Q,Q\}=0,\nonumber\\
&\{q,q^\perp\}=\mathcal{E},\quad \{q,Q^\perp\}=\nabla, \quad \{Q,Q^\perp\}=E.\end{align}
The relevance of this simple observation is that the algebra $\mathcal A$ is
generated by its fermionic elements: $q,\, q^\perp,\, Q$ and $Q^\perp$. 

There is a another natural Lie superalgebra of symmetric operators in $x$ and $\theta$.  Let us define the following first order (but non-homogeneous) linear operators: 
\begin{align} \label{virasoro}
&L_{n}=\sum_{i=1}^Nx_i^{-n}\left(x_i\partial_{x_i}+\frac{1-n}{2}\theta_i\partial_{\theta_i}\right),
\qquad n\leq 1, \nonumber\\
&G_{r}=\sum_{i=1}^N x_i^{-r+1/2}\left(\partial_{\theta_i}+\theta_i\partial_{x_i}\right),\qquad r\leq 1/2,
\end{align}
where it is understood that $n$ is an integer and $r$ is a half-integer.  These operators generate the negative half of the super-Virasoro algebra,
without central charge, which we will denote ${{\rm sVir}}^{(-)}$.
 \begin{equation}\label{sVir}
[L_n,L_m]=(n-m)L_{n+m},\quad [L_n,G_r]=\left(\frac{n}{2}-r\right)G_{n+r},\quad \{G_r,G_s\}=2L_{r+s}.\end{equation}
Obviously, the operators $L_n$ and $G_r$ preserve ${\bf \Lambda}^{\theta}$
for all $n\leq 1$ and $r\leq 1/2$.   The operators $G_{\pm 1/2},L_0,L_{\pm 1}$ generate a five-dimensional subalgebra isomorphic to $osp(1,2)$.
The relation between the latter operators and the homogeneous ones 
introduced in \eqref{homo} is simply:
\begin{align}\label{global}
&L_{1}=\nabla,\qquad
L_{-1}=\nabla^\perp -\frac{N  p_1}{\alpha},\qquad
L_{0}=\frac{1}{2}\left(\mathcal{E}+E\right)-\frac{N^2}{2\alpha},\nonumber\\
&G_{1/2}=q+Q^\perp,\qquad
G_{-1/2}=Q+q^\perp-\frac{N\tilde p_0}{\alpha} \,,
\end{align}
where we recall that $\tilde p_0$ and $p_1$ were defined in \eqref{spower}.
In view of relations \eqref{4gene}, these operators are almost completely generated by $q,q^\perp,\,Q$ and $Q^\perp$.
The missing pieces are $\tilde p_0$ in $G_{-1/2}$ and $p_1$ in $L_{-1}$. But since $L_{-1}=G_{-1/2}^2$, it suffices to add $\tilde p_0$ to $q,q^\perp,\,Q$ and $Q^\perp$ in order to generate $osp(1,2)$.  We will see in the next section
that, remarkably, the action of those operators on Jack superpolynomials
can be given explicitly.

\section{The ideal $\mathcal I_N^{(k,r)}$}\label{ideal}

Consider the subspace 
\begin{equation}
\mathcal I_N^{(k,r)} = {\rm span}_{\mathbb C} \bigl\{ 
P^{(\ad)}_{\Lambda}(x_1,\dots,x_N,\theta_1,\dots,\theta_N) \, \big| \, 
\Lambda~{\rm is}~(k,r,N){\text-}{\rm admissible}
\bigl\}.
\end{equation}
We will show that $\mathcal I_N^{(k,r)}$ is an ideal of
${\bf \Lambda}_N^\ta$ that is also stable under the action of the
algebras introduced in the previous section.

We first give the explicit action 
of $\tilde p_0 (=\sum_{i=1}^N\ta_i),\,q,\,q^\perp,\,Q$ and $Q^\perp$ on $P_\La$. 
In all cases, the expansion coefficients are expressed in terms of specific ratios of the hook-lengths  introduced  in \eqref{2hook}.
\begin{proposition}  \label{prop9}
We have 
\begin{align}
\tilde p_0 \, P_{\Lambda}^{(\aa)} &= \sum_{\Omega} (-1)^{\# \Omega^\circ} 
\left( \prod_{s \in {\rm col}_{\Omega^\circ}} 
\frac{h_{\Lambda}^{(\alpha)}(s)}{h_{\Omega}^{(\alpha)}(s)}  \right) P_{\Omega}^{(\aa)}
\label{forp0}\\
Q \, P_{\Lambda}^{(\aa)} & = \sum_{\Omega} (-1)^{\# \Omega^\circ} 
\left( \prod_{s \in {\rm col}_{\Omega^\circ}} 
\frac{h_{\Lambda}^{(\alpha)}(s)}{h_{\Omega}^{(\alpha)}(s)}  \right) 
\frac{(N+ 1- i +\alpha(j-1))}{\alpha}  \, P_{\Omega}^{(\aa)} \label{forQ}\\
Q^{\perp} P_{\Lambda}^{(\aa)} & = \sum_{\Omega} (-1)^{\# \Omega^\circ} 
\left( \prod_{s \in {\rm row}_{\Omega^\circ}} 
\frac{h^{\Omega}_{(\alpha)}(s)}{h^{\Lambda}_{(\alpha)}(s)}  \right) 
\bigl(N+1-i+\alpha(j-1)\bigr)  \, P_{\Omega}^{(\aa)} \label{forQperp}\\
q \,  P_{\Lambda}^{(\aa)} & = \sum_{\Omega} (-1)^{\# \Omega^\circ} 
\left( \prod_{s \in {\rm row}_{\Omega^\circ}} 
\frac{h_{\Lambda}^{(\alpha)}(s)}{h_{\Omega}^{(\alpha)}(s)}  \right) 
  \, P_{\Omega}^{(\aa)} \label{forq}\\
q^{\perp} P_{\Lambda}^{(\aa)} & = \sum_{\Omega} (-1)^{\# \Omega^\circ} 
\left( \prod_{s \in {\rm col}_{\Omega^\circ}} 
\frac{h^{\Omega}_{(\alpha)}(s)}{h^{\Lambda}_{(\alpha)}(s)}  \right) 
  \, P_{\Omega}^{(\aa)}\, . \label{forqperp}
\end{align}
The sum is taken in \eqref{forp0} and \eqref{forQ} over
all $\Omega$'s obtained by adding a circle to
$\Lambda$, in \eqref{forQperp} over
all $\Omega$'s obtained by removing a circle from
$\Lambda$,  in \eqref{forq} over
all $\Omega$'s obtained by converting a square of $\Lambda$ into
a circle, and in \eqref{forqperp} over
all $\Omega$'s obtained by converting a circle of $\Lambda$ into
a square.  Observe that in each of those cases, 
$\Lambda$ and $\Omega$ differ in exactly one cell
which we call the \underline{marked} cell.

\noindent The symbol $\# \Omega^\circ$ stands for the number
of circles in $\Omega$ above the marked cell.

\noindent 
The symbol ${\rm col}_{\Omega^\circ}$ stands for
the column of $\Omega$ and $\Lambda$ 
above the marked cell, while
${\rm row}_{\Omega^\circ}$ stands for
the row of $\Omega$ and $\Lambda$ 
to the left of the marked cell.

\noindent Finally, in \eqref{forQ} and \eqref{forQperp}, $(i,j)$ is the position
of the marked cell.
\end{proposition}
\begin{proof} When $N$ is large enough with respect to the degree
of $\Lambda$, the scalar product 
\eqref{scap} is well-defined.  In this case, $q$ and $q^\perp$ (resp.
$Q$ and $Q^\perp$) are adjoint of each others, and \eqref{forQperp} can
be obtained from \eqref{forQ} (resp. 
\eqref{forqperp} can
be obtained from \eqref{forq}).  Note that when a formula is established
for large $N$, it also holds for all $N$ due to the stability of the Jack
superpolynomials with respect to restriction of variables.

It thus suffices to prove \eqref{forp0}, \eqref{forQ} and \eqref{forq}.
We will only prove \eqref{forp0} and \eqref{forq}, as \eqref{forQ} 
can be proven with similar methods as those used to prove \eqref{forq}.
Given that the proofs are quite long and rely on intricate properties
of Jack superpolynomials, they will be relegated to  
Appendix~\ref{apppieri}.
\end{proof}

Here is an example illustrating formula \eqref{forp0}. 
Identifying $P_\La^{(\aa)}$ with the diagram of $\La$, we have
\begin{equation} 
\tilde p_0\times {\tableau[scY]{&\\&\\&\bl\tcercle{}}}\, =
(-1)^0\times 1\times {\tableau[scY]{&&\bl\tcercle{}\\&\\&\bl\tcercle{}}}\, 
+(-1)^1\times\frac{(2+2\a)(1+2\a)\a}{(3+2\a)(2+2\a)(1+\a)} \times{\tableau[scY]{&\\&\\&\bl\tcercle{}\\\bl\tcercle{}}}\, 
\end{equation}

We can now show that $p_0, Q,Q^\perp,q$ and $q^\perp$ preserve
$\mathcal I_N^{(k,r)}$.
\begin{proposition} \label{prop14} For 
$\upsilon \in\{\tilde p_0, Q,Q^\perp,q,q^\perp\}$, we  have
$\upsilon \, \mathcal I_N^{(k,r)} \subseteq \mathcal I_N^{(k,r)}$.
\end{proposition}
\begin{proof}
Let $\Lambda$ be a $(k,r,N)$-admissible superpartition.
We will prove that the expression for
$\upsilon P_{\Lambda}^{(\ad)}$ in Proposition~\ref{prop9}   
still holds at $\alpha = \ad$, and that, moreover,
only Jack superpolynomials indexed by $(k,r,N)$-admissible 
superpartitions remain in the expression.
Since,
the $P_\Omega^{(\ad)}$'s appearing in the expansion 
of $\upsilon P_{\Lambda}^{(\ad)}$ in Proposition~\ref{prop9}
are such that $\Omega$
is almost $(k,r,N)$-admissible,
we have from Proposition~\ref{propquasi} that the $P_{\Omega}$'s
are regular at $\alpha = \ad$.  We have thus left to prove that the
expansion coefficients are free of poles at $\alpha=\ad$
and vanish whenever $\Omega$
is not $(k,r,N)$-admissible.

\noindent {\it Case $\upsilon = \tilde p_0$}.  
If $\Lambda$ is $(k,r,N)$-admissible, then
all the superpartitions $\Omega$ that appear in the sum
in \eqref{forp0} are also $(k,r,N)$-admissible (adding
a circle to a $(k,r,N)$-admissible superpartition produces
a $(k,r,N)$-admissible superpartition).  We thus only have to show
that $h_{\Omega}^{(\alpha)}(s)$ does not have zeros for any
$s \in {\rm col}_{\Omega^\circ}$ when $\alpha=\ad$.

Let $j$ be the row of the marked cell 
(the cell of the circle that was added to $\Lambda$ to obtain 
$\Omega$).  We have that 
\begin{equation}
\prod_{s \in {\rm col}_{\Omega^\circ}} h_{\Omega}^{(\alpha)}(s) = 
\prod_{i<j} \bigl(j -i + \ad (\Lambda_i^*-\Lambda_j^*)\bigr).
\end{equation}
Hence we need to show that
\begin{equation}
j-i -\frac{k+1}{r-1} (\Lambda_i^*-\Lambda_j^*) \neq 0 ,
\end{equation}
for all $i<j$.  Suppose on the contrary that $j-i=(k+1)t$ and $\Lambda_i^*-\Lambda_j^*=(r-1)t$ for some positive integer $t$.
But, using Lemma~\ref{lemadm}, this leads to 
the contradiction
\begin{equation}
(r-1)t= \Lambda_i^*-\Lambda_j^* \geq \left \lfloor\frac{j-i}{k+1} \right 
\rfloor r
= \left \lfloor \frac{(k+1)t}{k+1} \right \rfloor r =rt.
\end{equation}

\noindent {\it Case $\upsilon = Q$}.
Using \eqref{forQ} instead of
\eqref{forp0}, the case is exactly as Case $\upsilon =\tilde p_0$.

\noindent {\it Case $\upsilon = Q^\perp$}.
We use \eqref{forQperp}.
In this case, $h^{\Lambda}_{(\alpha)}(s)$ does not
have zeros at $\ad$ by the proof of Proposition~\ref{lemmapoles}.
It thus suffices to prove that
 if $\Omega$ is not $(k,r,N)$-admissible, then 
\begin{equation}
\prod_{s \in {\rm row}_{\Omega^\circ}} h^{\Omega}_{(\ad)}(s) = 0,
\end{equation}
that is, there is at least one square $s \in {\rm row}_{\Omega^\circ}$ such that 
$h^{\Omega}_{(\ad)}(s) = 0$. 
Since $\Omega$ is obtained by removing a circle to $\Lambda$, 
$\Omega$ is not $(k,r,N)$-admissible only
if the circle is removed in a certain row $i$ such that
$\Lambda_i^\cd-\Lambda_{i+k}^*=r$.  
Suppose that $\Lambda_{i+k} = 0$ and
$N=i+k$.
Then the position of the marked cell is $(i,j)=(i,r)$
and the factor $N+1-i+\alpha(j-1)$ in  \eqref{forQperp}
is equal to zero.   
Otherwise, we have either $\Lambda_{i+k} > 0$ or
$N>i+k$.  If $N>i+k$, then 
$\Lambda_{i+k+1}^*<\Lambda_{i+k}^*$ by Lemma~\ref{lemadm} since
$\Lambda_i^*-\Lambda^*_{i+k}=r-1$.  If we let 
$j=\Lambda_{i+k}^*$, the leg-length of $s=(i,j)$ in $\Omega^*$
is thus equal to $k$. This gives
\begin{equation}
h^{\Omega}_{(\ad)}(i,j)=  k+1 -\frac{(k+1)}{(r-1)}(\Omega_i^\cd-j)=
k+1 -\frac{(k+1)}{(r-1)}(\Lambda_i^*-\Lambda^*_{i+k}) =
k+1 -\frac{(k+1)}{(r-1)}(r-1)=0\, .
\end{equation}
Finally, if $N=i+k$ and $j=\Lambda_{i+k}^*>0$
then the leg-length of $s=(i,j)$ in $\Omega^*$ is again equal to $k$.
The rest of the previous argument can thus be used and 
our claim holds.

\noindent {\it Case $\upsilon = q$}.
As in Case $\upsilon = \tilde p_0$, 
if $\Lambda$ is $(k,r,N)$-admissible, then
all the superpartitions $\Omega$ that appear in the sum
in \eqref{forq} are also $(k,r,N)$-admissible (converting
a square into a circle in
a $(k,r,N)$-admissible superpartition produces
a $(k,r,N)$-admissible superpartition).  We thus only have to show
that $h_{\Omega}^{(\alpha)}(s)$ does not have zeros for any
$s \in {\rm row}_{\Omega^\circ}$ when $\alpha=\ad$.

Let $i$ be the row of the marked cell.
Then
\begin{equation}
\prod_{s \in {\rm row}_{\Omega^\circ}} h_{\Omega}^{(\alpha)}(s) = 
\prod_{j<\Omega_i^{\cd}} 
\bigl({\Omega^\cd_j}'-i + \ad (\Omega_i^{*}-j+1)\bigr)
=\prod_{j<\Omega_i^{\cd}} 
\bigl({\Omega^\cd_j}'-i + \ad (\Omega_i^{\cd}-j)\bigr).
\end{equation}
Hence we need to show that
\begin{equation}
{\Omega^\cd_j}'-i + \ad (\Omega_i^{\cd}-j) \neq 0,
\end{equation}
for all $j<\Omega_i^\cd$.   Suppose  on the contrary  that
\begin{equation}
\Omega_i^\cd-j= t(r-1)\quad \text{and} \quad {\Omega_j^\cd}'-i= t(k+1)\quad\text{ for some positive integer $t$}. 
\end{equation}
Using $\Omega^\cd_{{\Omega^\cd}'_j} \geq j$ (see 
(\ref{ine2})) and the second inequality in \eqref{2ine},
we are then led to the contradiction
\begin{equation}
t(r-1)= \Omega^\cd_i-j \geq \Omega^\cd_i-\Omega^\cd_{{\Omega^\cd}'_j} \geq \left\lfloor \frac{{\Omega^\cd}'_j-i}{k+1}\right \rfloor r = \left \lfloor \frac{t(k+1)}{k+1}\right \rfloor r= tr ,
\end{equation}
and the claim holds.

\noindent {\it Case $\upsilon = q^\perp$}.
We use \eqref{forqperp}.  As in the proof of Case $\upsilon = Q^\perp$, 
$h^{\Lambda}_{(\alpha)}(s)$ does not
have zeros at $\ad$ by the proof of {Proposition~\ref{lemmapoles}}.
It thus
suffices to prove that
if $\Omega$ is not $(k,r,N)$-admissible, then 
\begin{equation}
\prod_{s \in {\rm col}_{\Omega^\circ}} h^{\Omega}_{(\ad)}(s) = 0.
\end{equation}
In this case $\Omega$ is obtained by converting a circle of $\Lambda$
into a square.  If the circle converted into a square is in row $i$ and
$\Omega$ is not $(k,r,N)$-admissible then $\Lambda^\cd_{i-k}-\Lambda_i^*=r$
and $\Lambda^\cd_{i-k}-\Lambda_i^\cd=r-1$.
If we let $j=\Lambda_i^\cd=\Omega_i^*$,
the leg-length of $s=(i-k,j)$ in $\Omega^*$
is thus equal to $k$ (since $\Omega_{i+1}^*=\La_{i+1}^* \leq \La_{i}^* < 
\La_{i}^\cd=\Omega_i^*$).  Therefore, we observe immediately that
\begin{equation}
h^{\Omega}_{(\ad)}(i-k,j)= k+1 -\frac{(k+1)}{(r-1)}(\Omega_{i-k}^\cd-j) =
 k+1 -\frac{(k+1)}{(r-1)}(\Lambda_{i-k}^\cd-\Lambda_i^\cd) =
k+1 -\frac{(k+1)}{(r-1)}(r-1)=0,
\end{equation} 
and our claim is verified.
\end{proof}

We can now prove our main result.
Recall that the super Lie algebras $\mathcal A$ and ${\rm sVir}^{(-)}$
were defined in \eqref{homo} and \eqref{virasoro} respectively.
\begin{theorem} \label{theoideal}   $\mathcal I_N^{(k,r)}$ is an ideal of
${\bf \Lambda}_N^\ta$.  Furthermore, the algebra $\mathcal A$ 
and the negative half of the super-Virasoro algebra ${\rm sVir}^{(-)}$
preserve $\mathcal I_N^{(k,r)}$.
\end{theorem}
\begin{proof}
To prove the first statement, we use the fact that
${\bf \Lambda}_N^\ta$
is generated by the power sums $p_m$ and $\tilde p_n$ for
$m \geq 1$ and $n \geq 0$ (see \eqref{spower}).  
{From} Proposition~\ref{prop14}, we
have that $p_1 \,  \mathcal I_N^{(k,r)} \subseteq \mathcal I_N^{(k,r)}$
given that
$\{\tilde p_0, q^\perp \}=p_1$.  Since $[\nabla^\perp,p_n] = n p_{n+1}$
and $\nabla^\perp=\{q^\perp,Q\}$ we then get recursively by 
Proposition~\ref{prop14}
that 
$p_m \,  \mathcal I_N^{(k,r)} \subseteq \mathcal I_N^{(k,r)}$ 
for all $m \geq 1$.  Finally, since
$[q, p_n]=n \tilde p_{n-1}$, we get from Proposition~\ref{prop14}
that $\tilde p_n \,  
\mathcal I_N^{(k,r)} \subseteq \mathcal I_N^{(k,r)}$ for all $n \geq 0$ and
the first statement follows.

Since $\mathcal A$ is generated by $Q,Q^\perp,q$ and $q^\perp$, it preserves
$\mathcal I_N^{(k,r)}$ by  Proposition~\ref{prop14}.

As already mentioned, 
${\rm sVir}^{(-)}$
is generated by $\tilde p_0, Q,Q^\perp,q,q^\perp$ and $L_{-2}$.
By Proposition~\ref{prop14}, it thus remains to show that 
$L_{-2}\mathcal I_N^{(k,r)} \subseteq \mathcal I_N^{(k,r)}$.  We have the relation
\begin{equation}
L_{-2} = \frac{3}{4\ad}[\Delta^{(\ad)},p_2]
 +\frac{1}{2\alpha}[D^{(\ad)},p_2]-\frac{1}{2}p_2 -\frac{1}{\ad}(p_1^2-p_2) \,,
\end{equation}
where $\Delta^{(\ad)}$ and $D^{(\ad)}$ correspond respectively
to the operators
$\Delta$ and $D$ at $\alpha=\ad$ (see \eqref{eqD} and \eqref{eqDelta}).
Obviously, $\Delta^{(\ad)} \, \mathcal I_N^{(k,r)} \subseteq \mathcal I_N^{(k,r)}$ and  $D^{(\ad)} \, \mathcal I_N^{(k,r)} \subseteq \mathcal I_N^{(k,r)}$ since
they both have the Jack superpolynomials 
as eigenfunctions. Finally, as we have just shown, 
$p_1$ and $p_2$ leave   $\mathcal I_N^{(k,r)}$ invariant
and the result follows.
\end{proof}

The theorem has an immediate consequence that will prove important in the
next section.  The proof is similar to that of an analogous property in 
\cite{FJMM}.
\begin{proposition} \label{proponm1}
Let $P \in \mathcal I_{N}^{(k,r)}$.  Then
\begin{equation}
[\partial_{x_N}^j P]_{x_N=\theta_N=0} \in   \mathcal I_{N-1}^{(k,r)}
\quad {\rm and} \quad [\partial_{\theta_N} \partial_{x_N}^j  P]_{x_N=\theta_N=0} \in   \mathcal I_{N-1}^{(k,r)} \quad {\rm for~all~}j \geq 0
\end{equation}
\end{proposition}
\begin{proof} We first prove by induction that
$[\partial_{x_N}^j P]_{x_N=\theta_N=0} \in   \mathcal I_{N-1}^{(k,r)}$.
Recall that for any Jack superpolynomial $P_{\Lambda}^{(\aa)}$ in $N$ variables
we have that $[P_{\Lambda}^{(\aa)}]_{x_N=\theta_N=0}$ is either equal to zero or to
$P_{\Lambda}^{(\aa)}$ in
the variables $x_1,\dots,x_{N-1},\theta_1,\dots,\theta_{N-1}$. 
Since the Jack superpolynomials indexed by $(k,r,N)$-admissible
superpartitions form a basis of $\mathcal I_{N}^{(k,r)}$, 
we then have 
immediately that $[P]_{x_N=\theta_N=0} \in   \mathcal I_{N-1}^{(k,r)}$
if $P \in \mathcal I_{N}^{(k,r)}$.  Let  $\nabla^{(N)}= \sum_{i=1}^N \partial_{x_i}$
be the operator $\nabla \in \mathcal A$, with the number  of variables
$N$ made explicit.
The result then holds by induction from Theorem~\ref{theoideal}
using
\begin{equation}
[\partial_{x_N}^j P]_{x_N=\theta_N=0}= 
[\partial_{x_N}^{j-1} \nabla^{(N)} P]_{x_N=\theta_N=0} -\nabla^{(N-1)} [\partial_{x_N}^{j-1} P]_{x_N=\theta_N=0} \, .
\end{equation}

We now prove by induction that
$[\partial_{\theta_N} 
\partial_{x_N}^j P]_{x_N=\theta_N=0} \in   \mathcal I_{N-1}^{(k,r)}$.
Again, let $Q^\perp_{(N)} = \sum_{i=1}^N \partial_{\theta_i}$ be the operator
$Q^\perp \in \mathcal A$ with the number $N$
of variables made explicit.  We have that
$[\partial_{\theta_N} P]_{x_N=\theta_N=0} \in   \mathcal I_{N-1}^{(k,r)}$
from Theorem~\ref{theoideal} since
\begin{equation}
[\partial_{\theta_N} P]_{x_N=\theta_N=0}= 
[Q^\perp_{(N)} P]_{x_N=\theta_N=0} - Q^\perp_{(N-1)} [ P]_{x_N=\theta_N=0} \, .
\end{equation}
Finally, the general case follows again from Theorem~\ref{theoideal} since
\begin{equation}
[\partial_{\theta_N}\partial_{x_N}^j P]_{x_N=\theta_N=0}= 
[\partial_{\theta_N}\partial_{x_N}^{j-1} \nabla^{(N)} P]_{x_N=\theta_N=0} -\nabla^{(N-1)} [\partial_{\theta_N} \partial_{x_N}^{j-1} P]_{x_N=\theta_N=0} \, .
\end{equation}
\end{proof}

\section{Clustering properties of the Jack polynomials}
\label{les0}

Consider the subspace of symmetric superpolynomials 
in $N$ variables that vanish   whenever $k+1$ {\it commuting} 
variables are equal:
\begin{equation}   
{\mathcal F}_N^{(k)} = \bigl\{ f(x,\theta) 
\in {\bf \Lambda}_N^{\theta} \, 
\big| \,  f(x,\theta)=0 {\rm ~if~}
x_1=\cdots=x_{k+1} \bigr\} \, .
\end{equation}

We first show that ${\mathcal I}_N^{(k,r)} \subseteq {\mathcal F}_N^{(k)}$ 
for all $r$ {(as already pointed out in the introduction, following relation
\eqref{jackinnonsym}
this could also have been deduced from \cite{Kas})}.
\begin{proposition}\label{Propvanishwithout}
  Suppose that $N \geq k+1$.  Then
${\mathcal I}_N^{(k,r)} \subseteq {\mathcal F}_N^{(k)}$ for all $r$.
\end{proposition}
\begin{proof}
Let $\La$ be a $(k,r,N)$-admissible superpartition
of fermionic degree $m$.  We will show that
$P_{\La}^{(\ad)}=0$ if $x_1=\cdots=x_{k+1}$.  

Suppose first that $N=k+1$.  In this case, we set 
$x_1=\cdots=x_N$ in $P_{\La}^{(\ad)}$.  If $m>1$, we can factorize
a Vandermonde determinant $\Delta(x_{i_1},\dots,x_{i_m})$ from
the coefficient
$f_{i_1,\dots,i_m}(x_1,\dots,x_N)$ of $\theta_{i_1} \cdots \theta_{i_m}$ 
in $P_{\La}^{(\ad)}$.  Therefore, $f_{i_1,\dots,i_m}(x_1,\dots,x_N)$
vanishes if $x_1=\cdots=x_N$, and so does $P_{\La}^{(\ad)}$.
We thus have left to consider the cases $m=0,1$.
Let $\P^{(\ad)}_{\La,m}$ be such as in \eqref{defPs}.  By symmetry, it
suffices to prove that $\P^{(\ad)}_{\La,m}=0$ if $x_1=\cdots=x_N=x$.
 By the homogeneity of $\P^{(\ad)}_{\La,m}$,
we have
\begin{equation}
\P^{(\ad)}_{\La,m}(x,\dots,x)=x^{t} \P^{(\ad)}_{\La,m}(1,\dots,1)
\end{equation}
for some power $t$.
Hence we can use (\ref{spe}) to determine whether $\P^{(\ad)}_{\La,m}$ vanishes.
The admissibility condition when $N=k+1$ reduces to the 
single requirement: 
\begin{equation}
\La^\cd_1\geq \La^*_{N}+r.
\end{equation} 
The vanishing of the RHS of (\ref{spe}) can only  be due to that of $b_{\La}^{(\ad,k+1)}(1,r)$ 
(it is clear by inspection that $b_\La^{(\ad,N+1)}(s)$ 
cannot be zero for any other square $s$).
This square belongs to $\La^\cd$ since $\La^\cd_1\geq r$. But in addition, it must necessarily be part of $\S\La$: only the first square of the first row of $\La^\cd$ possibly does not belong to
$\S\La$ when $m=0,1$. When this is satisfied, we indeed see that
\begin{equation}
b_{\La}^{(\ad,k+1)}(1,r)=k+1-\frac{k+1}{r-1}(r-1)=0.\end{equation}

Now suppose by induction that 
the result holds for $N-1$, with $N \geq k+2$.
We have that 
\begin{equation}
P_{\Lambda}^{(\ad)} = \sum_{i \geq 0} x_N^i {\tilde f}^{(i)} + \sum_{i \geq 0} \theta_N x_N^i {\tilde g}^{(i)}
\end{equation}
where ${\tilde f}^{(i)}$ and ${\tilde g}^{(i)}$ are symmetric superpolynomials in 
${\bf \La}_{N-1}^{\theta}$ of fermionic degree $m$ and $m-1$ respectively.  Observe that
${\tilde f}^{(i)}$ and ${\tilde g}^{(i)}$ belong to ${\mathcal I}_{N-1}^{(k,r)}$
by Proposition~\ref{proponm1}, and the result follows by induction.
\end{proof}
We conjecture that when $r=2$ the inclusion is actually an equality.

\begin{conjecture} \label{I=F}We have ${\mathcal I}_N^{(k,2)} = {\mathcal F}_N^{(k)}$.
In other words, the Jack superpolynomials in $N$ variables whose superpartitions are $(k,2,N)$-admissible furnish a basis of the space
$\mathcal F_N^{(k)}$.
\end{conjecture}

Thanks to Proposition \ref{Propvanishwithout}, in order to prove ${\mathcal I}_N^{(k,2)} = {\mathcal F}_N^{(k)}$,  it is sufficient to prove the equality 
of the two following characters:   
\begin{equation}\label{cara}
 \ch {\mathcal I}_N^{(k,2)} (u,v):=\sum_{n,m\geq 0} \dim {\mathcal I}_{N,n,m}^{(k,2)} u^n v^m,\qquad  \ch {\mathcal F}_N^{(k)}(u,v) :=\sum_{n,m\geq 0} \dim {\mathcal F}_{N,n,m}^{(k)} u^n v^m,
\end{equation} 
where the sub-indices $n,m$ refer to the bi-homogeneous component of 
degree $(n|m)$ (total degree $n$ and fermionic degree $m$).
Obviously, $\dim {\mathcal I}_{N,n,m}^{(k,2)}$ is equal to the number of $(k,2,N)$-admissible superpartitions,  and thus
$\ch {\mathcal I}_N^{(k,2)} (u,v)$ can
be easily evaluated degree by degree.
As explained in Appendix \ref{Characters}, the series expansion of $\ch {\mathcal F}_N^{(k)}$ is also computable degree by degree  (a few examples of characters, up to degree 10, are presented 
in Appendix  \ref{Characters}).  
The equality of the two characters has been verified up to degree 12, for $k+1\leq N\leq 5$ and $1\leq k\leq 4$.
 This is the computational evidence for Conjecture \ref{I=F}.

 As mentioned in the introduction,  this conjecture 
not only reveals a  remarkable property of the Jack superpolynomials 
 but appears to be
a clear indication of potential applications in superconformal 
field theory, the admissibility condition capturing a new variant of the generalized exclusion relation.


When the fermionic degree $m$ of the $(k,r,N)$-admissible 
superpartition $\Lambda$ is smaller than $r$, we actually have
a stronger result than ${\mathcal I}_N^{(k,r)} \subseteq {\mathcal F}_N^{(k)}$.
\begin{proposition}\label{sJ0} Suppose that $N \geq k+1$ and $r>m$, where
$m$ is the fermionic degree of the $(k,r,N)$-admissible superpartition
$\La$.  Then
$\P^{(\ad)}_{\La,m}$ vanishes whenever $k+1$ of the variables $\{x_1,\dots,x_N\}$
are equal,  where we recall that $\P^{(\ad)}_{\La,m}$ was defined in 
\eqref{defPs}.
\end{proposition} 
\begin{proof}
The proof basically follows the steps of the proof of Proposition~\ref{Propvanishwithout}.  But the
loss of the overall symmetry and the division by the Vandermonde determinant
make the arguments somehow more involved.

We first consider the case $N=k+1$ and show that
$\P^{(\ad)}_{\La,m}$ vanishes when $x_1=\cdots=x_{k+1}=x$
if and only if $r>m$.  As in the proof of 
Proposition~\ref{Propvanishwithout}, we use (\ref{spe}) to determine whether $\P^{(\ad)}_{\La,m}$ vanishes.  Again the admissibility condition when $N=k+1$ reduces to the 
single requirement 
\begin{equation}
\La^\cd_1\geq \La^*_{N}+r, 
\end{equation} 
and the vanishing of the RHS of (\ref{spe}) can only  be due to that of $b_{\La}^{(\ad,k+1)}(1,r)$.  Since $s=(1,r)$
must necessarily be part of $\S\La$, this forces $r>m$ (since the first $m$ squares of the first row of $\La^\cd$ are not part of $\S\La$). When this is satisfied, we indeed see that
\begin{equation}
b_{\La}^{(\ad,k+1)}(1,r)=k+1-\frac{k+1}{r-1}(r-1)=0.\end{equation}

Now suppose by induction that 
the result holds for $N-1$, with $N \geq k+2$.
We have that 
\begin{equation}
P_{\Lambda}^{(\ad)} = \sum_{i \geq 0} x_1^i {\tilde f}^{(i)} + \sum_{i \geq 0} \theta_1 x_1^i {\tilde g}^{(i)}
\end{equation}
where ${\tilde f}^{(i)}$ and ${\tilde g}^{(i)}$ are symmetric superpolynomials in the variables $x_2,\dots,x_N,\theta_2,\dots,\theta_N$
of fermionic degree $m$ and $m-1$ respectively.  Observe that
${\tilde f}^{(i)}$ and ${\tilde g}^{(i)}$ belong to ${\mathcal I}_{N-1}^{(k,r)}$
by Proposition~\ref{proponm1} (the proposition is used with the variables
$x_1$ and $\theta_1$ instead of $x_N$ and $\theta_N$, and with the 
variables of ${\mathcal I}_{N-1}^{(k,r)}$ taken to be
$x_2,\dots,x_N,\theta_2,\dots,\theta_N$).
Therefore
\begin{equation}
\P^{(\ad)}_{\La,m} = \sum_{i \geq 0} \frac{x_1^{i}}{(x_1-x_2)\cdots (x_1-x_m)} 
\frac{{\tilde g}^{(i)}|_{\theta_2 \cdots \theta_{m}}}{\Delta(x_2,\dots,x_{m})}
\end{equation}
vanishes 
by induction 
 whenever $k+1$ of the
variables $\{x_2,\dots,x_{N}\}$ are equal (${\tilde g}^{(i)}$ is a sum of
Jack superpolynomials of fermionic degree $m-1<r$
indexed by $(k,r,N-1)$-admissible superpartitions).  
If $m=N$ the proof is over since
$\P^{(\ad)}_{\La,m}$ is symmetric in the variables $x_1,\dots,x_m$
and thus vanishes 
 whenever $k+1$ of the
variables $\{x_1,\dots,x_{m}\}$ are equal.  
If $m<N$ we use
\begin{equation}
P_{\Lambda}^{(\ad)} = \sum_{i \geq 0} x_N^i f^{(i)} + \sum_{i \geq 0} \theta_N x_N^i g^{(i)}
\end{equation}
where $f^{(i)}$ and $g^{(i)}$ are symmetric superpolynomials 
of fermionic degree $m$ and $m-1$ respectively
in the variables
$x_1,\dots,x_{N-1},\theta_1,\dots,\theta_{N-1}$.  Again
$f^{(i)}$ and $g^{(i)}$ belong to ${\mathcal I}_{N-1}^{(k,r)}$
by Proposition~\ref{proponm1}.  Hence
\begin{equation}
\P^{(\ad)}_{\La,m} = \sum_{i \geq 0} x_N^{i} \frac{f^{(i)}|_{\theta_1 \cdots \theta_m}}{\Delta(x_1,\dots,x_m)}
\end{equation}
vanishes  by induction whenever $k+1$ of the
variables $\{x_1,\dots,x_{N-1}\}$ are equal ($f^{(i)}$ is a sum of
Jack superpolynomials of fermionic degree $m<r$
indexed by $(k,r,N-1)$-admissible superpartitions).  
We have thus shown that $\P^{(\ad)}_{\La,m}$ vanishes
whenever $k+1$ of the
variables $\{x_1,\dots,x_{N-1}\}$ are equal or
whenever $k+1$ of the
variables $\{x_2,\dots,x_{N}\}$ are equal.  But
$\P^{(\ad)}_{\La,m}$ is symmetric in the variables $x_1,\dots,x_m$ and in 
the variables $x_{m+1},\dots,x_N$, which implies
that
$\P^{(\ad)}_{\La,m}$ vanishes  whenever $k+1$ of the
variables $\{x_1,\dots,x_N\}$ are equal.
\end{proof}

For instance, taking $k=1$,
  \begin{align}\P^{(-2)}_{\La,m}(x,x)&\ne 0\quad {\rm if} \quad
\Lambda={{\tiny \tableau[scY]{&&&&\bl\tcercle{}\\&&\bl\tcercle{}}}}
\qquad \text{ since $r=m=2$}\nonumber\\
  \P^{(-2/3)}_{\La,m}(x,x)&=0\quad {\rm if} \quad
\Lambda={{\tiny \tableau[scY]{&&&&&\bl\tcercle{}\\&&\bl\tcercle{}}}}
\quad \, \, \text{ since $r=4$ and $m=2$.}
  \end{align}
This last  example illustrates the fact that a purely antisymmetric superpartition $\La$
can lead to 
the vanishing of $\P_{\La,m}^{(\ad)}$.

We finally indicate the form of
the clustering property of the Jack superpolynomials. 
Let $x_{i_1}=\cdots=x_{i_k}=x$ and let $x_{i_{k+1}}=x'$ be a variable that does 
not belong to $\{x_{i_1},\dots,x_{i_k} \}$. Let also 
$a$ be the number of elements in 
$\{x_{i_1},\dots,x_{i_k},x_{i_{k+1}} \} \cap \{x_1,\cdots, x_m\}$.  We conjecture
the following:
\begin{conjecture} \label{conjecrma}
 If $\La$ is $(k,r,N)$-admissible, then
\begin{equation}
 (x-x')^{r-a} \quad {\rm divides}  \quad
\P_{\La,m}^{(\ad)} \, .\end{equation}
If moreover $N \geq k+m+1$ and  $r>m > 0$, then the multiplicity of the factor $(x-x')$ is exactly equal to $r-a$.
\end{conjecture}
This conjecture has been heavily tested: it has been checked for all $(k,r,N)$-admissible superpartitions of fermionic degree $m\geq 1$, bosonic degree $n\leq 10$, and such that $k\leq 6$, $r\leq 6$, $N\leq 8$ (for a total of  17924 cases).  Among all cases, only 489 have multiplicities strictly greater than $r-a$. Of course, none of these exceptional cases also satisfies $N\geq k+m+1$ and $r>m$.  



Conjecture~\ref{conjecrma} gives the clustering property of the
Jack polynomials with prescribed symmetry $\P_{\La,m}^{(\ad)}$ when
$r>m$, since in this case $r$ is always larger than $a$. 
Conjecture~\ref{conjecrma} also readily implies the clustering property
of the Jack superpolynomials described in \eqref{newcluster} and
\eqref{newcluster2}.   Recall that
\begin{equation} \label{eqantisym}
P_{\Lambda}^{(\ad)}(x_1,\dots,x_{N}) 
\Big |_{\theta_1 \cdots \theta_m} = \Delta(x_1,\dots,x_m) \, \P_{\Lambda,m}^{(\ad)}(x_1,\dots,x_{N}) \, .
\end{equation}
If $N \geq k+m+1$, we then get that  
\begin{equation} \label{ex11}
P_{\Lambda}^{(\ad)}(x_1,\dots,x_{N-k-1},x',\overbrace{x,\ldots,x}^{k \text{ times}}) 
\Big |_{\theta_1 \cdots \theta_m} \quad  \text{vanishes as } \quad (x-x')^r 
\quad  \text{when} \quad x \to x'  \, ,
\end{equation}
from Conjecture~\ref{conjecrma} with $a=0$. In this case, the condition $r>m$ plays no role and can be relaxed.   As pointed out previously,  the inequality $N \geq m+k+1$ ensures that the sets
$\{1,...,m\}$ and $\{N-k,\dots,N\}$ do not intersect and that the equality $a=0$  holds. But in this situation Conjecture~\ref{conjecrma} states  that not only  $(x-x')^r$ divides the polynomial, 
but that the multiplicity of $(x-x')$
is precisely $r$.
This is the situation described in
\eqref{newcluster} and the discussion following it.

If two of the $k$ variables set equal to $x$ belong to $\{x_1,\dots,x_m\}$ 
we have that 
\begin{equation} 
\left[ P_{\Lambda}^{(\ad)}(x_1,\dots,x_{N}) 
\Big |_{\theta_1 \cdots \theta_m} \right]_{x_{i_1}=\cdots=x_{i_k}=x,x_{i_{k+1}}=x'} = 0 \, ,
\end{equation}
from the antisymmetry of the Vandermonde determinant in \eqref{eqantisym}.
If one of the $k+1$ variables $\{x_{i_1},\dots,x_{i_{k+1}} \}$
belongs to $\{x_1,\dots,x_m\}$, then Conjecture~\ref{conjecrma} with $a=1$ implies that
\begin{equation} 
\left[ P_{\Lambda}^{(\ad)}(x_1,\dots,x_{N}) 
\Big |_{\theta_1 \cdots \theta_m} \right]_{x_{i_1}=\cdots=x_{i_k}=x,x_{i_{k+1}}=x'} \quad { \text{vanishes as } }\quad
(x-x')^{r-1} \quad   \text{when} \quad x \to x' 
 \, .
\end{equation}
Finally, if 
one of the $k$ variables $\{x_{i_1},\dots,x_{i_{k}} \}$
belongs to $\{x_1,\dots,x_m\}$ and $x_{i_{k+1}}$ belongs to $\{x_1,\dots,x_m\}$, 
then Conjecture~\ref{conjecrma} with $a=2$ leads to
\begin{equation} 
\left[ P_{\Lambda}^{(\ad)}(x_1,\dots,x_{N}) 
\Big |_{\theta_1 \cdots \theta_m} \right]_{x_{i_1}=\cdots=x_{i_k}=x,x_{i_{k+1}}=x'}  \quad { \text{vanishes as } }\quad
(x-x')^{r-1}  \quad  \text{when} \quad x \to x' 
\, ,
\end{equation}
since there is a factor $(x-x')$ in the Vandermonde determinant.  The last
three equations immediately imply that
\begin{equation} 
P_{\Lambda}^{(\ad)}(x_1,\dots,x_{N-k-1},x',\overbrace{x,\ldots,x}^{k \text{ times}})
  \quad { \text{vanishes as } }\quad (x-x')^{r-1}  \quad  \text{when} \quad x \to x' 
\, ,
\end{equation}
which corresponds to \eqref{newcluster2}.

 We conclude with some  examples of the clustering property.  For simplicity,
we use the notation for superpartitions where 
$\La$ is written as $\La^*$ with circles in the entries for which 
$\La^\cd_i-\La^*_i=1$.  With $k=2$ and $r=3$, setting $x_2=x_3=x$, we have:
\begin{equation}
 \P_{{\tiny{(4,\cercle{2},1,0)}}}^{(-3/2)}(x_1,x,x,x_4) = (x-x_1)^2 
(x-x_4)^3 f(x_1,x_4,x) \, ,
\end{equation}
where $f(x_1,x_4,x)$ is not divisible by either $(x-x_1)$ or $(x-x_4)$.
Compare the power 2 (when $a=1$) and 3 (when $a=0$).
With $k=3$ and $r=2$, and setting $x_2=x_3=x_4=x$, we get
\begin{equation}
\P_{{\tiny{(3,\cercle{1},1,0,0)}}}^{(-4)}(x_1,x,x,x,x_5) = -x (x-x_1)(x-x_5)^2 
(x+3x_5-x_1).
\end{equation}
Note in this case that if $x_5=0$ then 
\begin{equation} \label{eqexample}
\P_{{\tiny{(3,\cercle{1},1,0)}}}^{(-4)}(x_1,x,x,x) = -x^3 (x-x_1)^2,
\end{equation}
and the power of $(x-x_1)$ is equal to two instead of one. This does not contradict Conjecture~\ref{conjecrma} however since in this case $N=4<k+m+1=5$.



\begin{appendix}
\section{Proof of Proposition~\ref{propquasi}}\label{Prop8}

We recall Proposition~\ref{propquasi} and then give its proof.

\noindent {\it Let $\Lambda$ be a superpartition obtained from a 
$(k,r,N)$-admissible superpartition $\Gamma$ by doing one of the following:
\begin{itemize}
\item[i)] removing a circle
\item[ii)] adding a circle
\item[iii)] changing a circle into a square
\item[iv)] changing a square into a circle
\end{itemize}
Then $P_{\Lambda}$ does not have a pole at $\alpha=\ad$.}
\begin{proof}
Case ii) and iv) are immediate from {Proposition~\ref{lemmapoles}}
since $\Lambda$ is still $(k,r,N)$-admissible
in those cases.  We will only prove Case i) as Case iii) follows similarly.

Consider case i).  We have
$
\Lambda^*= \Gamma^* 
$
and $\Lambda^\cd$ is obtained by removing a cell from $\Gamma^\cd$ in a
certain row $a$ (which is thus non-fermionic in $\La$). 
The proof consists in supposing that $P_\Lambda^{(\alpha)}$ 
has a pole at $\alpha =\ad$ and deriving a contradiction
to Lemma~\ref{lemsw}.
Let $\Omega$ and $w,\sigma \in S_N$  be such as in Lemma~\ref{lemsw}.
We will first prove that $w$ and $\sigma$ differ at most in two positions.
This will follow from two claims.

\noindent {\bf Claim I:} if $\sigma(i) > w(i)$ then 
\begin{itemize}
\item[(\itta)] $\sigma(i)=w(i)+k+1$,
\item[(\ittb)] $\Omega_i^\cd = \Omega_i^*$,
\item[(\ittc)] $\Lambda^*_{w(i)}-\Lambda^\cd_{\sigma(i)}=r-1$,
\item[(\ittd)] $a \neq \sigma(i)$ and 
$\Lambda_{\sigma(i)}^\cd-\Lambda_{\sigma(i)}^*=1$ ,
\item[(\itte)] if $a \neq w(i)$ then 
$\Lambda_{w(i)}^\cd-\Lambda_{w(i)}^*=1$ .
\end{itemize} 
By (\ref{3.10}) and (\ref{3.12})
, we have that $\sigma(i)=w(i)+t(k+1)$ for some integer $t>0$.
Hence
\begin{equation}
\Omega_i^*= \Lambda^*_{w(i)} + (w(i) -i ) \frac{r-1}{k+1},
\end{equation}
and
\begin{equation}
\Omega_i^\cd= \Lambda^\cd_{w(i)+t(k+1)} + (w(i)+t(k+1) -i ) \frac{r-1}{k+1}.
\end{equation}
Therefore
\begin{equation} \label{eqclaim1}
\Omega^\cd_i - \Omega^*_i = \Lambda^\cd_{w(i)+t(k+1)} - \Lambda^*_{w(i)} + t(r-1).
\end{equation}
Now $\Lambda^*=\Gamma^*$ is $(k+1,r,N)$-admissible
by Lemma~\ref{lemadm} and thus
$\Lambda^*_{w(i)}- \Lambda^\cd_{w(i)+t(k+1)} \geq rt -1$, which gives
\begin{equation}
0 \leq \Omega^\cd_i - \Omega^*_i \leq -rt+1 +t(r-1)=1-t \, .
\end{equation}
Therefore the only possibility is $t=1$ and (\itta) follows.  Setting $t=1$ in the
previous equation implies (\ittb).  Letting $\Omega^\cd=\Omega^*$ and $t=1$
in (\ref{eqclaim1})
gives (\ittc).  
For (\ittc) to occur, we need $\Lambda_{\sigma(i)}^\cd \neq \Lambda_{\sigma(i)}^*$, and thus (\ittd) follows
since row $a$ is not fermionic.
Since $\Lambda_{w(i)}^\cd-\Lambda_{\sigma(i)}^\cd \geq r$ whenever $w(i) \neq a$,
(\itte) follows from (\ittc).

\noindent {\bf Claim II:} if $\sigma(i) < w(i)$ then 
\begin{itemize}
\item[(\ittf)] $w(i)=\sigma(i)+k+1$,
\item[(\ittg)] $\Omega_i^\cd - \Omega_i^*=1$,
\item[(\itth)] $\Lambda^\cd_{\sigma(i)}-\Lambda^*_{w(i)}=r$,
\item[(\itti)] $\Lambda_{\sigma(i)}^\cd=\Lambda_{\sigma(i)}^*$ ,
\item[(\ittj)] if $a \neq \sigma(i)$ then 
$\Lambda_{w(i)}^\cd=\Lambda_{w(i)}^*$. 
\end{itemize}
By (\ref{3.10}) and (\ref{3.12}), we have that $w(i)=\sigma(i)+t(k+1)$ for some integer $t>0$.
Hence
\begin{equation}
\Omega_i^*= \Lambda^*_{\sigma(i)+t(k+1)} + (\sigma(i)+t(k+1) -i ) \frac{r-1}{k+1},
\end{equation}
and
\begin{equation}
\Omega_i^\cd= \Lambda^\cd_{\sigma(i)} + (\sigma(i) -i ) \frac{r-1}{k+1}.
\end{equation}
Therefore
\begin{equation} \label{eqclaim2}
\Omega^\cd_i - \Omega^*_i = \Lambda^\cd_{\sigma(i)} - \Lambda^*_{\sigma(i)+t(k+1)} - t(r-1).
\end{equation}
We have that $\Lambda^*=\Gamma^*$ is $(k+1,r,N)$-admissible
by Lemma~\ref{lemadm} and thus $\Lambda^\cd_{\sigma(i)}- \Lambda^*_{\sigma(i)+t(k+1)} \geq rt$, which gives
\begin{equation}
1 \geq \Omega^\cd_i - \Omega^*_i \geq rt -t(r-1)=t \, .
\end{equation}
Therefore the only possibility is $t=1$ and (\ittf) follows.  Setting $t=1$ in the
previous equation implies (\ittg).  Letting $\Omega^\cd-\Omega^*=1$ and $t=1$
in (\ref{eqclaim2}) gives (\itth).  
Since $\Lambda_{\sigma(i)}^*-\Lambda_{w(i)}^* \geq r$, for (\itth) to occur,
we need $\Lambda_{\sigma(i)}^\cd = \Lambda_{\sigma(i)}^*$, and thus (\itti) follows.
Since $\Lambda_{\sigma(i)}^\cd-\Lambda_{w(i)}^\cd \geq r$ whenever 
$\sigma(i) \neq a$,
(\ittj) follows from (\itth).

If $w$ and $\sigma$ do not coincide, then there exists a $i$ such that
$\sigma(i)>w(i)$.  From (\itta), this implies
$\sigma(i)=w(i)+k+1$. Now let $j$ be such that
$w(j)=\sigma(i)=w(i)+k+1$.
{From} (\itta) and (\ittf), we have
$\sigma(j)=w(i)+ k+1 \pm (k+1)$ (the case $\sigma(j)=w(i)+k+1$
is impossible since $\sigma(i)=w(i)+k+1$ and $i \neq j$).  
But if $\sigma(j)=w(i)+2(k+1)$
and $l$ is such that $w(l)= w(i)+2 (k+1)$ then the only option is
$\sigma(l)=w(i)+3(k+1)$.  Continuing in this way leads to a contradiction since
$\sigma$ is finite.  Therefore  $\sigma(j)=w(i)$.  We will show that this is impossible if $a \neq w(i)$.  We have from (\itth) (using $j$ instead of $i$)
\begin{equation} \label{eqcontr}
\Lambda^\cd_{w(i)} - \Lambda_{\sigma(i)}^* =r.
\end{equation}
If $a \neq w(i)$, we have from (\itte) that $\Lambda^\cd_{w(i)}-\Lambda^*_{w(i)}=1$.
Hence
\begin{equation}
\Lambda^\cd_{w(i)} - \Lambda_{\sigma(i)}^*  > \Lambda^*_{w(i)} - \Lambda_{\sigma(i)}^*=
 \Lambda^*_{w(i)} - \Lambda^*_{w(i)+k+1} \geq r,
\end{equation}
which
contradicts \eqref{eqcontr}.  Therefore $\sigma$ and $w$ coincide, except possibly at
two positions $i$ and $j$, where 
\begin{equation}\label{eqcoincide}
\sigma(i)=w(i)+k+1 \qquad {\rm and} \qquad
\sigma(j)=w(j)-k-1=w(i)=a.
\end{equation} 

We say that a permutation $w$ has a {\it descent} at $l$ if
$w(l)>w(l+1)$.  The next two claims give some consequences of $w$ and
$\sigma$ having a descent at $l$.

\noindent {\bf Claim III:}
if $w(l)>w(l+1)$ then
\begin{itemize}
\item[(\ittk)] $w(l)=w(l+1)+k$,
\item[(\ittl)] $\Lambda_{w(l+1)}^*-\Lambda_{w(l)}^*=r-1$,
\item[(\ittm)] $\Lambda^\cd_{w(l+1)}-\Lambda^*_{w(l+1)}=1$ if $a \neq w(l+1)$,
\item[(\ittn)] $\Omega_l^*=\Omega_{l+1}^*$.
\end{itemize}
Let $m$ be such that $w(l)-w(l+1)=m>0$.  Using $\Omega^*_{l}\geq \Omega^*_{l+1}$,
we have from \eqref{3.9} that
\begin{equation}
\frac{m+1}{k+1} (r-1) \geq \Lambda^*_{w(l+1)}- \Lambda^*_{w(l)}. 
\end{equation}
Since $\Gamma$ is $(k,r,N)$-admissible we get
\begin{equation}
\Lambda_{w(l+1)}^*-\Lambda_{w(l)}^* \geq \left\lfloor \frac{m}{k}
\right\rfloor r -1.
\end{equation}
It then follows that
\begin{equation}
\frac{m+1}{k+1}(r-1) \geq 
\Lambda_{w(l+1)}^*-\Lambda_{w(l)}^* \geq \left\lfloor \frac{m}{k}
\right\rfloor r -1,
\end{equation}
which implies $m=k$ given that $m \equiv k \mod k+1$ from \eqref{3.10}.
This immediately gives (\ittk) and (\ittl).
Assertion (\ittm) follows since  $\Gamma$ 
is $(k,r,N)$-admissible.  Using (\ittk) and (\ittl), we get
\begin{equation}\Omega_l^*-\Omega_{l+1}^*=\Lambda_{w(l)}^*-\Lambda_{w(l+1)}^*+(w(l)-w(l+1)+1) \frac{r-1}{k+1}=0,
\end{equation}
and (\ittn) follows.

\noindent {\bf Claim IV:}
if $\sigma(l)>\sigma(l+1)$ and $a \neq 
\sigma(l),\sigma(l+1)$
then
\begin{itemize}
\item[(\itto)] $\sigma(l)=\sigma(l+1)+k$ ,
\item[(\ittp)] $\Lambda_{\sigma(l+1)}^\cd-\Lambda_{\sigma(l)}^\cd=r-1$,
\item[(\ittq)] $\Lambda^\cd_{\sigma(l)}-\Lambda^*_{\sigma(l)}=1$,
\item[(\ittr)] $\Omega_l^\cd=\Omega_{l+1}^\cd$.
\end{itemize}
Following the steps of the proof in Claim~{\bf III}, let 
$m$ be such that $\sigma(l)-\sigma(l+1)=m>0$. 
Using $\Omega^\cd_{l}\geq \Omega^\cd_{l+1}$,
we have from \eqref{3.11} that
\begin{equation}
\frac{m+1}{k+1} (r-1) \geq \Lambda^\cd_{\sigma(l+1)}- \Lambda^\cd_{\sigma(l)}.
\end{equation}
Since $\Gamma$ is $(k,r,N)$-admissible we get
\begin{equation}
\Lambda_{\sigma(l+1)}^\cd-\Lambda_{\sigma(l)}^\cd \geq \left\lfloor 
\frac{m}{k}\right\rfloor r -1.
\end{equation}
It then follows that
\begin{equation}
\frac{m+1}{k+1}(r-1) \geq \Lambda_{\sigma(l+1)}^\cd-\Lambda_{\sigma(l)}^\cd
 \geq \left\lfloor \frac{m}{k}\right\rfloor r -1,
\end{equation}
which implies $m=k$  since $m \equiv k \mod k+1$
from {3.12}. 
This immediately gives (\itto) and (\ittp).
Assertion (\ittq) follows since  $\Gamma$ 
is $(k,r,N)$-admissible.
Using (\itto) and (\ittp), we get
\begin{equation}\Omega_l^\cd-\Omega_{l+1}^\cd=\Lambda_{\sigma(l)}^\cd-\Lambda_{\sigma(l+1)}^\cd
+(\sigma(l)-\sigma(l+1)+1) \frac{r-1}{k+1}=0,
\end{equation}
and (\ittr) follows.

Suppose that 
$w(l) > w(l+1)$, with
$w(l)=\sigma(l)$, $w(l+1)=\sigma(l+1)$ and $a\neq \sigma(l),\sigma(l+1)$.  From (\ittn) and (\ittr) we get that $\Omega_l^*=\Omega_{l+1}^*$
and $\Omega_l^\cd=\Omega_{l+1}^\cd$. We will show that this is not possible since
$\Omega_l^\cd-\Omega_{l}^*=1$ and thus $\Omega$ would not be a superpartition.
Indeed, using
\begin{equation}
\Omega_l^* = \Lambda_{\sigma(l)}^* + (\sigma(l)-l) \frac{r-1}{k+1},
\end{equation}
and
\begin{equation}
\Omega_l^\cd = \Lambda_{\sigma(l)}^\cd + (\sigma(l)-l) \frac{r-1}{k+1},
\end{equation}
we have from (\ittq) that $\Omega_l^\cd-\Omega_{l}^*=1$.

Recall from \eqref{eqcoincide} that $w$ and $\sigma$ coincide,  
except possibly at
two positions $i$ and $j$, where $\sigma(i)=w(i)+k+1$ and 
$\sigma(j)=w(j)-k-1=w(i)=a$ (and thus $i \equiv j \mod k+1$).   
{From} what we have established in the previous paragraph, 
we can have $w(l)>w(l+1)$ only if $l=i,i-1$ or $l=j,j-1$.  We now show
that $l$ cannot be equal to $i$ or $j-1$.  Suppose $l=i$. We have
$\sigma(i)=w(i)+k+1$ and, from (\ittk),
$w(i+1)=\sigma(i+1)=w(i)-k$.  Hence
$\sigma(i)-\sigma(i+1)=2k+1$ which violates (\itto) given that $a =w(i)$ is 
not equal
to either $\sigma(i)$ or $\sigma(i+1)$.  Suppose $l=j-1$.  From (\ittn) we
have $\Omega_{j-1}^*=\Omega_j^*$.  But since $\sigma(j)=w(j)-k-1$, we have
from (\ittg) that $\Omega_j^\cd -\Omega^*_j=1$, and hence we get the contradiction that
$\Omega$ is not a
superpartition.

Therefore, $w(l)>w(l+1)$ only if $l=i-1$ or $l=j$.  Suppose $i<j$ and that
we have descents 
at $i-1$ and $j$.  The only option is $j=i+k+1$ 
since otherwise there would be extra descents.  Thus $w(i+k+1)=w(i)+k+1$,
and hence from (\ittk) we get
$w(i-1)=w(i)+k$ and $w(i+k+2)=w(i+k+1)-k=w(i)+1$.   The relevant portion
of the permutation $w$ is thus
\begin{equation}
\begin{array}{cccccccc}
\cdots &i-1 & i & \cdots & i+k & i+k+1 & i+k+2 & \cdots\\
\cdots & w(i)+k & w(i) & \cdots & w(i+k) & w(i)+k+1 & w(i)+1 & \cdots
\end{array}
\end{equation}
But this is impossible since $w(i+k)$ cannot be equal to $w(i)+k$ (given that $w(i-1)=w(i)+k$)
and thus $w(i+k) \leq w(i)-1$ and there would necessarily be extra descents
in $w$.  Therefore, in this case there can be at most
one descent (in position $i+k+1$).  From the argument at the end of the proof of Proposition~2.6
of \cite{FJMM}, we conclude that $w$ is the identity\footnote{The argument goes as follows.  Suppose $w$ has exactly one descent (at $i+k+1$).
Then $w(j) \geq j$ for $1\leq j \leq i+k+1$.  Since $w(i+k+1) \equiv i+k+1 \mod k+1$, and $w(i+k+1)=i+k+1$ is impossible (there would not be a descent at 
$i+k+1$), this implies that $w(i+k+1)\geq i+2k+2$.  Hence
$w(i+k+2)=w(i+k+1)-k \geq i+k+2$, and given that
$w(j)<w(j+1)$ holds for $j\geq i+k+1$, we have also $w(j)\geq j$ for $j\geq i+k+1$.  Therefore the permutation $w$ is such that $w(j)\geq j$ for all $j$,
which
is obviously a contradiction}.  Therefore 
$\sigma$ needs to be the transposition $(i,i+k+1)$ in order not to be equal to the identity.  We will show that this implies
the contradiction $\Omega=\Lambda$.  We have $\Omega^*=\Lambda^*$ since
$w$ is the identity.  We also have that $\Omega^\cd$ and $\Lambda^{\cd}$
coincide except possibly in rows $i$ and $i+k+1$.  We have 
\begin{equation}
\Omega^\cd_i = \Lambda^\cd_{i+k+1} + (i+k+1-i) \frac{r-1}{k+1}= \Lambda^\cd_{i+k+1}
+ r-1 .\end{equation}
{From} (\ittc), we have $\Lambda^*_i-\Lambda^\cd_{i+k+1}=r-1$, and thus
$
\Omega^\cd_i = \Lambda^*_i
$.
But since $a=w(i)=i$, we have that $\Lambda^*_i=\Lambda^\cd_i$ and thus
$\Omega^\cd_i=\Lambda^\cd_i$, which means that $\Omega^\cd=\Lambda^\cd$. 

Finally, suppose $j<i$.  As in the previous case, it is easy to deduce that
$j=i-k-1$, $\sigma(i-1)=w(i-1)=w(i)+k$,
$\sigma(i)=w(i)+k+1$, $\sigma(i-k-1)=w(i)$ and $\sigma(i-k)=w(i-k)=w(i)+1$.
The relevant portion
of the permutation $\sigma$ is thus
\begin{equation}
\begin{array}{ccccccc}
\cdots &i-k-1 & i-k & \cdots & i-1 & i & \cdots\\
\cdots & w(i) & w(i)+1 & \cdots & w(i)+k & w(i)+k+1  & \cdots
\end{array}
\end{equation}
Since $w$ does not have descents in $i$ and $i-k-2$, we have
$\sigma(i+1)=w(i+1) > w(i)$ and $\sigma(i-k-2)=w(i-k-2)< w(i-k-1)=w(i)+k+1$.
This implies
$\sigma(i+1) \geq w(i)+k+2$ and $\sigma(i-k-2)\leq w(i)-1$
 and hence 
$\sigma$ has no descent, which means that $\sigma$ is the identity.  
Therefore 
$w$ needs to be the transposition $(i-k-1,i)$ in order not to be equal 
to the identity.  We will show that this is impossible.
Since $\sigma$ is the identity we have
$\Omega^\cd=\Lambda^\cd$.  Now, from (\itth) with $i$ replaced by 
$i-k-1$ we have 
$\Lambda^\cd_{i-k-1} -\Lambda_i^*=r$ and thus
\begin{equation}
\Omega_{i-k-1}^* = \Lambda_i^* + \bigl(i-(i-k-1)\bigr) \frac{r-1}{k+1}
=
\Lambda_i^* + r-1= \Lambda_{i-k-1}^\cd-1= \Lambda_{i-k-1}^*-1,
\end{equation}
since $w(i)=i-k-1=a$.  Hence 
\begin{equation}
\Lambda^\cd_{i-k} = \Omega^\cd_{i-k} \leq \Omega^*_{i-k-1} < \Lambda^*_{i-k-1}=r + \Lambda_i^*,
\end{equation}
which implies that $\Lambda_{i-k}^\cd- \Lambda^*_{i}=
\Gamma^\cd_{i-k} -\Gamma^*_{i} < r$.  This contradicts  the 
$(k,r,N)$-admissibility of $\Gamma$ and completes the proof.
\end{proof}

\section{Proof of formulas~\eqref{forp0} and \eqref{forq}}
\label{apppieri}
We first need to define a few concepts introduced in \cite{DLMeva}.

If the first column of the diagram of $\La$ does not contain a circle,
we introduce the ``column-removal'' operation $\C$ defined such that
$\C \La$ is the
superpartition
whose diagram is obtained by removing the
first column of the diagram of $\La$ (the
operation is illustrated in Fig.~\ref{FigC}).

If the first column of the diagram of
$\La $ contains a circle, we define the ``circle-removal'' operation  $\widetilde{\C}$ such that the diagram of $\widetilde{\C} \La$ is obtained from that of $\La$ by removing the circle in the first column of the diagram of $\La$
(also illustrated in Fig.~\ref{FigC}).
\begin{figure}[h]\caption{Operators $\C$ and $\widetilde{\C}$}\label{FigC}
 \begin{equation*}\mathcal{C} \,:\quad {\tableau[scY]{&&&\\&&\bl\tcercle{}\\&\\&\bl\tcercle{}}}  \longmapsto \quad {\tableau[scY]{&&\\&\bl\tcercle{}\\ \\\bl\tcercle{}}}
\qquad \widetilde{\mathcal{C}} \,:\quad {\tableau[scY]{&&&\bl\tcercle{}\\&\\&\bl\tcercle{}\\\bl\tcercle{}}}  \longmapsto \quad {\tableau[scY]{&&&\bl\tcercle{}\\&\\&\bl\tcercle{}\\\bl }}
 \end{equation*}
\end{figure}

Similarly, we can introduce two row operations whose
actions on diagrams is illustrated in Fig. \ref{FigRtilde}.
\begin{figure}[h]\caption{Operators $\R$ and $\widetilde{\R}$}\label{FigRtilde}
\begin{equation*}\mathcal{R} \,:\quad {\tableau[scY]{&&&\\&&\bl\tcercle{}\\&\\&\bl\tcercle{}}}  \longmapsto \quad {\tableau[scY]{ &&\bl\tcercle{}\\& \\&\bl\tcercle{}\\ \bl }}
\qquad \widetilde{\mathcal{R}} \,:\quad {\tableau[scY]{&&&\bl\tcercle{}\\&\\&\bl\tcercle{}\\&\bl }}  \longmapsto \quad {\tableau[scY]{&&&\bl \\&\\&\bl\tcercle{}\\&\bl }}
 \end{equation*}
\end{figure}

The following proposition from \cite{DLMeva} will prove essential.
\begin{proposition}[\cite{DLMeva}] 
Let $\La$ be a superpartition
such that $\ell(\La^\cd)=\ell(\La^*)=\ell$.  Then
\begin{equation} \label{PropFactoI}
P_\La^{(\aa)}(x_1,\ldots,x_\ell;\theta_1,\ldots,\theta_\ell)=  x_1\cdots x_\ell \,
 P_{\C\La}^{(\aa)}(x_1,\ldots,x_\ell;\theta_1,\ldots,\theta_\ell) .\end{equation}
Similarly, let $\La$ be a superpartition such that 
$\ell(\La^\cd)=\ell(\La^*)+1=\ell$. Then
\begin{equation}\label{PropFactoII}
(-1)^{m-1}
\Bigl[ \partial_{\theta_\ell}\,P_\La^{(\aa)}(x_1,\ldots,x_\ell;\theta_1,\ldots,\theta_\ell)
\Bigr]_{x_\ell=0}
=
P_{\widetilde{\C}\La}^{(\aa)}(x_1,\ldots,x_{\ell-1};\theta_1,\ldots,\theta_{\ell-1}).
\end{equation}

Let $(x_-;\theta_-)=(x_2,x_3,\dots;\theta_2, \theta_3,\dots)$.  
If the first row of the diagram of $\Lambda$ is bosonic
(that is, $\La^*_1= \La^\circledast_1=k$),
then
\begin{equation} \label{PropFactoIdual}
\coeff{x_1^k} P_\La^{(\aa)}(x;\theta) = P_{\R\La}^{(\aa)}(x_-;\theta_-).
\end{equation}
Similarly, if the first row of the diagram of $\La$ is fermionic
(that is, $\La_1^*=\La_1^\circledast-1=k$),
then
\begin{equation} \label{PropFactoIIdual}
\coeff{x_1^k}\, \partial_{\theta_1} \, P_\La^{(\aa)}(x;\theta)
= P_{\R\widetilde{\R}\La}^{(\aa)}(x_-;\theta_-).
\end{equation}
\end{proposition}

We now proceed to the proof of formula \eqref{forp0}.
\begin{proof}[Proof of formula \eqref{forp0}]
An equivalent form of formula \eqref{forp0}
is given in \cite{jeff}. We reproduce this proof, with minor improvements, since its pattern also applies to the proof of \eqref{forq}. 
The structure of the proof follows the original derivation of the Pieri formulae for Jack polynomials in \cite{Stan}. The key steps rely heavily on the results of \cite{DLMeva}.

The proof is done in the case where the number of variables is infinite.
The finite case is recovered by letting $x_i=0$ and $\theta_i=0$
for all $i>N$.  We will use the notation
\begin{equation}\label{p0P}\tilde p_0^{(\ell)}= \sum_{i=1}^\ell \theta_i  \quad {\rm and} \quad
P_{\Lambda}^{(\alpha,\ell)} = P_{\Lambda}
(x_1,\dots,x_\ell,\theta_1,\dots,\theta_\ell;\alpha).
\end{equation}
Note that $P_{\Lambda}^{(\alpha,\ell)}=0$ if $\ell(\Lambda^\cd)>\ell$.

The proof proceeds by induction on the degree of $\Lambda^\cd$.
{From} \cite{DLMeva} (cf. Proposition 11 applied to the case $n=0$),
we know that
\begin{equation} \label{twosides}
\tilde p_0 \, P_{\Lambda}^{(\aa)} = \sum_{\Omega} c_{\Lambda \Omega}
P_{\Omega}^{(\aa)},
\end{equation}
for some coefficients $c_{\Lambda \Omega} \in \mathbb Q(\alpha)$,
where the sum is over all $\Omega$'s obtained by adding a circle to
$\Lambda$.

Suppose first that the first column of $\Lambda$ does not have
a circle and is of length $\ell(\La^\cd)=\ell(\La^*)=\ell$.
Restricting \eqref{twosides} to $\ell$ variables
we get
\begin{equation}
\tilde p_0^{(\ell)}  P_{\Lambda}^{(\alpha,\ell)} = \sum_{\Omega} c_{\Lambda \Omega} \, 
P_{\Omega}^{(\alpha,\ell)},
\end{equation}
where 
all the $\Omega$'s in the sum have length $\ell(\Om^\cd)=\ell$
 and do not have a circle
in the first column (otherwise $P_{\Omega}^{(\alpha,\ell)}$ would
be equal to zero). Using \eqref{PropFactoI}, we get
\begin{equation}
\tilde p_0^{(\ell)}   x_1 \cdots x_\ell 
P_{\mathcal C \Lambda}^{(\alpha,\ell)} = \sum_{\Omega} c_{\Lambda \Omega} \,  x_1 \cdots x_\ell \,
P_{\mathcal C \Omega}^{(\alpha,\ell)}.
\end{equation}
Canceling $x_1 \cdots x_\ell$ on both sides, we can then use
induction since the degree of $(\mathcal C \Lambda)^\cd$ is smaller
than that of $\Lambda^\cd$.  The result follows in this
case since 
\begin{equation}
 \prod_{s \in {\rm col}_{\Omega^\circ}} 
\frac{h_{\Lambda}^{(\alpha)}(s)}{h_{\Omega}^{(\alpha)}(s)} =
 \prod_{s \in {\rm col}_{\mathcal C \Omega^\circ}}
\frac{h_{ \mathcal C\Lambda}^{(\alpha)}(s)}{h_{\mathcal C\Omega}^{(\alpha)}(s)} ,
\end{equation}
and $\# \Omega^\circ = \# \mathcal C \Omega^\circ$.  This covers all the cases
where $\ell(\Om^\cd)=\ell$.

Suppose now
that $\Gamma$ is obtained from $\Lambda$ by adding a circle to the first
column of $\La$. 
Isolating the term $P_\Gamma$ and using the result we just established,
we then have
\begin{equation}
\tilde p_0 \, P_{\Lambda}^{(\aa)} = 
 c_{\Lambda \Gamma} P_{\Gamma}^{(\aa)} + \sum_{\Omega} (-1)^{\# \Omega^\circ} 
\left( \prod_{s \in {\rm col}_{\Omega^\circ}} 
\frac{h_{\Lambda}^{(\alpha)}(s)}{h_{\Omega}^{(\alpha)}(s)}  \right) P_{\Omega}^{(\aa)},
\end{equation}
where the sum is over $\Omega$'s 
with $\ell(\Om^\cd)=\ell$.
Applying the endomorphism $\hat \omega_\aa$, we obtain from \eqref{dual}
\begin{equation}
(-1)^{m} \alpha \tilde p_0 
\|P_{\Lambda}\|^2 P_{\Lambda'}^{(1/\alpha)}
= 
 c_{\Lambda \Gamma} \|P_{\Gamma}\|^2 P_{\Gamma'}^{(1/\alpha)}
+ 
\sum_{\Omega} (-1)^{\# \Omega^\circ} 
\left( \prod_{s \in {\rm col}_{\Omega^\circ}} 
\frac{h_{\Lambda}^{(\alpha)}(s)}{h_{\Omega}^{(\alpha)}(s)}  \right) 
 \|P_{\Omega}\|^2 P_{\Omega'}^{(1/\alpha)} \, ,
\end{equation}
where we recall that the norm $\|P_{\Lambda}\|^2$ is defined in \eqref{norm}.
Differentiating with respect to $\ta_1$ on both sides and 
using \eqref{PropFactoIIdual}, we get
\begin{equation}
(-1)^{m}\alpha \|P_{\Lambda}\|^2 P_{\mathcal R \Lambda'}^{(1/\alpha)}(x_-,\ta_-)
= 
 c_{\Lambda \Gamma} \|P_{\Gamma}\|^2 
P_{\mathcal R \tilde{ \mathcal R} \Gamma'}^{(1/\alpha)}(x_-,\ta_-).
\end{equation}
Since $\Lambda'=\tilde{ \mathcal R} \Gamma'$
we obtain, from \eqref{norm},
\begin{equation}
c_{\Lambda \Gamma} = (-1)^m \alpha \frac{\|P_{\Lambda}\|^2}{\|P_{\Gamma}\|^2}
= (-1)^m \left( \prod_{s \in {\Lambda}} 
\frac{h_{\Lambda}^{(\alpha)}(s)}{h^{\Lambda}_{(\alpha)}(s)}  \right) 
\left( \prod_{s \in \Gamma} 
\frac{h^{\Gamma}_{(\alpha)}(s)}{h_{\Gamma}^{(\alpha)}(s)}  \right) 
=(-1)^{\# \Gamma^\circ} 
\left( \prod_{s \in {\rm col}_{\Gamma^\circ}} 
\frac{h_{\Lambda}^{(\alpha)}(s)}{h_{\Gamma}^{(\alpha)}(s)}  \right) .
\end{equation} 
The last equality holds since the hook-lengths of the same type cancel two-by-two from the two diagrams for all $s$  except those
in the row and the column that are rendered fermionic by the added circle. But since the circle has been added in the first column, the fermionized row has no square.  Hence only the squares in the first column have to be considered. Finally, only the ratio of the upper-hooks do not cancel (since for $s \in {\rm col}_{\Gamma^\circ}$, we have $l_{\Gamma^\cd}(s)=l_{\La^\cd}(s)+1=l_{\La^*}(s)+1$).
This concludes the case where the first column of $\Lambda$ does not have
a circle.

Finally, suppose that the first column of $\Lambda$ has a circle and is 
of length $\ell(\La^\cd)=\ell$.  
In this case the $\Omega$'s in \eqref{twosides}
also have a circle in the first column and are of
length $\ell(\Om^\cd)=\ell$
(each
$\Omega$ is obtained by adding a circle to $\Lambda$ in a column other than the first and necessarily shorter since two circles cannot appear in the same row).  Working in
$\ell$ variables we obtain
\begin{equation}
(-1)^{m} \bigl[ \partial_{\theta_{\ell}} ( \tilde p_0^{(\ell)} 
P_{\Lambda}^{(\alpha,\ell)}) 
\bigr]_{x_\ell=\theta_\ell=0} =
(-1)^{m} \bigl[  P_{\Lambda}^{(\alpha,\ell)} 
\bigr]_{x_\ell=\theta_\ell=0} + (-1)^{m-1} \bigl[\tilde p_0^{(\ell)} 
\partial_{\theta_{\ell}} P_{\Lambda}^{(\alpha,\ell)} 
\bigr]_{x_\ell=\theta_\ell=0} = \tilde p_0^{(\ell-1)} P_{\tilde{\mathcal C} \Lambda}^{(\alpha,\ell-1)} ,
\end{equation}
where the last equality follows from \eqref{PropFactoII}.
Hence, we get from \eqref{twosides} 
\begin{equation}
\tilde p_0^{(\ell-1)} P_{\tilde{\mathcal C} \Lambda}^{(\alpha,\ell-1)}
 =  \sum_{\Omega} c_{\Lambda \Omega}
(-1)^{m} \bigl[ \partial_{\theta_{\ell}} (P_{\Omega}^{(\alpha,\ell)}) 
\bigr]_{x_\ell=\theta_\ell=0}
=
 \sum_{\Omega} c_{\Lambda \Omega} P_{\tilde{\mathcal C} \Omega}^{(\alpha,\ell-1)} ,
\end{equation}
and the result follows again by induction (the degree of 
$(\tilde{\mathcal C} \La)^\cd$ is smaller than that of $\La^\cd$)
since
\begin{equation}
 \prod_{s \in {\rm col}_{\Omega^\circ}} 
\frac{h_{\Lambda}^{(\alpha)}(s)}{h_{\Omega}^{(\alpha)}(s)} =
 \prod_{s \in {\rm col}_{{\tilde{\mathcal C}} \Omega^\circ}}
\frac{h_{ \tilde{\mathcal C}\Lambda}^{(\alpha)}(s)}{h_{\tilde{\mathcal C}\Omega}^{(\alpha)}(s)} ,
\end{equation}
and $\# \Omega^\circ = \# \tilde{\mathcal C} \Omega^\circ$.
\end{proof}

We finally proceed to the proof of formula \eqref{forq}.
\begin{proof}[Proof of formula \eqref{forq}]
Unless otherwise stated, we will assume throughout the proof that
the number of variables is infinite.
The finite case will 
follow immediately by setting $x_i=0$ and
$\theta_i=0$  
for every $i >N$.
As in the proof of formula \eqref{forp0},
we use the notation
\begin{equation}q^{(\ell)}= \sum_{i=1}^\ell \theta_i \partial_{x_i}
 \quad {\rm and} \quad
P_{\Lambda}^{(\alpha,\ell)} = P_{\Lambda}
(x_1,\dots,x_\ell,\theta_1,\dots,\theta_\ell;\alpha).
\end{equation}

It is easy to check that $q$ and $\hat \omega_{\alpha}$ commute 
when
acting on the powers sums
\begin{equation}
q \,  \hat \omega_{\alpha} (p_{\Lambda}) =
\hat \omega_{\alpha} ( q \, p_{\Lambda} ) \, , 
\end{equation}
hence $q$ and $\hat \omega_{\alpha}$ commute when acting on the whole
space of symmetric superpolynomials.

The proof proceeds again by induction on the degree of $\Lambda^\cd$.
We have proven in \cite{DLMeva} (cf. Eq. (A.22) with $d=q$) that 
\begin{equation} \label{twosides2}
q \,  P_{\Lambda}^{(\aa)} = \sum_{\Omega}  d_{\Lambda \Omega}
  \, P_{\Omega}^{(\aa)},
\end{equation}
where the sum is  
over certain $\Omega$'s such that $\Omega^{\cd}=\Lambda^\cd$.  
We now show
that if $\Omega$ appears in \eqref{twosides2}, then
the length of $\Omega^*$ cannot be larger than that of 
$\Lambda^*$.  Using the duality \eqref{dual}
and the commutation of
$q$ and $\hat \omega_{\alpha}$, we have
\begin{equation}
\hat \omega_{\alpha} ( q \, P_{\Lambda}^{(\aa)}) =
q\, \omega_{\alpha} ( P_{\Lambda}^{(\aa)} ) = (-1)^{\binom{m}{2}} \| P_{\Lambda} \|^2 
q \, 
P_{\Lambda'}^{(1/\alpha)}.
\end{equation}
Hence, we get from \eqref{twosides2}
\begin{equation}
(-1)^{\binom{m}{2}} \| P_{\Lambda} \|^2 
q \, 
P_{\Lambda'}^{(1/\alpha)} = 
(-1)^{\binom{m}{2} }\sum_{\Omega} \| P_{\Omega} \|^2  
d_{\Lambda \Omega}
  \, 
 P_{\Omega'}^{(1/\alpha)}.
\end{equation}
Suppose there are some $\Omega$'s in the sum
such that the length of $\Omega^*$ is larger than that of $\Lambda^*$, and
let $\ell$ be the length of  $\Lambda^\cd$.
Since $\Omega^\cd=\Lambda^\cd$,
 this can only happen if the length of $\Omega^*$
is $\ell$ and that of $\Lambda^*$ is $\ell-1$.
Taking the coefficient of $x_1^\ell$ on both sides of the equation and
using \eqref{PropFactoIdual}
we obtain
\begin{equation}
0 = 
(-1)^{\binom{m}{2} }\sum_{\Omega} \| P_{\Omega} \|^2  
d_{\Lambda \Omega}
  \,  P_{\mathcal R \Omega'}^{(1/\alpha)}(x_-,\theta_-),
\end{equation}
where the sum is over $\Omega$'s such that $\ell(\Omega^*)=\ell$.
But this is a contradiction since the 
$P_{\mathcal R \Omega'}^{(1/\alpha)}(x_-,\ta_-)$'s
are linearly independent.

Let $\ell$ be again the length of $\Lambda^\cd$.
We consider first the $\Omega$'s in \eqref{twosides2} 
such that the length of $\Omega^*$ is smaller than that
of
$\Lambda^*$. 
 We will show that
there is at most one such $\Omega$ and that, as formula \eqref{forq} claims,
$\Omega$ is obtained 
by replacing (if possible) the lowest square in the
first column of $\Lambda$ by a circle.

By \eqref{PropFactoII},
we have when we restrict to $\ell$ variables:
\begin{align}
\bigl[ \partial_{\theta_\ell} (q^{(\ell)} P_{\Lambda}^{(\alpha,\ell)})\bigr]_{x_\ell=0}
=\bigl[ \partial_{\theta_\ell} (q^{(\ell)}  x_1 \cdots x_{\ell } 
P_{\mathcal C \Lambda}^{(\alpha,\ell)})\bigr]_{x_\ell=0}
&= \bigl[ \partial_{\theta_\ell} (\theta_\ell  x_1 \cdots x_{\ell-1} 
P_{\mathcal C \Lambda})^{(\alpha,\ell)}\bigr]_{x_\ell=0}\nonumber\\
&=  x_1 \cdots x_{\ell-1} \bigl[
P_{\mathcal C \Lambda}^{(\alpha,\ell)}\bigr]_{x_\ell=\theta_\ell=0}.
\end{align}
But since 
\begin{equation}\bigl[P_{\mathcal C \Lambda}^{(\alpha,\ell)} \bigr]_{x_\ell=\theta_\ell=0}=
\left \{ 
\begin{array}{ll}
P_{\mathcal C \Lambda}^{(\alpha,\ell-1)}  & {\rm if~} \ell({\mathcal C} \Lambda)< \ell ,\\
0 & {\rm otherwise}
\end{array}
\right .
\end{equation}
we get
\begin{equation}
\bigl[ \partial_{\theta_\ell} (q^{(\ell)} P_{\Lambda}^{(\alpha,\ell))})\bigr]_{x_\ell=0}
=
\left \{ 
\begin{array}{ll}
P_{\Gamma}^{(\alpha,\ell-1)}  & {\rm if~} \ell({\mathcal C} \Lambda)< \ell, \\
0 & {\rm otherwise},
\end{array}
\right .
\end{equation}
where $\Gamma$ is obtained from $\Lambda$ by removing the lowest square in the
first column of $\Lambda$.    Using \eqref{twosides2} this implies
\begin{equation}
\sum_{\Omega} d_{\Lambda \Omega}
(-1)^m \bigl[ \partial_{\theta_\ell} (P_{\Omega}^{(\alpha,\ell))})\bigr]_{x_\ell=0}
= \left \{ 
\begin{array}{ll}
(-1)^m P_{\Gamma}^{(\alpha,\ell-1)}  & {\rm if~} \ell({\mathcal C} \Lambda)< \ell, \\
0 & {\rm otherwise}.
\end{array}
\right .
\end{equation}
If $\Omega$ appears in the sum
and
the length of
$\Omega^*$ is equal to $\ell$, then
\begin{equation}
\bigl[P_{\Omega}^{(\alpha,\ell)}\bigr]_{x_\ell=0}= \bigl[x_1 \cdots x_{\ell} P_{\mathcal C\Omega}^{(\alpha,\ell))}\bigr]_{x_\ell=0}=0 
\end{equation}
Therefore the remaining terms in the sum are such that 
$\ell(\Omega^*)=\ell-1$. 
Those $\Omega$'s need to have a circle in their first
column since $\Omega^\cd=\Lambda^\cd$ and the length of 
$\Lambda^\cd$ is $\ell$.
Hence from \eqref{PropFactoII},
\begin{equation}
\sum_{\Omega} d_{\Lambda \Omega}
(-1)^m \bigl[ \partial_{\theta_\ell} P_{\Omega}^{(\alpha,\ell)}
\bigr]_{x_\ell=0} = \sum_{\Omega} d_{\Lambda \Omega} P_{\tilde{\mathcal C} \Omega}^{(\alpha,\ell-1)}
=  \left \{ 
\begin{array}{ll}
(-1)^m P_{\Gamma}^{(\alpha,\ell-1)}  & {\rm if~} \ell({\mathcal C} \Lambda)< \ell ,\\
0 & {\rm otherwise},
\end{array}
\right .
\end{equation}
where the sum is over $\Omega$'s such that $\ell(\Omega^*)=\ell-1$.
Equating both sides of the equation, we get that
the only possibly  non-zero $d_{\Lambda \Omega}$ is such that
$\Gamma=\tilde{\mathcal C} \Omega$ and that 
\begin{equation}
d_{\Lambda \Omega} = (-1)^m.
\end{equation}
Formula \eqref{forq} is thus proved in this case.

We now prove the proposition in the cases where the 
length of $\Omega^*$ is equal to the length of $\Lambda^*$.
By duality and the commutation of $q$ and $\hat \omega_\alpha$, we have
from \eqref{twosides2}
\begin{equation} \label{eqspeciale}
(-1)^m \| P_{\Lambda}\|^2 q P_{\Lambda'}^{(1/\alpha)} = \sum_{\Omega} d_{\Lambda \Omega}
\| P_{\Omega} \|^2 P_{\Omega'}^{(1/\alpha)} \, .
\end{equation} 
Suppose first that $\Lambda$ does not have a circle in its first column.
Using \eqref{PropFactoIdual}, we get
\begin{equation}
\underset{x_1^{\ell}}{\rm coeff} \left( q P_{\Lambda'}^{(1/\alpha)} \right)=
q P_{(\mathcal C \Lambda)'}^{(1/\alpha)}(x_-,\ta_-) .
\end{equation}
We have seen that all the $\Omega$'s that appear in \eqref{twosides2}
are such that $\ell(\Omega^*)\leq \ell(\Lambda^*)=\ell$.  Therefore,
taking the coefficient of $x_1^\ell$ on both sides of \eqref{eqspeciale},
we obtain by \eqref{PropFactoIdual}
\begin{equation}
(-1)^m \| P_{\Lambda}\|^2 q P_{(\mathcal C \Lambda)'}^{(1/\alpha)}(x_-,\ta_-)
= 
\sum_{\Omega} d_{\Lambda \Omega}
\| P_{\Omega} \|^2  P_{(\mathcal C \Omega)'}^{(1/\alpha)}(x_-,\ta_-),
\end{equation}
where the sum is over $\Omega$'s such that $\ell(\Omega^*)=\ell$.
This gives immediately
\begin{equation}
(-1)^m \| P_{\Lambda}\|^2 d_{(\mathcal C \Lambda)' (\mathcal C \Omega)' } = 
d_{\Lambda \Omega}
\| P_{\Omega} \|^2  \, ,
\end{equation}
and can thus conclude by induction, since the degree of $({\mathcal C} \La)^\cd$
is smaller than that of $\La^\cd$, that
\begin{equation}
d_{\Lambda \Omega} = (-1)^{m-\# {\Gamma^\circ}} \frac{\| P_{\Lambda}\|^2}{\| P_{\Omega}\|^2}
\left( \prod_{s \in {\rm row}_{{\Gamma}^\circ}} 
\frac{h_{(\mathcal C  \Lambda)'}^{(1/\alpha)}(s)}{h_{(\mathcal C \Omega)'}^{(1/\alpha)}(s)}  
\right),
\end{equation}
where $\Gamma=(\mathcal C \Omega)'$.  Since $(\mathcal C  \Omega)'$ is
obtained by replacing a square of $(\mathcal C  \Lambda)'$ by a circle,
we have that $\Omega$ is also obtained from $\Lambda$ by replacing
a square by a circle.  It is easy to see that
$m-\# {\Gamma^\circ}= \#\Omega^\circ$ and that
\begin{equation}
 \prod_{s \in {\rm row}_{{\Gamma}^\circ}} 
\frac{h_{(\mathcal C  \Lambda)'}^{(1/\alpha)}(s)}{h_{(\mathcal C \Omega)'}^{(1/\alpha)}(s)} 
= \prod_{s \in {\rm col}_{\Omega^\circ}} 
\frac{h^{\Lambda}_{(\alpha)}(s)}{h^{\Omega}_{(\alpha)}(s)} ,
\end{equation}
given that $h^\Lambda_{(\alpha)}(i,j)=\alpha h_{\Lambda'}^{(1/\alpha)}(j,i)$.
Formula \eqref{forq} then follows in this case since
\begin{equation}
 \frac{\| P_{\Lambda}\|^2}{\| P_{\Omega}\|^2} =  \left( \prod_{s \in {\rm col}_{\Omega^\circ}} \frac{h^{\Omega}_{(\alpha)}(s)}{h^{\Lambda}_{(\alpha)}(s)}  \right)
 \left( \prod_{s \in {\rm row}_{\Omega^\circ}} 
\frac{h_{\Lambda}^{(\alpha)}(s)}{h_{\Omega}^{(\alpha)}(s)} \right) .
\end{equation}

Finally, the case where $\Lambda$ has a circle in its first column is
proven in a similar way.  Using \eqref{PropFactoIIdual} instead of
\eqref{PropFactoIdual},  we obtain
the recursion
\begin{equation}
(-1)^m \| P_{\Lambda}\|^2 d_{(\mathcal C \tilde {\mathcal C} \Lambda)' 
(\mathcal C \tilde {\mathcal C} \Omega)' } = 
d_{\Lambda \Omega}
\| P_{\Omega} \|^2  \, ,
\end{equation}
which gives again the desired result by induction.
\end{proof}

\section{Character formulas for $ {\mathcal{F}}_{N}^{(k)}$}\label{Characters}

{  In this appendix, we present sample expressions for the character
  \begin{equation} \ch  {\mathcal{F}}_{N}^{(k)}:=\sum_{n\geq 1, m\geq 0}   \dim { {\mathcal F}}_{N,n,m}^{(k)} \, u^n v^m. 
\end{equation} 
  Let us illustrate these computations by considering the determination of  $ \dim {\mathcal F}_{3,3,2}^{(1)}$.  Any element $f\in {\mathcal F}_{3,3,2}^{(1)}$ is of the form }
\begin{equation}\label{eqfF} f=a_1m_{(3,0;0)}+a_2m_{(2,1;0)}+a_3m_{(2,0;1)}+a_4 m_{(1,0;2)},\qquad a_i\in\mathbb{C},
\end{equation} and such that
\begin{equation} \label{eqfF2} f(x_1,x_2,x_3,\theta_1,\theta_2,\theta_3)|_{x_1=x_2}= 0.
\end{equation}
Now substitute \eqref{eqfF} into \eqref{eqfF2} and in the resulting equation,  collect the coefficients of all distinct non-symmetric monomials $\theta_{i_2}\theta_{i_3}x_2^{\eta_2}x_3^{\eta_3}$, {with} $i_2<i_3$ and $\eta=(\eta_2,\eta_3)$ {being} a composition.  This leads to the following linear system of equation:
\begin{equation} a_{1}=0,\qquad a_{2}+a_3=0,\qquad a_{2}
-a_{4}=0,\qquad a_1+a_3+a_4=0.
\end{equation}
Hence, $f\in {\mathcal F}_{3,3,2}^{(1)}$ if and only if 
\begin{equation} a_{1}=0,\qquad a_{3}=-a_{2},\qquad a_{4}=a_2,
\end{equation}
for some $a_2\in\mathbb{C}$. Consequently,  $\dim {\mathcal F}_{3,3,2}^{(1)}=1$. {Granting our conjecture, this is not surprising given that} there is only one $(1,2,3)$-admissible and of degree $(3|1)$, namely $(2,1;)$.  

The procedure just illustrated for the computation of $\dim {\mathcal F}_{3,3,2}^{(1)}$ can be easily generalized and computer implemented. Collecting together the various  $\dim {\mathcal F}_{N,n,m}^{(k)}$  for a fixed value of $N$, we thereby construct the lowest terms of the character $\ch {\mathcal F}_N^{(k)}$.

  We have been able to construct the characters  $\ch \mathcal{F}_N^{(k)}$ up to degree $n=12$ (and note that $m$ is always bounded by the relations $m\leq n$ and $m(m-1)/2\leq n$).   Here are some examples (for presentation purposes, the series are truncated  at degree $n=10$):
{\small 
 \begin{multline*}
\ch \mathcal{F}_2^{(1)} = uv+ \left( 2\,v+{v}^{2}+1 \right) {u}^{2}+ \left( 1+2\,{v}^{2}+3\,v
 \right) {u}^{3}+ \left( 2+2\,{v}^{2}+4\,v \right) {u}^{4}+ 
\\
\left( 2+5
\,v+3\,{v}^{2} \right) {u}^{5}+ \left( 3+6\,v+3\,{v}^{2} \right) {u}^{
6}+ \left( 4\,{v}^{2}+7\,v+3 \right) {u}^{7} \\
+ \left( 4\,{v}^{2}+8\,v+4
 \right) {u}^{8}+ \left( 9\,v+5\,{v}^{2}+4 \right) {u}^{9}
+O(u^{10})
\end{multline*}
\begin{multline*}
\ch \mathcal{F}_3^{(1)} = \left( {v}^{2}+{v}^{3} \right) {u}^{3}+ \left( v+2\,{v}^{2}+{v}^{3}
 \right) {u}^{4} \\
+ \left( 2\,v+4\,{v}^{2}+2\,{v}^{3} \right) {u}^{5}+
 \left( 1+4\,v+6\,{v}^{2}+3\,{v}^{3} \right) {u}^{6}+ \left( 1+6\,v+9
\,{v}^{2}+4\,{v}^{3} \right) {u}^{7}\\
+ \left( 2+9\,v+12\,{v}^{2}+5\,{v}
^{3} \right) {u}^{8}+ \left( 3+12\,v+16\,{v}^{2}+7\,{v}^{3} \right) {u
}^{9}
+O(u^{10})
\end{multline*}
\begin{multline*}
\ch \mathcal{F}_4^{(1)} = 
\left( {v}^{3}+{v}^{4} \right) {u}^{6}+ \left( {v}^{2}+2\,{v}^{3}+{v}
^{4} \right) {u}^{7} + \left( 2\,{v}^{2}+4\,{v}^{3}+2\,{v}^{4} \right) 
{u}^{8}+ \left( v+5\,{v}^{2}+7\,{v}^{3}+3\,{v}^{4} \right) {u}^{9}
+O(u^{10})
\end{multline*}
\begin{multline*}
\ch \mathcal{F}_3^{(2)} = 
uv+ \left( 1+2\,{v}^{2}+3\,v \right) {u}^{2}+ \left( 2+{v}^{3}+4\,{v}^
{2}+5\,v \right) {u}^{3}+ \left( 8\,v+6\,{v}^{2}+{v}^{3}+3 \right) {u}
^{4}\\
+ \left( 4+11\,v+9\,{v}^{2}+2\,{v}^{3} \right) {u}^{5}+ \left( 6+
15\,v+12\,{v}^{2}+3\,{v}^{3} \right) {u}^{6}+ \left( 16\,{v}^{2}+19\,v
+4\,{v}^{3}+7 \right) {u}^{7}\\
+ \left( 24\,v+20\,{v}^{2}+9+5\,{v}^{3}
 \right) {u}^{8}+ \left( 29\,v+25\,{v}^{2}+7\,{v}^{3}+11 \right) {u}^{
9}
+O(u^{10})
\end{multline*}
\begin{multline*}
\ch \mathcal{F}_4^{(2)} = 
 \left( v+2\,{v}^{2}+{v}^{3} \right) {u}^{3}+ \left( 3\,v+4\,{v}^{2}+2
\,{v}^{3}+1 \right) {u}^{4} \\
+ \left( 1+6\,v+9\,{v}^{2}+4\,{v}^{3}
 \right) {u}^{5}
+ \left( 3+11\,v+14\,{v}^{2}+{v}^{4}+7\,{v}^{3}
 \right) {u}^{6}+ \left( 23\,{v}^{2}+17\,v+{v}^{4}+11\,{v}^{3}+4
 \right) {u}^{7}\\
+ \left( 2\,{v}^{4}+25\,v+32\,{v}^{2}+7+16\,{v}^{3}
 \right) {u}^{8}+ \left( 3\,{v}^{4}+35\,v+46\,{v}^{2}+23\,{v}^{3}+9
 \right) {u}^{9}
+O(u^{10})
\end{multline*}
\begin{multline*}
\ch \mathcal{F}_5^{(2)} = 
\left( {v}^{2}+{v}^{3} \right) {u}^{5}+ \left( v+3\,{v}^{2}+3\,{v}^{3
}+{v}^{4} \right) {u}^{6}+ \left( 3\,v+8\,{v}^{2}+7\,{v}^{3}+2\,{v}^{4
} \right) {u}^{7}\\
+ \left( 1+7\,v+15\,{v}^{2}+13\,{v}^{3}+4\,{v}^{4}
 \right) {u}^{8}+ \left( 2+13\,v+27\,{v}^{2}+23\,{v}^{3}+7\,{v}^{4}
 \right) {u}^{9}
+O(u^{10})
\end{multline*}
\begin{multline*}
\ch \mathcal{F}_4^{(3)} = 
uv+ \left( 1+2\,{v}^{2}+3\,v \right) {u}^{2}+ \left( 2+{v}^{3}+5\,{v}^
{2}+6\,v \right) {u}^{3}+ \left( 10\,v+8\,{v}^{2}+2\,{v}^{3}+4
 \right) {u}^{4}\\
+ \left( 5+15\,v+14\,{v}^{2}+4\,{v}^{3} \right) {u}^{5
}+ \left( 8+22\,v+20\,{v}^{2}+{v}^{4}+7\,{v}^{3} \right) {u}^{6}+
 \left( 30\,{v}^{2}+30\,v+{v}^{4}+11\,{v}^{3}+10 \right) {u}^{7}
\\+
 \left( 2\,{v}^{4}+40\,v+40\,{v}^{2}+14+16\,{v}^{3} \right) {u}^{8}+
 \left( 3\,{v}^{4}+52\,v+55\,{v}^{2}+23\,{v}^{3}+17 \right) {u}^{9}
+O(u^{10})
\end{multline*}
\begin{multline*}
\ch \mathcal{F}_5^{(3)} = 
\left( v+2\,{v}^{2}+{v}^{3} \right) {u}^{3}+ \left( 4\,v+5\,{v}^{2}+2
\,{v}^{3}+1 \right) {u}^{4} \\
+ \left( 2+8\,v+11\,{v}^{2}+5\,{v}^{3}
 \right) {u}^{5}+ \left( 4+15\,v+19\,{v}^{2}+{v}^{4}+9\,{v}^{3}
 \right) {u}^{6}+ \left( 32\,{v}^{2}+24\,v+2\,{v}^{4}+16\,{v}^{3}+6
 \right) {u}^{7}\\
+ \left( 4\,{v}^{4}+37\,v+48\,{v}^{2}+10+25\,{v}^{3}
 \right) {u}^{8}+ \left( 7\,{v}^{4}+53\,v+71\,{v}^{2}+39\,{v}^{3}+14
 \right) {u}^{9}
+O(u^{10})
\end{multline*}
\begin{multline*}
\ch \mathcal{F}_6^{(3)} = 
\left( v+2\,{v}^{2}+{v}^{3} \right) {u}^{5}+ \left( 1+3\,v+5\,{v}^{2}
+4\,{v}^{3}+{v}^{4} \right) {u}^{6}+ \left( 1+7\,v+13\,{v}^{2}+9\,{v}^
{3}+2\,{v}^{4} \right) {u}^{7}\\
+ \left( 3+14\,v+24\,{v}^{2}+18\,{v}^{3}
+5\,{v}^{4} \right) {u}^{8}+ \left( 5+25\,v+44\,{v}^{2}+33\,{v}^{3}+9
\,{v}^{4} \right) {u}^{9}
+O(u^{10})
\end{multline*}
}

As mentioned in the text, these expansions match those of the characters $\ch {\mathcal I}_N^{(k,2)}$. {For instance, to the term $4v^2u^4$ in $\ch \mathcal{F}_4^{(2)}$, there corresponds the following four $(2,2,4)$-admissible superpartitions of degree $(4|2)$: $(3,1;0,0),\,(1,0;3,0),\,(2,1;1,0),\,(2,0;2,0)$.}

\end{appendix}
\begin{acknow}
We  thank the referees for their valuable comments and suggestions  that
have greatly improved the article. 
We are also grateful to Stephen Griffeth for helpful discussions.
 This work was  supported by NSERC (Natural Sciences and Engineering Research Council of Canada),
FONDECYT (Fondo Nacional de Desarrollo Cient\'{\i}fico y
Tecnol\'ogico de Chile) grants \#1090016 (L.L.) and \#1090034 (P.D.), 
and by CONICYT (Comisi\'on Nacional de Investigaci\'on Cient\'ifica y Tecnol\'ogica de Chile) via 
the Anillo de Investigaci\'on ACT56 (Lattices and Symmetry).
\end{acknow}

\end{document}